\documentclass[11pt]{article}

\usepackage[T1]{fontenc}
\usepackage{amsmath}
\usepackage{amssymb}
\usepackage{amsthm}
\usepackage{aliascnt}
\usepackage{mathtools}
\usepackage{braket}
\usepackage{booktabs}
\usepackage{tabularx}
\usepackage{pgfplots}
\usepackage{microtype}
\usepackage{placeins}
\usepackage{etoolbox}
\usepackage[margin=1in]{geometry}
\usepackage[bottom]{footmisc}
\usepackage{xcolor}
\usepackage[pagebackref=true]{hyperref}
\usepackage[capitalise]{cleveref}
\usepackage{newpxtext}
\usepackage[varg,bigdelims]{newpxmath}
\usepackage{titlesec}

\titlespacing*{\paragraph}{0pt}{1ex plus 0.5ex minus 0.2ex}{1em}

\renewcommand{\leq}{\leqslant}

\renewcommand{\geq}{\geqslant}

\allowdisplaybreaks
\renewcommand{\arraystretch}{1.25}
\pgfplotsset{compat=1.17,compat/show suggested version=false}

\hypersetup{
    colorlinks=true,
    linkcolor=purple!70!black,
    citecolor=purple!70!black,
    urlcolor=purple!70!black,
    pdftitle={Binary code rate bounds via cq channels},
    pdfauthor={Omar Alrabiah and Venkatesan Guruswami}
}

\newtheorem{theorem}{Theorem}
\newaliascnt{prop}{theorem}
\newtheorem{prop}[prop]{Proposition}
\aliascntresetthe{prop}
\newaliascnt{fact}{theorem}
\newtheorem{fact}[fact]{Fact}
\aliascntresetthe{fact}
\newaliascnt{lemma}{theorem}
\newtheorem{lemma}[lemma]{Lemma}
\aliascntresetthe{lemma}
\newaliascnt{corollary}{theorem}
\newtheorem{corollary}[corollary]{Corollary}
\aliascntresetthe{corollary}
\newtheorem{definition}{Definition}[section]
\theoremstyle{remark}
\newtheorem{remark}{Remark}[section]

\AtEndEnvironment{remark}{\hfill$\triangleleft$}

\crefname{theorem}{Theorem}{Theorems}
\crefname{prop}{Proposition}{Propositions}
\crefname{fact}{Fact}{Facts}
\crefname{lemma}{Lemma}{Lemmas}
\crefname{corollary}{Corollary}{Corollaries}
\crefname{definition}{Definition}{Definitions}
\crefname{remark}{Remark}{Remarks}
\crefname{appendix}{Appendix}{Appendices}


\newcommand*\Bell{\ensuremath{\boldsymbol\ell}}
\newcommand{\ketbra}[2]{\ket{#1}\!\bra{#2}}
\newcommand{\cC}{\mathcal{C}}
\newcommand{\cD}{\mathcal{D}}
\newcommand{\cM}{\mathcal{M}}
\newcommand{\cX}{\mathcal{X}}
\newcommand{\cY}{\mathcal{Y}}
\newcommand{\cE}{\mathcal{E}}
\newcommand{\cF}{\mathcal{F}}
\newcommand{\cG}{\mathcal{G}}
\newcommand{\F}{\mathbb{F}}
\newcommand{\T}{\mathcal{T}}
\newcommand{\C}{\mathbb{C}}
\newcommand{\Ber}{\operatorname{Ber}}
\newcommand{\Tr}{\operatorname{Tr}}
\newcommand{\supp}{\operatorname{supp}}

\newcommand{\wt}{\operatorname{wt}}
\newcommand{\Dens}{\mathsf{D}}
\newcommand{\one}{\mathbf{1}}
\newcommand{\Eb}{\mathbf{E}}
\newcommand{\Prb}{\mathbf{Pr}}
\newcommand{\bits}{\{0,1\}}
\newcommand{\MRRW}{R_{\mathrm{MRRW}}}
\newcommand{\bchat}{\hat{\mathbf{c}}}
\newcommand{\bxhat}{\hat{\mathbf{x}}}
\newcommand{\byhat}{\widehat{\mathbf{y}}}
\newcommand{\bbhat}{\widehat{\mathbf{b}}}
\newcommand{\bc}{\mathbf{c}}
\newcommand{\bb}{\mathbf{b}}
\newcommand{\be}{\mathbf{e}}
\newcommand{\bx}{\mathbf{x}}
\newcommand{\by}{\mathbf{y}}
\newcommand{\bz}{\mathbf{z}}
\newcommand{\eps}{\varepsilon}
\newcommand{\abs}[1]{\left|#1\right|}
\newcommand{\norm}[1]{\lVert#1\rVert}
\newcommand{\N}{\mathcal{N}}
\newcommand{\inner}[2]{\left\langle #1, #2\right\rangle}
\newcommand{\ceil}[1]{\left\lceil #1 \right\rceil}

\newcommand{\BEC}{\mathrm{BEC}}
\newcommand{\BSC}{\mathrm{BSC}}
\newcommand{\PSC}{\mathrm{PSC}}
\newcommand{\MQC}{\mathrm{MQC}}
\newcommand{\mPSC}{\mathrm{mPSC}}
\newcommand{\mMQC}{2\mathrm{MQC}}

\title{Binary code rate bounds via classical--quantum channels}
\author{Omar Alrabiah\thanks{Department of Electrical Engineering and Computer Sciences, UC Berkeley, Berkeley, CA 94709, USA. Email: \url{oalrabiah@berkeley.edu}. Research supported in part by a Saudi Arabian Cultural Mission (SACM) Scholarship, a UC Noyce initiative award, ONR grant N00014-24-1-2491, and V. Guruswami's Simons Investigator Award.}
\and Venkatesan Guruswami\thanks{Department of Electrical Engineering and Computer Sciences, Department of Mathematics, and the Simons Institute for the Theory of Computing, UC Berkeley, Berkeley, CA 94709, USA. Email: \url{venkatg@berkeley.edu}. Research supported in part by a Simons Investigator Award, ONR grant N00014-24-1-2491, and NSF grant 2211972.}
}
\date{}
\begin{document}
\pagestyle{empty}
\maketitle
\thispagestyle{empty}

\begin{abstract}
We derive the four principal asymptotic rate-distance tradeoffs for binary codes---Plotkin, Elias--Bassalygo, and the two McEliece--Rodemich--Rumsey--Welch (MRRW) bounds---from one theorem, the ``pretty good criterion.'' If the bit error rate under the pretty good measurement (PGM)---the quantum analog of posterior sampling---of a binary-input output-symmetric classical--quantum (cq) channel lies below $\delta$, then every length-$n$ binary code, linear or nonlinear, of relative distance $\delta$ has rate at most the channel's capacity, up to an $O(n^{-1/2})$ correction. Rate--distance bounds thereby reduce to a channel design problem, wherein the task is to minimize channel capacity subject to the posterior bit error rate constraint. Via the pretty good criterion, the binary erasure channel (BEC) yields Plotkin, the binary symmetric channel (BSC) yields Elias--Bassalygo, the pure-state channel (PSC) yields the first MRRW bound, and a masked PSC yields the second MRRW bound exactly. 

\smallskip 
This framework is then instantiated with new channels to improve upon the MRRW bounds. Specifically, the mixed-qubit channel (MQC), a mixed-state version of PSC, strictly improves the first MRRW bound at every $0 < \delta < \frac{1}{2}$, while the masked mixed-qubit channel (2MQC) strictly improves the second MRRW bound throughout the same interval.
\end{abstract}

\begingroup
\small
\linespread{1.2}\selectfont
\makeatletter
\noindent\begin{minipage}{\textwidth}
\tableofcontents
\end{minipage}
\endgroup

\clearpage
\pagestyle{plain}

\section{Introduction}
\label{sec:introduction}

The theory of error-correcting codes developed around two error models. In Hamming's model~\cite{Ham50}, errors are adversarial, and a code of minimum distance $d$ guarantees recovery from every pattern of fewer than $d/2$ substitutions, wherever they occur. In Shannon's model~\cite{Sha48}, errors are stochastic, produced by a fixed channel that corrupts the transmitted symbols at random, and reliability asks that the block error probability vanish as the blocklength grows. The two models measure a code against different standards---worst-case correction radius versus typical-case reliability---and they are not interchangeable, for rare error patterns of a given weight may be undecodable even when random errors of that weight are usually harmless. Each model carries its own extremal theory---the tradeoff between rate and minimum distance on the Hamming side, capacity on the Shannon side.

On the Shannon side, the extremal theory was settled at its birth. The best information transmission rate across a memoryless channel is dictated by one corresponding number, namely its capacity. At every rate below the capacity, there are codes whose block decoding error vanishes~\cite{Sha48}, and at every rate above it, the block decoding error of every code tends to one~\cite{Wol57,Str64,Ari73}.

On the Hamming side, the tradeoff is far from being completely understood. Let $A_2(n,d)$ be the largest cardinality of a (possibly nonlinear) binary length-$n$ code with minimum distance at least $d$, and define the asymptotic rate function
\begin{equation}
\label{eq:rate-function}
R_2(\delta) \coloneqq \limsup_{n \to \infty}\frac{1}{n}\log{A_2(n,\ceil{\delta n})}.
\end{equation}
From below, the classic Gilbert--Varshamov greedy argument~\cite{Gil52,Var57} gives $R_2(\delta) \geq 1-h(\delta)$, where $h$ is the binary entropy function. This remains the best asymptotic lower bound on the rate of binary codes.\footnote{There have, however, been lower-order improvements. In particular, Jiang and Vardy~\cite{JV04} sharpened the Gilbert--Varshamov bound by a factor of order $n$, while Gaborit and Z\'emor~\cite{GZ08} obtained an analogous refinement for certain infinite families of rate-$\frac{1}{2}$ binary linear codes. By contrast, for square alphabet sizes $q \geq 49$, the Tsfasman--Vl\u{a}du\c{t}--Zink algebraic-geometric construction~\cite{TVZ82} improves upon the $q$-ary Gilbert--Varshamov exponent.} From above, McEliece, Rodemich, Rumsey, and Welch (MRRW)~\cite{MRRW77} proved two upper bounds. The first is $R_2(\delta) \leq \MRRW(\delta) \coloneqq h\left(\frac{1}{2}-\sqrt{\delta(1-\delta)}\right)$, and the second, which is never weaker, was the benchmark for nearly five decades. Since then, no stronger general bound for binary codes has been shown, even for linear codes.

The rates in the Shannon model far outpace their Hamming-model counterparts. For example, for the binary symmetric channel (BSC), which flips each transmitted bit independently with probability $p$, its capacity is $1-h(p)$. Correcting the same fraction $p$ of substitutions in Hamming's model, wherever they occur, requires a relative minimum distance $2p$, and the achievable rate is then at most $R_2(2p)$. Since $R_2(\delta) = 0$ at $\delta \geq \frac{1}{2}$, by the Plotkin bound~\cite{Plo60}, no positive rate corrects an adversarial error fraction of $\frac{1}{4}$, whereas the capacity remains positive for every $p < \frac{1}{2}$.

List decoding bridges the two models. Introduced by Elias~\cite{Eli57} and Wozencraft~\cite{Woz58}, it keeps the adversarial error model but weakens what the decoder must return, allowing a short list of candidates in place of a single codeword. For every $R < 1-h(p)$ there exist binary codes of rate $R$ that are list decodable up to radius $p$ with bounded list size~\cite{ZP81,Eli91}, whereas no rate above $1-h(p)$ is possible with subexponential list size~\cite{Gur04}. The list decoding model thereby provides an avenue to approach Shannon capacity even when the errors are adversarial.

\bigskip

\subsection{Main results}
\label{subsec:main-results}

This work crosses between the two error models on a different bridge, drawing on a feature built into the Hamming model from the outset---\emph{obliviousness} to the channel. A minimum-distance guarantee is not tied to a specific noise model. In particular, it is not necessarily tailored to the binary symmetric channel (BSC), the binary erasure channel (BEC), or any other binary memoryless symmetric channel, for that matter. Indeed, previous works have shown quantitative aspects of this, such as~\cite{TZ00,HSS21,PSW24,KYH26}.

Exploiting this connection in reverse, we can leverage the model-agnostic error-resilience offered by minimum distance to turn bounds for Shannon-model channels into Hamming-model bounds. Whenever codes of relative distance $\delta$ remain reliably decodable over a fixed channel, the channel-coding converse caps their rate by its capacity.

Our result below is a formal manifestation of this intuition and serves as the principal workhorse behind our approach. The strength of our approach draws from an expansion of the scope of this connection to include channels with quantum outputs.

Fix a finite-dimensional binary-input output-symmetric classical--quantum (cq) channel with output density matrices $\sigma_0,\sigma_1 \in \Dens(\C^d)$. Two quantities of interest govern our approach: (i) the \emph{uniform-prior Holevo information} $\chi(\sigma_0,\sigma_1)$, which equals the channel capacity by output-symmetry (\cref{prop:output-symmetric-capacity}), and (ii) the \emph{posterior bit error rate} $p_{\mathrm{e}}(\sigma_0,\sigma_1)$, the error probability of posterior sampling a single bit. For non-commuting outputs, this is defined as the bit error rate of the pretty good measurement (PGM), the quantum extension of posterior sampling. Both objects are defined formally in \cref{subsec:binary-input-cq-channels,subsec:pgm-definition}. With this backdrop, we now formally introduce our main theorem. 

\begin{theorem}[The pretty good criterion]
\label{thm:main-result}
Fix $\delta \in (0,\frac{1}{2})$ and a binary code $\cC \subseteq \bits^n$ with distance $d_{\min}(\cC) \geq \delta n$ and rate $R(\cC) \coloneqq n^{-1}\log{\abs{\cC}}$. Fix a finite-dimensional binary-input output-symmetric\footnote{\label{fn:output-symmetry-lift}The output-symmetry hypothesis is a mild constraint. Indeed, consider $\widehat{\sigma}_0 \coloneqq \frac{1}{2}\sigma_0 \oplus \frac{1}{2}\sigma_1$ and $\widehat{\sigma}_1 \coloneqq \frac{1}{2}\sigma_1 \oplus \frac{1}{2}\sigma_0$. For any $\sigma_0$ and $\sigma_1$, the pair $(\widehat{\sigma}_0,\widehat{\sigma}_1)$ is output-symmetric. Moreover, $p_{\mathrm{e}}(\widehat{\sigma}_0,\widehat{\sigma}_1) = p_{\mathrm{e}}(\sigma_0,\sigma_1)$ and $\chi(\widehat{\sigma}_0,\widehat{\sigma}_1) = \chi(\sigma_0,\sigma_1)$. This lift can double the output dimension, but it leaves the leading asymptotic rate bound unchanged.} classical--quantum (cq) channel with output states $\sigma_0,\sigma_1 \in \Dens(\C^d)$. If $p_{\mathrm{e}}(\sigma_0,\sigma_1) < \delta$, then
\begin{equation} 
R(\cC) \leq \chi(\sigma_0,\sigma_1) + O(n^{-1/2}),
\end{equation}
where the implicit constant depends only on $\delta$ and the cq channel $\Phi_\sigma$. In particular, maximizing over all binary codes with minimum distance at least $\delta n$ and letting $n \to \infty$ gives
\begin{equation}
R_2(\delta) \leq \chi(\sigma_0,\sigma_1).
\end{equation}
\end{theorem}

The criterion is already informative for classical channels, where the outputs $\sigma_0$ and $\sigma_1$ commute. Applied to the BEC, it recovers the Plotkin bound~\cite{Plo60}, and applied to the BSC, it recovers the Elias--Bassalygo bound~\cite{Eli55,Bas65} (\cref{subsec:bec,subsec:bsc}). A correspondence between these two bounds and their respective channels has been previously noted in a different context~\cite{Bla77,Dal15}.

For the state-of-the-art MRRW bounds, however, it was not previously known whether they can be captured as successful decoding over an appropriate channel. As it now turns out, there are barriers to such a pursuit: the Elias--Bassalygo bound is the strongest that \cref{thm:main-result} can yield when restricted to binary-input \emph{classical} channels.\footnote{\emph{Proof sketch.} For $q_y \coloneqq \Prb[\bb = 1 \mid \by = y]$, one has $p_{\mathrm{e}} = \Eb[2q_{\by}(1-q_{\by})]$ and $\chi = 1-\Eb[h(q_{\by})]$. Since $h(q) = f(2q(1-q))$ for the increasing concave function $f(t) \coloneqq h((1-\sqrt{1-2t})/2)$, Jensen's inequality gives $\chi \geq 1-f(p_{\mathrm{e}})$, with the equality case being BSC. Letting $p_{\mathrm{e}} \to \delta^-$ yields the Elias--Bassalygo bound.} Thus, if such a channel were to exist, it ought to be a genuinely quantum classical--quantum (cq) channel with non-commuting output states.

A driving question in this work, therefore, is determining whether the MRRW bound can be realized as successful decoding over an appropriate cq channel with genuinely quantum outputs. This leads us to the pure-state channel (PSC), whose connection to the first MRRW bound is laid out in~\cref{subsec:unique-list-pgm,subsec:syndrome-duality}. Concretely, the PSC is the BSC with a coherent output. Writing $\ket{x}$ for the standard basis vector indexed by $x \in \bits^n$, define the Bernoulli superposition state $\ket{\Ber(p)^n} \coloneqq \ket{\Ber(p)}^{\otimes n}$ whose amplitudes are the square roots of the BSC noise probabilities, where $\ket{\Ber(p)} \coloneqq \sqrt{1-p}\ket{0} + \sqrt{p}\ket{1}$. The blocklength-$n$ PSC then maps each codeword to the corresponding coherent noisy codeword---the superposition over all possible noisy codewords. More formally, the PSC is defined by the map
\begin{equation*}
c \longmapsto \ket{\sigma_c} \coloneqq X^c\ket{\Ber(p)^n}
= \sum_{e \in \bits^n}{\sqrt{p^{\wt(e)}(1-p)^{n-\wt(e)}}\ket{c \oplus e}},
\end{equation*}
where $X^c \coloneqq X^{c_1} \otimes \cdots \otimes X^{c_n}$ and $X$ is the X-Pauli matrix defined as $X\ket{0} = \ket{1}$ and $X\ket{1} = \ket{0}$. Note that each coordinate of the output state $\sigma_c$ carries one of two nonorthogonal states $X^{c_i}\ket{\Ber(p)}$. Moreover, measuring $\sigma_c$ in the computational basis yields the outcome $c \oplus e$ where $e \sim \Ber(p)^n$. Thus, decoding over $\PSC(p)$ is no harder than decoding over $\BSC(p)$. The PSC additionally retains coherence among the possible error patterns, which the computational-basis measurement discards.

The pretty good measurement (PGM) provides a way to exploit this coherence, and for binary linear codes it does so \emph{optimally}, thanks to previous work on discriminating geometrically uniform quantum states~\cite{BKMH97,UTHH99,EF01,EMV04,ZCCL25,RP21,PR25}. In classical terms, for the PSC, posterior sampling and maximum-a-posteriori (MAP) decoding coincide for linear codes. In fact, this optimality persists down at the bit level (\cref{fact:psc-bitwise-optimal}). In particular, no explicit decoder for one codeword bit can beat the corresponding bit of the PGM.

Now, fix a binary linear code $\cC \leq \F_2^n$ with $d_{\min}(\cC) \geq \delta n$, let $\bc$ be a uniformly random codeword from $\cC$, and let $\bchat$ be the outcome of the PGM applied to the output state $\ket{\sigma_{\bc}}$. Since measuring the $i$'th qubit in the computational basis recovers the $i$'th codeword bit $\bc_i$ with error $p$, bitwise optimality therefore tells us that $\Prb[\bchat_i \neq \bc_i] \leq p$ for every $i \in \{1, \ldots , n\}$. Thus,
\begin{equation*}
\delta n \cdot \Prb[\bchat \neq \bc]
\leq \Eb[d(\bchat,\bc)]
= \sum_{i=1}^{n}{\Prb[\bchat_i \neq \bc_i]}
\leq pn,
\end{equation*}
where the first inequality holds because $\bc$ and $\bchat$ are distinct codewords on the event $\{\bchat \neq \bc\}$. Hence $\Prb[\bchat = \bc] \geq 1-p/\delta$, meaning that for every $p < \delta$, a single measurement of the PSC output returns the whole transmitted codeword with constant probability (\cref{cor:block-error}). This is exactly the operational conclusion that a constant-size list would have supplied, and it is the sense in which the first MRRW bound is the list-decoding threshold of the PSC, as speculated in \cref{subsec:unique-list-pgm}. Indeed, when coupled with the well-established strong converse for cq channel coding~\cite{Win99a,ON99,MO17}, we immediately conclude that the code's rate cannot exceed the channel capacity. More concretely, we arrive at the rate bound $R(\cC) \leq h(\frac{1}{2}-\sqrt{p(1-p)})+O(n^{-1/2})$ for every $p < \delta$. Thus, letting $p \to \delta$ gives the first MRRW bound via our framework.

In the preceding argument, it turns out that linearity was really only a convenience rather than a necessity. Indeed, it only entered through bitwise optimality, which is what allowed the block PGM to be compared against a single-qubit measurement. For an arbitrary code, the same bound $\Prb[\bchat_i \neq \bc_i] \leq p$ follows instead from three general properties of the PGM: (1) closure under coarse-graining (\cref{lem:block-pgm-bitwise}), which is a key feature of posterior sampling in the classical case, (2) monotonicity of its error under discarding output systems (\cref{lem:pgm-data-processing}), and (3) the worst-case optimality of the uniform input prior (\cref{lem:uniform-prior-worst}). Nor was anything tied to pure-state outputs, as those three properties and the strong converse hold for every binary-input output-symmetric cq channel. The derivation therefore extends seamlessly to general cq channels, which is exactly the pretty good criterion (\cref{thm:main-result}).

Now, the pretty good criterion (\cref{thm:main-result}) provides us with the freedom to develop other rate upper bounds, specifically in the choice of the pair of output-symmetric density matrices $\sigma_0,\sigma_1$. As we will later see, one can choose $\sigma_0,\sigma_1$ to recover the second MRRW \emph{exactly}, albeit with four-dimensional states instead of qubit states. But for now, still focusing on channels with qubit outputs, there is the following natural question: is the first MRRW bound the strongest bound attainable by a qubit-output channel? As it so happens, it is not! By considering appropriate mixed qubits $\sigma_0$ and $\sigma_1$ and their corresponding mixed-qubit channel (MQC) (\cref{subsec:mqc}), we obtain a bound $R_{\MQC}(\delta)$ that strictly improves upon the first MRRW bound throughout the entire interval $\delta \in (0,\frac{1}{2})$.

\begin{theorem}[MQC improvement over the first MRRW bound]
\label{thm:first-mrrw-improvement}
Define
\begin{equation*}
R_{\MQC}(\delta)
\coloneqq
\min_{0 \leq r \leq \delta}
\left[s\left((1-2r)^2\right)-s\left(2(\delta-r)(1-2r)\right)\right],
\end{equation*}
where $s(D) \coloneqq h((1-\sqrt{1-D})/2)$. For every $\delta \in (0,\frac{1}{2})$, we have
\begin{equation*}
R_2(\delta) \leq R_{\MQC}(\delta) < \MRRW(\delta).
\end{equation*}
\end{theorem}

Unlike the first MRRW bound, the MQC bound is additionally never weaker than the Elias--Bassalygo bound. Indeed, when $r = 0$, it recovers the bound obtained by $\BSC(J(\delta))$ where $J(\delta) \coloneqq (1-\sqrt{1-2\delta})/2$ is the binary Johnson radius~\cite{Joh62}. Moreover, when $r = \delta$, it recovers the bound obtained by $\PSC(\delta)$. In fact, by~\cref{rem:mqc-has-all-qubits}, the MQC rate bound is never worse than \emph{any} rate bound obtained by any other binary-input output-symmetric cq channel with qubit outputs.

The MQC does not by itself outperform the second MRRW bound in the low-distance regime (where it is better than the first MRRW bound). The second MRRW bound instead enters our framework through a channel family with higher-dimensional outputs. This is achieved by \emph{channel masking} (\cref{subsec:masking}), an operation that enables one to consider non-uniform prior output states. Briefly, it is defined as follows: given an input bit $b$, the channel draws a hidden bit $m \sim \Ber(\alpha)$ for a fixed parameter $\alpha \in [0,1]$ and reveals the ``masked'' bit $b \oplus m$, along with an accompanying quantum state $\tau_m$ as a ``quantum hint'' for the hidden bit $m$. Decoding $b$ is therefore equivalent to decoding the ``mask'' $m$ from a generally non-uniform state ensemble. Applying this operation to the PSC gives the masked PSC and recovers the second MRRW bound \emph{exactly} (\cref{cor:mpsc-rate}).

Applying the same masking operation to the MQC gives the masked mixed-qubit channel (2MQC). This family contains the masked PSC and, when the masking bias $\alpha = 1/2$, recovers the MQC (along with an input-independent, and thus inconsequential, classical bit). These inclusions make the 2MQC a natural refinement, but the strict improvement is not automatic; it follows from the optimization and analysis underlying the next theorem.

\begin{theorem}[2MQC improvement over the second MRRW bound]
\label{thm:second-mrrw-improvement}
For every $\delta \in (0,\frac{1}{2})$, we have\footnote{See~\cref{tab:2mqc-mrrw-numerics} for numerical estimates of the improvement.}
\begin{equation*}
R_2(\delta) \leq R_{\mMQC}(\delta) < \MRRW^{(2)}(\delta).
\end{equation*}
Here, $R_{\mMQC}(\delta)$ is the variational bound defined in \eqref{eq:2mqc-rate}, and $\MRRW^{(2)}(\delta)$ is the second MRRW bound in \eqref{eq:second-mrrw}.
\end{theorem}

It is worth noting that the four channels corresponding to the previously known rate bounds---Plotkin, Elias--Bassalygo, and the two MRRW bounds---have a common structural feature. Namely, they are all orthogonal mixtures of pure states. Thus, their mixedness can be perfectly distinguished by examining their block label. On the other hand, generic MQC outputs have no common one-dimensional orthogonal decomposition, so their ``mixedness'' cannot be entirely distinguished by a readable label. In this sense, the MQC introduces genuine quantum mixtures beyond the orthogonal pure-state mixtures underlying the four previous rate bounds.

\paragraph{Rate bounds for LDPC codes.}
So far, the discussion has only focused on arbitrary binary codes with good distance. A natural question is whether this framework is amenable to accounting for further structural properties. Indeed, as was shown in~\cite{IS15,SY26}, one can exploit the dual covering radius methods developed for the first MRRW bound to show sharper rate-distance tradeoffs for LDPC codes. Given the connection between the PSC and dual covering methods (\cref{prop:syndrome-duality}), it becomes natural to wonder if one can obtain similar state-of-the-art rate bounds for LDPC codes~\cite{SY26}.

Fix a positive integer $w \geq 3$ denoting the locality of the LDPC code. For integers $n,d \geq 1$, let $A_{2,w}(n,d)$ be the largest cardinality of a binary linear code $\cC \leq \F_2^n$ with $d_{\min}(\cC) \geq d$ whose dual is generated by parity checks of weight at most $w$, and put
\begin{equation*}
R_w(\delta) \coloneqq \limsup_{n \to \infty}\frac{1}{n} \log{A_{2,w}(n,\ceil{\delta n})}.
\end{equation*}
For simplicity, we focus here on $w = 3$, although the setup extends to general $w$. We denote the one-check PSC curve by $U_3^{\PSC}$ and the Shangguan--Yang curve by $B_3^{\mathrm{SY}}$; both are defined formally in \eqref{eq:psc-ldpc-rate} and \eqref{eq:sy-ldpc-base}, respectively. Adapting the pretty good criterion (\cref{thm:main-result}) and the bitwise optimality of PGM for binary linear codes (\cref{fact:psc-bitwise-optimal}) using one sparse parity check yields the following theorem, proven in~\cref{sec:psc-ldpc}.

\begin{theorem}[PSC-locality improvement for LDPC codes]
\label{thm:psc-ldpc-improvement}
For every $\delta \in (0,\frac{1}{2})$,
\begin{equation}
R_3(\delta)
\leq U_3^{\PSC}(\delta)
< \MRRW(\delta).
\end{equation}
Moreover, for every $\delta \in [\frac{1}{6},\frac{1}{2})$, we have
\begin{equation}
U_3^{\PSC}(\delta)
< B_3^{\mathrm{SY}}(\delta).
\end{equation}
\end{theorem}

We close this subsection with two open problems.

\paragraph{Optimizing the cq channel.}
\label{open:channel-optimization-open}
By the pretty good criterion (\cref{thm:main-result}), the search for sharper upper bounds on $R_2(\delta)$ for $\delta \in (0,\frac{1}{2})$ now reduces to minimizing the Holevo information subject to keeping the posterior bit error rate below $\delta$. For a fixed output dimension $d \geq 2$, the closed formulas in~\cref{subsec:pgm-bit-error-formulas,subsec:holevo-formulas} make this the problem of minimizing, over output-symmetric density matrices $\sigma_0,\sigma_1 \in \Dens(\C^d)$,
\begin{equation}
\label{eq:objective-function-open}
\chi(\sigma_0,\sigma_1) = S\left(\frac{\sigma_0+\sigma_1}{2}\right) - \frac{1}{2}\left[S(\sigma_0)+S(\sigma_1)\right]
\end{equation}
subject to the constraint that
\begin{equation}
\label{eq:constraint-open}
p_{\mathrm{e}}(\sigma_0,\sigma_1) = \Tr\left(\sigma_0(\sigma_0+\sigma_1)^{-1/2}\sigma_1(\sigma_0+\sigma_1)^{-1/2}\right) \leq \delta,
\end{equation}
where the inverse is taken over $\supp(\sigma_0+\sigma_1)$. Note that by~\cref{fn:output-symmetry-lift}, assuming that $\sigma_0$ and $\sigma_1$ are output-symmetric is only a mild restriction to this optimization problem. The MQC and 2MQC show that noncommuting mixed outputs can improve upon the standard choices, but they do not determine the optimal channel in~\cref{thm:main-result}. Pure outputs can attain at best the first MRRW bound, while the exact optimum across all finite dimensions remains open. As a simple sanity check, Fano's inequality~\cite{Fan61} and the Holevo bound~\cite{Hol73a} show that $\chi(\sigma_0,\sigma_1) \geq 1-h(\delta)$ under the constraint $p_{\mathrm{e}}(\sigma_0,\sigma_1) \leq \delta$, so the method cannot cross the Gilbert--Varshamov bound. Nonetheless, the limitations of this optimization, specifically in relation to previous barriers~\cite{Sam01,NS05,Sam25}, remain to be understood.

\paragraph{BSC decoding up to the minimum distance.}
\label{open:bsc-distance-decoding}
The most conspicuous unresolved question that we wish to flag is in fact entirely classical. For every fixed $0 < \eps < \delta < \frac{1}{2}$, does maximum-likelihood decoding on $\BSC(\delta-\eps)$ succeed with high probability, uniformly over all binary linear codes $\cC \leq \F_2^n$ with $d_{\min}(\cC) \geq \delta n$? The best general guarantee is decoding up to the Johnson radius $J(\delta) = (1-\sqrt{1-2\delta})/2$ via the Johnson bound~\cite{Joh62} and the sharp-threshold theorem of Tillich and Z\'emor~\cite{TZ00}\footnote{This sharp threshold result was more recently generalized to $q$-ary symmetric channels by~\cite{PSW24}.}. Kougang-Yombi and H\k{a}z{\l}a~\cite{KYH26} make progress on this problem by proving that the Johnson barrier can be surpassed for linear codes over larger alphabets. Note that an affirmative answer to this problem, when combined with the BSC strong converse and the Gilbert--Varshamov bound, would prove that the linear analog of $R_2(\delta)$ precisely equals $1-h(\delta)$.

\subsection{Why posterior sampling?}
\label{subsec:why-posterior}

In the classical setting, the go-to decoder is maximum-a-posteriori (MAP) decoding, which returns a codeword of largest posterior probability.\footnote{MAP decoding is the Bayesian form of maximum-likelihood decoding (MLD): it maximizes the posterior probability $\Prb[\bc = c \mid \by = y]$ rather than the likelihood $\Prb[\by = y \mid \bc = c]$, and so it incorporates the prior on $\bc$. For the uniform codeword prior used throughout this work, the two rules coincide.} In place of MAP, the decoder used in this work is posterior sampling, whose bit error rate is what dictates the pretty good criterion (\cref{thm:main-result}), \emph{not} MAP. Now, posterior sampling is not optimal by definition, since it merely samples the posterior that MAP maximizes. Nonetheless, it is never far behind: it errs at most twice as often as MAP.\footnote{\emph{Proof sketch.} Conditional on an output $y$, set $q_x \coloneqq \Prb[\bc = x \mid \by = y]$ for $x \in \cC$. Posterior sampling succeeds with probability $\sum_{x \in \cC}{q_x^2}$, whereas MLD succeeds with probability $\max_{x \in \cC}{q_x}$. Averaging the pointwise upper bound $\sum_{x \in \cC}{q_x^2} \leq \max_{x \in \cC}{q_x}$ and combining $\sum_{x \in \cC}{q_x^2} \geq (\max_{x \in \cC}{q_x})^2$ with Jensen's inequality gives $e_{\mathrm{MLD}} \leq p_{\mathrm{e}} \leq 2e_{\mathrm{MLD}}-e_{\mathrm{MLD}}^2 \leq 2e_{\mathrm{MLD}}$.} In fact, Barnum and Knill~\cite{BK02} prove an analogous result for the pretty good measurement (PGM)---the quantum extension of posterior sampling---and the optimal quantum measurement. Since the channel coding strong converse~\cite{Wol57,Str64,Ari73,Win99a,ON99} only requires a constant chance of successful block decoding to kick in, this loss is immaterial. 

The main motive to prefer posterior sampling over MAP is its closure under coarse-graining. In particular, blockwise posterior sampling \emph{simultaneously} samples from the \emph{bitwise} posteriors as well. Indeed, let $\bc$ be uniform on a binary code $\cC \subseteq \bits^n$. Let $\by$ be its output when passed through a classical binary memoryless symmetric channel. Given $\by = y$, draw $\bchat$ from the posterior law of $\bc$. That is, sample $\bchat$ from $\cC$ with chance proportional to $\Prb[\bc = \bchat \mid \by = y]$. The output is a codeword $\bchat \in \cC$ by construction, and, conditioned on $\by = y$, the transmitted word $\bc$ and the decoded word $\bchat$ are independent draws from the same posterior, so sampling commutes with taking coordinates. Indeed, for every coordinate $i \in \{1, \ldots , n\}$, one has
\begin{equation*}
\Prb[\bchat_i = b \mid \by = y]
= \sum_{c \in \cC : c_i = b}{\Prb[\bc = c \mid \by = y]}
= \Prb[\bc_i = b \mid \by = y].
\end{equation*}
Thus, the $i$'th coordinate of the sampled codeword is itself a posterior sample of the transmitted bit. On the other hand, block MAP decoding outputs a codeword, but its coordinates escape single-letter control, since the bits of the most likely codeword need not be the most likely.\footnote{For example, suppose a uniform codeword from $\cC=\{000,110,101\}$ is transmitted over $\BSC(3/7)$ and that $y=000$ was received. The posterior probabilities of the three codewords are $8/17,9/34,9/34$, so blockwise MAP returns $000$. The posterior probabilities that the three coordinates equal $1$ are instead $9/17,9/34,9/34$, so bitwise MAP returns $100\notin\cC$.}

While this contrast between MAP and posterior sampling might appear trifling, these fine distinctions are dramatically amplified in the quantum setting. Indeed, let us revisit the pure-state channel (PSC), defined by the map
\begin{equation*}
c \longmapsto \ket{\sigma_c} = \sum_{e \in \bits^n}{\sqrt{p^{\wt(e)}(1-p)^{n-\wt(e)}}\ket{c \oplus e}}.
\end{equation*}
As mentioned previously, one must exploit the coherence in $\sigma_\bc$ to benefit from the PSC. This coherence is a scarce resource: classically, one may read an output without disturbing it, whereas a quantum output generally affords only one measurement. That measurement should therefore extract as much information about the transmitted codeword $\bc$ as possible. The coarse-graining property of the PGM, inherited from classical posterior sampling (\cref{lem:coarse-graining}), is indispensable for this purpose. MAP decoding has no analogous property in the quantum setting.

As discussed in \cref{subsec:unique-list-pgm}, successful MAP decoding in the classical setting often relies on combinatorial arguments, most of which falter in the quantum setting. The quantum formulation instead makes the pretty good criterion (\cref{thm:main-result})---an unadorned criterion that nonetheless underpins all known rate upper bounds---transparent, while also furnishing a streamlined treatment of the classical case.

\subsection{How does mixing help?}
\label{subsec:how-mixing}

At first glance, the pretty good criterion (\cref{thm:main-result}) might seem to favor either fully classical states---leading to the Elias--Bassalygo bound---or pure quantum states---leading to the first MRRW bound. Its recovery of the second MRRW bound suggests that mixed states can do better, as the MQC and 2MQC constructions demonstrate.

To explain conceptually why moving away from pure states can lower the bound in~\cref{fig:mixed-optimization}, we first describe how the MQC is realized.

The MQC output states can be realized as passing the PSC output states through the bit-flip channel. Formally, let $\mathcal{B}_\eta$ be the X-Pauli extension of $\BSC(\eta)$, given by 
$\mathcal{B}_\eta(\varrho) \coloneqq (1-\eta)\varrho + \eta X\varrho X$. For $b \in \bits$ and $r \in (0,\frac{1}{2})$, let $\sigma_b^{\PSC}(r)$ be the output states of $\PSC(r)$. Then, for a BSC noise level $\eta \in (0,\frac{1}{2})$, the output states of MQC are defined as\footnote{\label{fn:bsc-psc-composition} Using the fact that the bit-flip channel commutes with the PSC (cf.~\cref{subsubsec:mqc-pgm-error,subsubsec:2mqc-pgm-error}), one can recast the MQC at the blocklength level as the map $c \mapsto c \oplus \be \mapsto X^{\bc \oplus \be}\ket{\Ber(r)^n}$, where $\be \sim \Ber(\eta)^n$. Put compactly, we have the identity $\MQC(r,\eta) = \PSC(r) \circ \BSC(\eta)$. The fact that $\be$ is not retained is precisely what produces ``mixedness'' in the MQC. Note that this perspective is supplemental to the current discussion but is implicitly exploited in the proof of~\cref{prop:binary-pgm-preprocessing} and is closely related to~\cref{rem:mqc-endpoints}.}
\begin{equation*}
\sigma_b^{\MQC}(r,\eta) \coloneqq \mathcal{B}_\eta\left(\sigma_b^{\PSC}(r)\right) = (1-\eta)\sigma_b^{\PSC}(r)+\eta\sigma_{b \oplus 1}^{\PSC}(r).
\end{equation*}
On one hand, by data processing---the monotonicity of quantum relative entropy under channels~\cite{Lin75,Uhl77}---adding BSC noise naturally decreases the Holevo information of this new channel. On the other hand, also by data processing (\cref{lem:pgm-data-processing}), the PGM bit error rate goes up (cf.~\cref{rem:mqc-data-processing}). Thus, the success of this approach comes down to numerically inspecting whether the Holevo information can decay faster than the PGM bit error rate. Remarkably, for a tiny value of $\eta$, the drop in Holevo information outpaces the increase in the PGM bit error rate, leading to an improvement over the bound attained by the PSC, i.e., the first MRRW bound (see~\cref{fig:mixed-optimization}). This is the content of \cref{prop:mqc-strict-improvement}.

\subsection{Unique decoding, list decoding, and PGM decoding}
\label{subsec:unique-list-pgm}

There is a useful analogy behind our work that we expound on in this subsection. If one insists on including the rate bounds derived from unique-decoding arguments in our narrative, then Hamming~\cite{Ham50} and Singleton~\cite{Sin64} bounds should correspond to the BSC and BEC, respectively. By resorting to more standard notions than the ones in~\cref{thm:main-result}, one can make the case that the stronger Elias--Bassalygo and Plotkin bounds should instead be interpreted as the `list-decoding thresholds' for the BSC and BEC, respectively. This alludes to a preliminary theory of unique-decoding and list-decoding thresholds for all channels. Upon initial inspection, the Tiet\"av\"ainen~\cite{Tie90} bound and the first MRRW bound corroborate this pattern, the former being the unique decoding threshold and the latter being the list-decoding threshold of some appropriate channel. When presented with this pattern, one could surmise the existence of a channel whose unique-decoding and list-decoding thresholds are prescribed by the Tiet\"av\"ainen and the first MRRW bounds, respectively. As the preceding discourse has inadvertently revealed, the pure-state channel (PSC) precisely fulfills that role, but that comes at the expense of reassessing what list-decoding is ultimately offering us, as we will now explain. This analogy is summarized in the table below.
\begin{center}
\renewcommand{\arraystretch}{1.45}
\small
\begin{tabularx}{\textwidth}{@{}c > {\centering\arraybackslash}X > {\centering\arraybackslash}X@{}}
\toprule
\textbf{Channel} & \textbf{Unique decoding} & \textbf{List/PGM decoding}\\
\midrule
$\BSC(p)$ & Hamming: $p < \delta/2$ gives $R_2(\delta) \leq 1-h(\delta/2)$ &
Elias--Bassalygo: $p < J(\delta)$ gives $R_2(\delta) \leq 1-h(J(\delta))$\footnote{Here, $J(\delta) = (1-\sqrt{1-2\delta})/2$ is the binary Johnson radius~\cite{Joh62}.}\\
\addlinespace
$\BEC(p)$ & Singleton: $p < \delta$ gives $R_2(\delta) \leq 1-\delta$ &
Plotkin: $p < 2\delta$ gives $R_2(\delta) \leq 1-2\delta$\\
\addlinespace
$\PSC(p)$ & Tiet\"av\"ainen: computational-basis decoding at $p < \delta/2$ gives
$R_2(\delta) \leq h\left(\frac{1}{2}-\sqrt{\frac{\delta}{2}\left(1-\frac{\delta}{2}\right)}\right)$ &
first MRRW: PGM bit error rate $p < \delta$ gives
$R_2(\delta) \leq h\left(\frac{1}{2}-\sqrt{\delta(1-\delta)}\right)$ \\
\bottomrule
\end{tabularx}
\end{center}

The fact that the Tiet\"av\"ainen bound naturally corresponds to the unique-decoding analog over PSC is not difficult to observe. Indeed, measuring the output of $\PSC(p)$ in the computational basis yields an output over the classical channel $\BSC(p)$. Thus, if $p < \delta/2$, then the concentration of the error fraction below $\delta/2$ and the minimum distance assumption lead to the block error tending to zero as $n \to \infty$. Invoking the Holevo bound~\cite{Hol73a} and letting $p \to (\delta/2)^-$ therefore yields the rate bound
\begin{equation*}
R_2(\delta) \leq h\left(\frac{1}{2}-\sqrt{\frac{\delta}{2}\left(1-\frac{\delta}{2}\right)}\right),
\end{equation*}
which is the Tiet\"av\"ainen bound~\cite{Tie90}.\footnote{The main focus of~\cite{Tie90} was on proving a much stronger claim. Namely, it showed that the relative \emph{exact} dual covering radius of a binary linear code of relative distance $\delta$ is at most $\frac{1}{2}-\sqrt{\frac{\delta}{2}\left(1-\frac{\delta}{2}\right)}$.} Note that this bound did not leverage the coherence of the PSC output at all.

Enticed by this observation, one is therefore led to speculate that the first MRRW bound might be attainable over PSC for $p < \delta$ via list-decoding. As the preceding discussion reveals, this roadway runs into a revealing obstacle. Indeed, the coherence of the output that one needs to crucially leverage is not amenable to the combinatorial nature of list-decoding, eventually leading us to consider the alternative (and more concrete) approach of PGM decoding.

Nonetheless, the two approaches have a key common feature worth spelling out. Observe that the constant-size list already gives the operational conclusion needed by the rate argument. Namely, if a decoder returns a list of at most $L = O(1)$ codewords containing $\bc$ with probability $1-o(1)$, then choosing uniformly from that list returns $\bc$ with probability at least $(1-o(1))/L$, yielding a constant chance of successful block decoding. This is exactly what one derives from applying PGM over PSC (\cref{cor:block-error}) and what is ultimately needed. Whereas the PGM approach engulfs both classical and quantum channels, it appears that list decoding does not, as far as we understand.

\subsection{Syndrome duality}
\label{subsec:syndrome-duality}

We now motivate our shift to quantum in reproving the first MRRW bound by presenting a connection between decodability over PSC and the well-established (and elegant) dual covering arguments of~\cite{FT05,NS09}. For a binary linear code $\cC \leq \F_2^n$ whose distance is $d$, their arguments prove the first MRRW bound by showing that translates of the Hamming ball of radius $\frac{n}{2}-\sqrt{d(n-d)}+o(n)$ by the words of $\cC^\perp$ cover at least a $1/n$ fraction of $\F_2^n$. Since $\abs{\cC} \cdot \abs{\cC^\perp} = 2^n$, such a covering forces $\abs{\cC}$ to be at most $n$ times the volume of that ball. Letting $d/n \to \delta$, we arrive at the first MRRW bound.
 
Two vital features of this route stand out: (i) it is carried out entirely on the dual code $\cC^\perp$, and (ii) the relative radius $\frac{1}{2}-\sqrt{\delta(1-\delta)}$ that it requires is exactly the Bernoulli parameter that will govern the dual side of the PSC. By observing the Holevo information identity for PSC below, the connection between our PSC approach for the first MRRW bound and previous approaches becomes well-motivated. Note that this proposition has been observed in previous works~\cite{Ren18,RP21,CT24,BCT26}. We mention it here as it substantiates our channel-based approach with previous approaches to the first MRRW bound.

\begin{prop}[{Syndrome duality --- \cite[Proposition 4]{BCT26};\cite[Lemma 16]{RP21}}]
\label{prop:syndrome-duality}
Let $\bc$ be a uniform codeword of a binary linear code $\cC \leq \F_2^n$. Let $\sigma_{\bc}$ be the output state of the input codeword $\bc$ over $\PSC(p)$. Let $\mathbf{z} \sim \Ber(\frac{1}{2}-\sqrt{p(1-p)})^n$. Then
\begin{equation}
\label{eq:psc-linear-mutual-information}
I(\bc;\sigma_{\bc}) = H\left(\mathbf{z} \bmod{\cC^\perp}\right),
\end{equation}
where the entropy on the right is that of the random coset of $\mathbf{z}$ in $\F_2^n/\cC^\perp$. Note that by~\eqref{eq:holevo-formula}, the left-hand side is the block-level Holevo information.\footnote{As an interesting consequence of syndrome duality, if one measures $\sigma_{\bc}$ in the computational basis, the quantum data processing inequality thus yields the inequality $H(\mathbf{z} \bmod{\cC^\perp}) = I(\bc;\sigma_{\bc}) \geq I(\bc;\bc+\be)$ where $\be \sim \Ber(p)^n$. This inequality by itself recovers a statistical analog of Tiet\"av\"ainen~\cite{Tie90} but not the first MRRW bound. Nonetheless, this elegant inequality readily follows from~\cref{prop:syndrome-duality}, albeit using a non-trivial but simple quantum tool, and, to the best of our knowledge, seems not to have been previously noted in the literature.} 
\end{prop}

Indeed, one can reinterpret the works of~\cite{FT05,NS09} as constructing a distribution $\cD$ over $\F_2^n$ satisfying the following conditions for a sample $\mathbf{v} \sim \cD$: (i) permuting coordinates of $\mathbf{v}$ yields the same distribution $\cD$, (ii) the event $\wt(\mathbf{v}) \leq \frac{n}{2}-\sqrt{d(n-d)}+o(n)$ occurs with probability $1$, and (iii) the quantity $\mathbf{v} \bmod{\cC^\perp}$ attains at least $1/n$ of the elements in $\F_2^n/\cC^\perp$. Given conditions (i) and (ii), it seems natural to consider beforehand whether one can just instantiate a similar approach but with $\Ber(\frac{1}{2}-\sqrt{\delta(1-\delta)})^n$ in place of $\cD$. By~\cref{prop:syndrome-duality}, this is indeed what decoding over the PSC is achieving under the hood.

Next, we sketch the proof of~\cref{prop:syndrome-duality} by setting up the PSC as a pure bipartite system. Indeed, write $\ket{\cC} \coloneqq \abs{\cC}^{-1/2}\sum_{c \in \cC}{\ket{c}}$ for the uniform coherent superposition over the code. Let $A$ and $B$ be two $n$-qubit registers, and let $\mathrm{CNOT}_{A \to B} \coloneqq \sum_{x \in \bits^n}{\ketbra{x}{x}_A \otimes X^x_B}$ be the transversal CNOT with control $A$ and target $B$. Applying it to the code register and the coherent Bernoulli state, and using $X^c\ket{\Ber(p)^n} = \ket{\sigma_c}$, gives
\begin{equation}
\label{eq:psc-bipartite}
\ket{\Psi_{\cC}(p)}
\coloneqq \mathrm{CNOT}_{A \to B}\left(\ket{\cC}_A \otimes \ket{\Ber(p)^n}_B\right)
=\frac{1}{\sqrt{\abs{\cC}}}\sum_{c \in \cC}{\ket{c}_A \otimes \ket{\sigma_c}_B}.
\end{equation}
Measuring $A$ in the computational basis leaves $B$ in the PSC output of a uniform codeword, so this pure state purifies the ensemble of \cref{prop:syndrome-duality}.

Now apply the Hadamard gate $\mathsf{H}$ to each of the $2n$ qubits. On the code register it exchanges the code for its dual, $\mathsf{H}^{\otimes n}\ket{\cC} = \ket{\cC^\perp}$. On the noise register, it moves the Bernoulli parameter to $q \coloneqq \frac{1}{2}-\sqrt{p(1-p)}$. That is, $\mathsf{H}^{\otimes n}\ket{\Ber(p)^n} = \ket{\Ber(q)^n}$. It additionally reverses the direction of every CNOT, leading us to the identity
\begin{equation}
\label{eq:psc-bipartite-dual}
\begin{split}
\mathsf{H}^{\otimes 2n}\ket{\Psi_{\cC}(p)}
&=\mathrm{CNOT}_{B \to A}\left(\ket{\cC^\perp}_A \otimes \ket{\Ber(q)^n}_B\right) \\
&=\frac{1}{\sqrt{\abs{\cC^\perp}}}\sum_{d \in \cC^\perp}{\sum_{z \in \bits^n}{\sqrt{q^{\wt(z)}(1-q)^{n-\wt(z)}}\,\ket{d \oplus z}_A \otimes \ket{z}_B}}.
\end{split}
\end{equation}

The two sides reflect two different but dual communication problems. Measuring \eqref{eq:psc-bipartite} in the computational basis returns a uniform codeword $\bc$ in $A$ and $\bc \oplus \be$ in $B$ with $\be \sim \Ber(p)^n$, which is the ordinary problem of decoding a codeword of $\cC$ against Bernoulli-$p$ noise. Measuring \eqref{eq:psc-bipartite-dual} the same way returns $\bz \sim \Ber(q)^n$ in $B$ and $\bz+\bc^\perp$ in $A$, where $\bc^\perp$ is a uniformly random codeword from $\cC^\perp$. What has to be recovered on that side is the Bernoulli string, and what obstructs recovering it is the dual code---in a phrase, decoding Bernoulli-$q$ against the noise $\cC^\perp$. In this dual decoding problem, only $\bz$ modulo $\cC^\perp$ can be recovered, which is exactly the coset appearing in \eqref{eq:psc-linear-mutual-information}.\footnote{Our nomenclature behind ``syndrome duality'' comes from attempting to ground a classical intuition initially suggested by~\cref{subsec:unique-list-pgm}. Indeed, first observe that the coset $\bz+\cC^\perp$ corresponds to the syndrome of $\bz$ with respect to the code $\cC^\perp$. More precisely, if $G \in \F_2^{k \times n}$ is a parity-check matrix for $\cC^\perp$ (equivalently, a generator matrix of $\cC$) then $x \mapsto Gx$ has kernel exactly $\cC^\perp$ and therefore identifies the quotient $\F_2^n/\cC^\perp$ with the syndrome space of $\cC^\perp$. Ideally, one would like to say that the syndrome distribution $\Ber(q)^n \pmod{\cC^\perp} = (\Ber(p)^n)^{\perp} \pmod{\cC^\perp}$ is linked to the dual syndrome distribution $\Ber(p)^n \pmod{\cC}$ in some form. Ultimately, it turns out that $\Ber(q)^n \pmod{\cC^\perp}$ is instead linked to ``$\cC \pmod{\Ber(p)^n}$,'' and it is this latter notion that seems to be best captured in the quantum identity linking~\eqref{eq:psc-bipartite} and~\eqref{eq:psc-bipartite-dual}.}

\subsection{Related works}
\label{subsec:related-works}

We now position our study within several relevant research topics.

\paragraph{The linear-programming route to the MRRW bounds.} The original proofs of both MRRW bounds run through Delsarte's linear program~\cite{Del72,Del73}. By linear programming (LP) duality, one can thus obtain rate upper bounds by finding dual LP solutions, and that involves analyzing polynomials with nonnegative Krawtchouk expansions that are nonpositive for all evaluations from the distance $d$ to $n$. McEliece, Rodemich, Rumsey, and Welch~\cite{MRRW77} constructed those certifying the first and second bounds, the second via the Bassalygo--Elias reduction to constant-weight codes. The optimal choices of such polynomials were later systematized by Levenshtein~\cite{Lev95}. 

Another pathway to the Delsarte LP is through the Lov\'asz theta function of the Hamming graph~\cite{Lov79,MRR78,Sch79}. This approach has been strengthened in several directions: semidefinite constraints from the Terwilliger algebra~\cite{Sch05}, complete hierarchies for linear codes~\cite{CJJ22,LL23b,CJJLL26}, and spectrally constructed dual solutions in general association schemes~\cite{CDA25}. These relaxations tighten finite-length bounds, yet no asymptotic improvement over~\cite{MRRW77} has been extracted from them. One can view our information-theoretic object---a cq channel over which minimum distance forces decodability---as a proxy for the dual certificates of the Delsarte LP in a different optimization problem. The resulting MQC and 2MQC bounds (\cref{thm:first-mrrw-improvement,thm:second-mrrw-improvement}) lead to asymptotic improvements over the first and second MRRW bounds.

Considering the quantum component of this work, we mention the extension of this line of work to quantum codes. Following the LP-style approach, one can place the parameters of a quantum code into a Delsarte-style linear program~\cite{AL99} that Rains sharpened through shadow enumerators~\cite{Rai99}. These apply in particular to the CSS codes~\cite{CS96,Ste96}, where the quantum code is assembled from a nested pair of classical codes, and more generally to stabilizer codes through their additive $\mathrm{GF}(4)$ description~\cite{CRSS98}, with the analysis again running through Krawtchouk polynomials.

\paragraph{Alternative routes to the first MRRW bound.} The first MRRW bound has since been reproved through approaches different from the LP/SDP bounds, some of which are directly relevant to our PSC derivation (see~\cref{subsec:syndrome-duality}). Friedman and Tillich~\cite{FT05} and Navon and Samorodnitsky~\cite{NS09} reproved the first MRRW bound through a neat dual covering argument, which involved reinterpreting the distance of $\cC$ as the expansion of the Cayley graph $\text{Cay}(\F_2^n/\cC^\perp,\{e_1, \ldots , e_n\})$ and arguing about the covering radius of the dual code $\cC^\perp$. Through syndrome duality (\cref{prop:syndrome-duality}), this approach directly links to our PSC derivation, though the pretty good criterion (\cref{thm:main-result}) needs no dual code and applies to nonlinear codes. The same identity also places the PSC beside code smoothing~\cite{BLVW19,YZ21,DADRT23,PB23,DAR25,PB25a,AGNTV26}, which asks when noise added to a random codeword flattens toward uniform. Indeed, the information the PSC retains is precisely what smoothing has not yet erased. Finally, we mention other similar Fourier--analytic proofs by Samorodnitsky~\cite{Sam23} and Linial and Loyfer~\cite{LL23a} and additionally barriers to pushing the Fourier--covering route toward the second MRRW bound~\cite{Sam25}.

\paragraph{Rate-distance tradeoffs for LDPC codes.} For Gallager's LDPC codes~\cite{Gal62}, the strongest known rate--distance bounds also come from the dual-covering side. Iceland and Samorodnitsky~\cite{IS15} sharpened the coset-ball bound of~\cite{NS09} for sparse duals, and Shangguan and Yang~\cite{SY26} sharpened it further through a combination of the coset-weight generating function and a shortening argument from~\cite{BHL06}. In light of syndrome duality (\cref{prop:syndrome-duality}), these approaches are naturally linked to the PSC-locality refinement in \cref{thm:psc-ldpc-improvement}.

\paragraph{Classical coding bounds and cq channels.} Perhaps the work closest in spirit to ours, in the sense of bounding a quantity in classical zero-error information theory using cq channels, is Dalai's celebrated work~\cite{Dal13}. Dalai gave a quantum-information-theoretic interpretation of Lov\'asz's bound on the Shannon capacity of a graph. He associated orthonormal representations of the confusability graph with pure-state classical--quantum channels and showed that the zero-error bound arising as the infinite-parameter limit of the cq sphere-packing bound recovers Lov\'asz's theta function. Thus, classical Shannon capacity can be upper-bounded through a cq-channel converse framework. Dalai and Winter~\cite{DW17} proved a sphere-packing bound for constant-composition cq codes and obtained an Elias-type bound by optimizing over a family of test channels. Their use of this optimization is perhaps the closest precedent for our channel-design viewpoint.

\paragraph{PSC decoding and channel duality.} The PSC has been studied extensively as a communication channel in its own right. Indeed, it is one instance of a large body of work on quantum state discrimination~\cite{Hel69,Hol73b,YKL75,HW94,BK02}. In the case of binary linear codes, decoding over PSC with a uniform prior is equivalent to state discrimination of \emph{geometrically uniform} ensembles, for which we know that PGM achieves the optimal block (and bit) decoding error~\cite{BKMH97,UTHH99,EF01,EMV04,ZCCL25,PR25}. 

The PSC has also been studied through a channel coding lens. Indeed, Burnashev and Holevo~\cite{BH98} initiated the reliability-function analysis of PSC; Wilde and Guha~\cite{WG13} constructed capacity-achieving polar codes for binary-input cq channels; Renes~\cite{Ren17} introduced belief propagation with quantum messages; Piveteau and Renes~\cite{PR25} gave efficient optimal decoders for codes with efficient trellis representations; and Mandal and Pfister~\cite{MP26} developed quantum message passing for group-covariant generalizations. Renes~\cite{Ren18} placed the PSC within a general duality between a code on a cq channel and its dual code on the dual channel, whose BSC--PSC specialization, made explicit by Rengaswamy and Pfister~\cite{RP21}, is the source of the syndrome duality used here (\cref{prop:syndrome-duality}). The decoders of Piveteau and Renes~\cite{PR25} proceed by quantum belief propagation, and their optimality holds bitwise as well as blockwise: grouping the outcomes of the block PGM according to a single codeword bit again yields the Helstrom-optimal measurement for that bit. That bitwise form is what our LDPC refinement rests on (\cref{fact:psc-bitwise-optimal}). Various decoding criteria have also been studied: code-aware unambiguous discrimination~\cite{MSPJ26}, its code-agnostic linear-programming counterpart~\cite{BC25}, and zero-error list decoding for pure-state channels~\cite{DGL26} certify figures of merit different from the posterior bit error rate driving our converse.

\paragraph{Regev's reduction and decoded quantum interferometry.} A recent algorithmic line makes coherent PSC decoding a computational primitive. Regev's quantum reduction for lattices~\cite{Reg09} was adapted to random linear codes by Debris-Alazard, Remaud, and Tillich~\cite{DART24}; Chailloux and Tillich~\cite{CT24} formulated the resulting quantum-decoding problem, whose Bernoulli case is exactly decoding over the binary PSC, and Blanvillain, Chailloux, and Tillich~\cite{BCT26} extended the analysis to general memoryless noise and the induced PGM threshold. Decoded quantum interferometry~\cite{JSWZSKIKB25} runs the same pipeline to reduce optimization to decoding: a coherent decoder uncomputes the codeword from a superposition of its noisy translates so that a Fourier transform can sample the dual code. Chailloux and Tillich~\cite{CT25} prove a generic reduction from syndrome decoding to coset sampling that accommodates imperfect decoders, Marwaha et al.~\cite{MFGH25} show that DQI constructively realizes a coding bound based on the MacWilliams identity, and Anschuetz, Gamarnik, and Lu~\cite{AGL25} prove overlap-gap obstructions for typical MAX-$k$-XOR-SAT instances. For all of these, the value of the PSC lies in decoding it efficiently and coherently. Our focus, in contrast, on PSC is purely information-theoretic. As such, it enables us to consider the PGM, which seems in general difficult to efficiently implement. A general-purpose quantum algorithm for the PGM is known~\cite{GLMQW22}, while efficient decoders over the PSC have so far been confined to structured families, such as codes with small trellises~\cite{PR25} or those reachable by quantum message passing~\cite{Ren17,MP26}.

\paragraph{Quantum proofs of classical coding statements.} There have been some notable successes in the use of quantum methods in establishing purely classical statements. We already mentioned Dalai's work connecting the Lov\'{a}sz theta function to cq channels~\cite{Dal13}. Kerenidis and de Wolf~\cite{KdW04} showed an exponential lower bound on the blocklength of any 2-query locally decodable code (2-LDC) by constructing a quantum random-access code (QRAC) from the 2-LDC and invoking Nayak's lower bound~\cite{Nay99}. Our approach happens to be related to that of~\cite{KdW04}. Indeed, in the high-noise regime $p = \frac{1}{2} - \Omega(n^{-1/2})$, one can show that the output of a random codeword from any fixed 2-LDC is in fact a QRAC. The connection to our approach becomes more pronounced when one replaces Nayak's bound with the PGM in the proof of~\cite{KdW04}, as was done in~\cite{LdW25}. For more on the quantum method, see the survey of Drucker and de Wolf~\cite{DdW11}.

Finally, Regev's reduction and the quantum method suggest a natural synthesis: the former has primarily been developed as an algorithmic tool, while the latter has largely served as a proof technique. Syndrome duality (\cref{prop:syndrome-duality}) reframes the Fourier transform at the heart of Regev's reduction as an information-theoretic identity rather than a sampling procedure. From this perspective, the present work represents a first step toward unifying these viewpoints, and we believe this direction has considerable further potential.

\subsection{Concurrent work}
\label{subsec:concurrent-work}

Concurrent work by OpenAI~\cite{Ope26} also improves the first and second MRRW bounds, albeit using a different approach. We stress that all our results reported here were obtained prior to the announcement of \cite{Ope26}. The announcement~\cite{Ope26} did, however, influence our choice to share our work at this moment, while we are pursuing a more comprehensive understanding of our framework.

A comparative check with GPT-5.6 Sol Pro suggested that the quantitative improvements to the first MRRW bound in our work and \cite{Ope26} match, and that the improvement to the second MRRW bound in \cite{Ope26} is captured by our framework through an appropriate four-dimensional binary-input output-symmetric cq channel, though we have not verified this.

A defining aspect of our work is its foundation in the general modular framework of cq channels. This resulting viewpoint reduces the derivation of new bounds to finding suitable channels and makes the framework readily applicable to $q$-ary codes and LDPC codes.\footnote{Given the close connection between binary codes and high-dimensional sphere packing~\cite{DGS77,KL78,CE03}, we suspect that our cq channel coding framework could shed new light on spherical codes as well. We leave this as an interesting direction for future work.}

\subsection{Acknowledgments}

We are grateful to Thiago Bergamaschi, Anthony (Chi-Fang) Chen, Louis Golowich, Dominik Hangleiter, Nathan Ju, Zeph Landau, Tony Metger, Chris Pattison, Angelos Pelecanos, Ewin Tang, Francisca Vasconcelos, \'Agi Vill\'anyi, and John Wright for helpful conversations surrounding quantum information theory at the early stages of this project. The first author especially thanks Jack Spilecki for numerous insightful discussions surrounding PGM and the Petz recovery map.

\subsection{AI use statement}

For a couple of years, the authors toyed with the possibility of reconciling the two approaches to obtaining the exponential lower bounds for 3-query locally correctable codes, via the Kikuchi method~\cite{KM24,Yan24,KM24b} and the dual covering radius approach~\cite{IS20,AG24}, in a manner that also yields the first MRRW bound in the spirit of~\cite{NS09}. 

The recent activity of instantiating Regev's reduction~\cite{Reg09} for binary linear codes~\cite{CLZ22,DART24,CT24,JSWZSKIKB25,MFGH25,CT25,PR25,BCT26,KOW26} provided inspiration in this regard. It rekindled the old results of~\cite{Tie90}, from which the authors conceived of the list-decoding table in~\cref{subsec:unique-list-pgm}. Coupled with the simple realization that the Hadamard transform of $\ket{\Ber(p)^n}$ is just $\ket{\Ber(1/2 - \sqrt{p(1-p)})^n}$, which was eerily similar to the formulas appearing in~\cite{Tie90,MRRW77}, it led the authors to investigate a new quantum proof of the first MRRW bound, ideally one without complex calculations. This investigation was also partially inspired by the success of the quantum method~\cite{KdW04} within the scope of locally decodable codes.

After a long-term investigation involving several interactions with GPT-5.2, GPT-5.4, and GPT-5.5, the authors made progress on the problem and arrived at~\cref{cor:block-error}, along with its potential to prove the first MRRW bound when applied to the PSC. Following the release of GPT-5.6, further interactions led to a PSC sharp-threshold result analogous to~\cite{TZ00,PSW24}.\footnote{This appears to be provable by using syndrome duality (\cref{prop:syndrome-duality}) to recast the sharp-threshold problem as a cutoff phenomenon on a suitable Cayley graph, after which the methods of~\cite{Sal24} can be adapted to obtain an exponential cutoff phenomenon.} Combined with~\cref{cor:block-error}, this yielded a new quantum-based proof of the first MRRW bound for binary linear codes.

Then, when the authors presented this new quantum-based proof of the first MRRW bound to GPT-5.6 Sol Pro and asked whether the same framework could recover, or possibly improve upon, the second MRRW bound, GPT-5.6 Sol Pro heuristically proposed the MQC as a candidate channel. Later interactions led the authors to realize that the strong converse for cq channels suffices to complete the proofs of the desired rate upper bound claimed in \cref{thm:main-result}, leading to a new proof and improvement of the first MRRW bound. Subsequent interactions focused on the second MRRW bound led to the proposals of the masked PSC and the 2MQC.

Codex and Claude Code were employed when drafting the paper. The exposition and mathematical setup of this work were orchestrated by the authors, and the authors take full responsibility for the contents of this paper.

\subsection{Organization}
\label{subsec:organization}

In \cref{sec:preliminaries}, we introduce the quantum information theory tools needed to prove~\cref{thm:main-result}, specifically the PGM error data processing inequality (\cref{lem:pgm-data-processing}) and the cq channel coding strong converse (\cref{thm:cq-strong-converse}). In the same section, we include closed formulas needed to derive rate bounds, and additionally the masking operation for cq channels to derive the second MRRW and 2MQC rate bounds. In~\cref{sec:main-theorem-proof}, we prove the pretty good criterion (\cref{thm:main-result}). In \cref{sec:known-rate-bounds}, we rederive the standard bounds of Plotkin, Elias--Bassalygo, and both MRRW bounds from the BEC, BSC, PSC, and masked PSC, respectively. In \cref{sec:new-cq-channels}, we introduce the MQC and 2MQC and derive their corresponding bounds. In \cref{sec:comparing-rate-bounds}, we show that MQC and 2MQC yield strict improvements over the first and second MRRW bounds over all $\delta \in (0,\frac{1}{2})$. In \cref{sec:qary-codes}, we lay down a blueprint for how one can extend the pretty good criterion to $q$-ary codes. Finally, \cref{sec:psc-ldpc} proves the LDPC refinement (\cref{thm:psc-ldpc-improvement}), while \cref{app:preliminary-proofs} presents all the deferred proofs from~\cref{sec:preliminaries}.

\section{Preliminaries}
\label{sec:preliminaries}

This section fixes the notation, recalls standard quantum-information tools~\cite{NC10,Wil21,Ren22}, and packages the calculations reused by the channel constructions. For completeness, the proofs of the stated results in this section are presented together in Appendix~\ref{app:preliminary-proofs}.

We first fix the classical coding notation. For $u,v \in \bits^n$, let $d(u,v)$ denote their Hamming distance. For ease of exposition, we will assume throughout this work that $\abs{\cC} \geq 2$ for every binary code $\cC \subseteq \bits^n$ considered. The rate of a code $\cC$ is defined to be $R(\cC) \coloneqq n^{-1}\log{\abs{\cC}}$. Its minimum distance is defined to be $d_{\min}(\cC) \coloneqq \min\{d(c,c') : c,c' \in \cC,\ c \neq c'\}$. Throughout this work, $\bc$ denotes a uniformly random codeword from $\cC$.

Next, we fix the quantum notation. A ket $\ket{y}$ is a unit column vector, whereas $\bra{y}$ is its conjugate transpose. The expression $\ketbra{\varphi}{\psi}$ denotes the outer product of $\ket{\varphi}$ and $\bra{\psi}$, while $\braket{\varphi|\psi}$ denotes the inner product between $\ket{\varphi}$ and $\ket{\psi}$. In particular, $\ketbra{\varphi}{\varphi}$ is the rank-one orthogonal projector onto the span of $\ket{\varphi}$.

Next, we introduce the quantum extension of a classical probability distribution.

\begin{definition}[Density matrices]
\label{def:density-matrices}
A \textbf{density matrix} is a finite-dimensional positive semidefinite matrix $\varrho$ satisfying the normalization $\Tr(\varrho) = 1$. We write $\Dens(\C^d)$ for the set of $d \times d$ density matrices, $\supp(\varrho)$ for the range of $\varrho$, and $\Pi_\varrho$ for the orthogonal projector onto that range.
\end{definition}

For a density matrix $\varrho$, its von Neumann entropy is defined to be $S(\varrho) \coloneqq -\Tr(\varrho\log{\varrho})$. For $t \in [0,1]$, the binary entropy function is defined to be $h(t) \coloneqq -t\log{t}-(1-t)\log(1-t)$. Note that all logarithms throughout this work are base $2$.

In order for a receiver to extract a classical outcome from a quantum state, they would need to measure it. This leads us to the formalism of a (general) measurement.

\begin{definition}[Positive operator-valued measurement (POVM)]
\label{def:povm}
A \textbf{positive operator-valued measurement (POVM)} with outcome set $\cX$ is a family of $d \times d$ positive semidefinite matrices $\{M_{\hat{x}} \in \C^{d \times d} : \hat{x} \in \cX\}$ satisfying the identity $\sum_{\hat{x} \in \cX}{M_{\hat{x}}} = I_d$. Measuring a density matrix $\varrho \in \Dens(\C^d)$ with respect to this POVM returns $\hat{x}$ with probability $\Tr(\varrho M_{\hat{x}})$.
\end{definition}

Throughout this work, Hermitian matrices will be called \textbf{observables}, the quantum extension of classical random variables. Consequently, we will call the POVM matrices $M_{\hat{x}}$ observables. Note that projective (von Neumann) measurements are POVMs in which the observables are orthogonal projection matrices ($M_{\hat{x}}^2 = M_{\hat{x}}$ for all $\hat{x} \in \cX$ and $M_{\hat{x}_1}M_{\hat{x}_2} = 0$ for all $\hat{x}_1, \hat{x}_2 \in \cX$ with $\hat{x}_1 \neq \hat{x}_2$). While projective (von Neumann) measurements are the prototypical notion of measurements, POVMs can be instantiated as projective measurements using ancilla qubits~\cite[Section 2.2.8]{NC10}.

\subsection{Binary-input cq channels}
\label{subsec:binary-input-cq-channels}

In this subsection, we recall quantum channels and then specialize to binary-input output-symmetric cq channels and their memoryless extensions.

\begin{definition}[Quantum channel]
\label{def:quantum-channel}
Fix positive integers $d_A$ and $d_B$. A \textbf{quantum channel} is a linear map $\N : \C^{d_A \times d_A} \to \C^{d_B \times d_B}$ satisfying the following two properties:
\begin{enumerate}
\item \emph{(Trace preservation)} For every $X \in \C^{d_A \times d_A}$, one has $\Tr(\N(X)) = \Tr(X)$.
\item \emph{(Complete positivity)} For every positive integer $r$, the map $\N \otimes \operatorname{id}_r$ takes positive semidefinite matrices in $\C^{d_A r \times d_A r}$ to positive semidefinite matrices in $\C^{d_B r \times d_B r}$, where $\operatorname{id}_r$ is the identity map on $r \times r$ matrices.
\end{enumerate}

We say $\N$ is a \textbf{classical--quantum (cq) channel} if it has the following measure-and-prepare form. Fix a finite input alphabet $\cX$, set $d_A = \abs{\cX}$, and choose one output state $\sigma_x \in \Dens(\C^{d_B})$ for each $x \in \cX$. Relative to the orthonormal basis $\{\ket{x}\}_{x \in \cX}$, the corresponding quantum channel acts as
\begin{equation*}
\N(\varrho) \coloneqq \sum_{x \in \cX}{\bra{x} \varrho \ket{x} \sigma_x}.
\end{equation*}
In this work, we abbreviate channels with the simpler notation $\N : x \mapsto \sigma_x$.
\end{definition}

We use two standard quantum channels. For an $n$-partite system of $d$-dimensional states $\varrho_1 \otimes \cdots \otimes \varrho_n$ and an index $i \in \{1, \ldots , n\}$, the \textbf{partial trace channel to the $i$'th qudit} $\T_i : \C^{d^n \times d^n} \to \C^{d \times d}$ is defined on product states by
\begin{equation}
\label{eq:partial-trace-product}
\T_i : \varrho_1 \otimes \cdots \otimes \varrho_n \longmapsto \varrho_i.
\end{equation}
By linearity, this definition extends to non-product states. Operationally, $\T_i$ discards every qudit except the $i$'th one. We use it to reduce blockwise error to bitwise error in the proof of~\cref{thm:expected-hamming}.

The other quantum channel we consider in this work is the \textbf{X-Pauli channel} (or bit-flip channel). For $\eta \in [0,1]$, it is the random-unitary qubit channel
\begin{equation}
\label{eq:pauli-x-channel}
\mathcal{B}_\eta(\varrho) \coloneqq (1-\eta)\varrho + \eta X\varrho X.
\end{equation}
Here, $X$ is the Pauli observable that exchanges the computational-basis vectors. Thus, $\mathcal{B}_\eta$ applies $X$ with probability $\eta$ and otherwise does nothing. In Bloch coordinates, it leaves the $X$ component unchanged and multiplies the $Z$ component by $1-2\eta$ (cf.~\eqref{eq:mqc-pauli-form}). In this work, the channel $\mathcal{B}_\eta$ is used as post-processing to construct the MQC in \cref{subsec:how-mixing,subsec:mqc}.

The binary-input cq channels employed in this work form a particularly simple class of cq channels, defined as follows.

\begin{definition}[Binary-input cq channel]
\label{def:binary-cq-channel}
For a positive integer $d$, a \textbf{binary-input classical--quantum (cq) channel} is a map $\Phi_\sigma : \bits \to \Dens(\C^d)$ specified by $\Phi_\sigma(b) = \sigma_b$. It is \textbf{output-symmetric} if there exists a $d \times d$ unitary matrix $U$ that exchanges the outputs, meaning that $U\sigma_0U^\dagger = \sigma_1$ and $U\sigma_1U^\dagger = \sigma_0$. The memoryless extension of $\Phi_\sigma$ sends a codeword $c = (c_1,\ldots,c_n)$ to the product state $\sigma_c \coloneqq \Phi_\sigma^{\otimes n}(c) = \sigma_{c_1} \otimes \cdots \otimes \sigma_{c_n}$.
\end{definition}

Next, we define the first key quantity, the Holevo information~\cite{Hol73a}.

\begin{definition}[Holevo information]
\label{def:holevo-information}
The \textbf{Holevo information} of a finite ensemble of density matrices $\cE \coloneqq \{(\pi_x,\sigma_x) : x \in \cX\}$ with average state $\bar{\sigma} \coloneqq \sum_{x \in \cX}{\pi_x\sigma_x}$ is defined to be 
\begin{equation}
\label{eq:holevo-formula}
\chi(\cE) \coloneqq I(\bx;\sigma_\bx) = S(\bar{\sigma})-\sum_{x \in \cX}{\pi_xS(\sigma_x)}.
\end{equation}
For the uniform binary ensemble $\cE = \{(\frac{1}{2},\sigma_0),(\frac{1}{2},\sigma_1)\}$, we write $\chi(\sigma_0,\sigma_1) \coloneqq \chi(\cE)$.\footnote{When the states $\{\sigma_x\}_{x \in \cX} \subseteq \Dens(\C^d)$ are specified without an associated prior distribution, we take the prior to be uniform when defining the Holevo information of the ensemble. This convention is assumed throughout this work.}
\end{definition}

The operational meaning of Holevo information is as follows. For a random label $\bx \sim \pi$, the Holevo information is precisely the von Neumann mutual information $I(\bx;\sigma_{\bx})$.\footnote{To be precise, this mutual information is relative to the bipartite state $\sum_{x \in \cX}{\pi_x\ketbra{x}{x} \otimes \sigma_x}$.} For commuting outputs $\{\sigma_x\}_{x \in \cX}$, it reduces to ordinary mutual information.

Next, we define the quantum extension of channel capacity for cq channels. For comparison, the capacity of a classical discrete memoryless channel $W : x \mapsto W(\cdot \mid x)$ is $C(W) = \max_\pi I(\bx;\by)$, where $\bx \sim \pi$ and $\by$ is the channel output. The cq definition is formally identical, with Holevo information in place of classical mutual information.

\begin{definition}[Classical (HSW) capacity of cq channels]
\label{def:holevo-capacity}
For any cq channel $\Phi : \cX \to \Dens(\C^d)$ with $\Phi(x) = \sigma_x$, the \textbf{classical capacity}, normalized per channel use, is defined to be
\begin{equation*}
C(\Phi)
\coloneqq
\limsup_{n \to \infty}\frac{1}{n}\max_{\pi^{(n)}}
\chi\left(\{(\pi^{(n)}(x^n),\sigma_{x_1} \otimes \cdots \otimes \sigma_{x_n}) : x^n \in \cX^n\}\right),
\end{equation*}
where the maximum is over all probability distributions $\pi^{(n)}$ on $\cX^n$.
\end{definition}

The above capacity has a single-letter characterization as the \emph{Holevo capacity} of the cq channel, defined as 
\begin{equation*}
    \max_{\pi} \chi \left(\left\{(\pi(x),\sigma_x) : x \in \cX\right\}\right), 
\end{equation*}
where the maximum is over all probability distributions $\pi$ on $\cX$.
Further, when the cq channel is output-symmetric, the Holevo capacity equals the Holevo information for the uniform prior, as stated in the following proposition. Consequently, this justifies the reference to the Holevo information as the channel capacity throughout this work. The proof appears in \cref{appsubsec:proof-output-symmetric-capacity}.

\begin{prop}[Capacity of a binary-input output-symmetric cq channel]
\label{prop:output-symmetric-capacity}
Let $\Phi_\sigma : b \mapsto \sigma_b$ be a binary-input output-symmetric cq channel. Its capacity satisfies
\begin{equation*}
C(\Phi_\sigma) = \chi(\sigma_0,\sigma_1).
\end{equation*}
\end{prop}

As a slight digression, the Holevo--Schumacher--Westmoreland (HSW) coding theorem~\cite{Hol98,SW97,HJSWW96} identifies cq channel capacity with the maximum Holevo information over input priors, much like Shannon's formula in the classical case. Its positive side states that, at every rate below $C(\Phi)$, there are length-$n$ codes and decoder POVMs whose average block error tends to zero. 

Now, we turn to the more relevant side of the cq channel coding theorem, namely the previously established cq channel coding strong converse~\cite{Win99a,ON99,MO17} that turns nonvanishing block-decoding success into a rate bound. Specifically, \cite[Theorem~13]{Win99a} shows that every code with decoding error at most $\lambda < 1$ has rate at most $C(\Phi)+O_\lambda(1/\sqrt{n})$, while \cite[Theorem~1]{ON99} lower-bounds the average error probability of an arbitrary code by an exponential expression that forces the success probability to decay exponentially at every rate above the capacity.\footnote{The exact strong-converse exponent was later determined in~\cite{MO17}.} 

For our purposes, the version below is tailored specifically to binary-input output-symmetric cq channels. Thus, by~\cref{prop:output-symmetric-capacity}, we can state it with the Holevo information $\chi(\sigma_0,\sigma_1)$ in place of the Holevo capacity. For completeness, its proof appears in~\cref{appsubsec:proof-fixed-channel-strong-converse}.

\begin{theorem}[{cq channel coding strong converse --- \cite[Theorem~13]{Win99a} and \cite[Theorem~1]{ON99}}]
\label{thm:cq-strong-converse}
Fix a finite-dimensional binary-input output-symmetric cq channel $\Phi_\sigma$. There is a constant $L_\sigma \geq 1$, depending only on its two output states $\sigma_0$ and $\sigma_1$, such that the following holds. Let $\cC \subseteq \bits^n$ be any binary code, and let $\bc$ be a uniformly random codeword from $\cC$. For any $\eps > 0$ and decoder POVM $\{D_c : c \in \cC\}$, if $\bchat$ is obtained by measuring $\sigma_{\bc}$ with this POVM, then
\begin{equation}
\label{eq:fixed-channel-exponential-converse}
\Prb[\bchat = \bc] \leq 2^{-n[R(\cC)-\chi(\sigma_0,\sigma_1)-2\eps]} + 4\exp\left(-\eps^2n/L_\sigma^2\right).
\end{equation}
Consequently, any fixed positive rate gap above $\chi(\sigma_0,\sigma_1)$ forces the average decoding success probability to decay exponentially in $n$.
\end{theorem}

\begin{remark}[Qualitative decay suffices]
\label{rem:qualitative-decay-suffices}
For the asymptotic conclusion of~\cref{thm:main-result}, it suffices if the successful block decoding probability tends to zero, at the expense of having weaker finite blocklength estimates. The exponential decay in~\cref{thm:cq-strong-converse} is used in the proof of~\cref{thm:main-result} to obtain the precise $O(n^{-1/2})$ lower-order term in~\cref{thm:main-result}, which is more than enough for the asymptotic version.
\end{remark}

\subsection{The pretty good measurement (PGM)}
\label{subsec:pgm-definition}

Throughout this subsection, we fix a finite ensemble of density matrices $\cE \coloneqq \{(\pi_x,\sigma_x) : x \in \cX\}$, where $\pi_x \geq 0$ and $\sum_{x \in \cX}{\pi_x} = 1$, and write $\bar{\sigma} \coloneqq \sum_{x \in \cX}{\pi_x \sigma_x} = \Eb_{\bx \sim \pi}[\sigma_{\bx}]$ for its average state. This subsection defines the pretty good measurement (PGM), also called the square-root measurement~\cite{Bel75,Hol79,HW94,HJSWW96}, records its exact compatibility with coarse-graining, and states the data-processing consequence used in \cref{sec:main-theorem-proof}. See Renes~\cite[Section~11.5]{Ren22} for a textbook treatment of PGM and~\cite[Lectures~6~and~7]{Wri24} for an instructive gentle exposition of it.

\begin{definition}[Pretty good measurement (PGM)]
\label{def:pgm}
The \textbf{pretty good measurement (PGM)} of the ensemble $\cE$ is the POVM over $\supp(\bar{\sigma})$ with observables $M_{\hat{x}} \coloneqq \pi_{\hat{x}}\bar{\sigma}^{-1/2}\sigma_{\hat{x}}\bar{\sigma}^{-1/2}$, where inverses are over $\supp(\bar{\sigma})$.\footnote{Note that these operators sum to the support projector $\Pi_{\bar{\sigma}}$.}
\end{definition}

When $\cE$ is a classical ensemble (i.e., the states $\{\sigma_x\}_{x \in \cX}$ are simultaneously diagonalizable), the PGM reduces to classical posterior sampling. Indeed, assume without loss of generality that all $\sigma_x$ are diagonal $d \times d$ density matrices. Then the average state $\bar{\sigma}$ and the operators $\{M_{\hat{x}}\}_{\hat{x} \in \cX}$ are also diagonal. Hence, conditioned on observing $\ell \in \{1, \ldots , d\}$, the PGM returns $\hat{x} \in \cX$ with probability
\begin{equation*}
(M_{\hat{x}})_{\ell,\ell} = (\pi_{\hat{x}}\bar{\sigma}^{-1/2}\sigma_{\hat{x}}\bar{\sigma}^{-1/2})_{\ell,\ell} = \frac{\pi_{\hat{x}} \Prb[\Bell = \ell \mid \bx = \hat{x}]}{\Prb[\Bell = \ell]} = \Prb[\bx = \hat{x} \mid \Bell = \ell],
\end{equation*}
which is precisely posterior sampling.

Next, we define our second key quantity, the PGM error, which we interchangeably call posterior error in this work.

\begin{definition}[PGM error]
\label{def:pgm-error}
Let $\bxhat$ be the outcome of the PGM for the ensemble $\cE$. The \textbf{PGM error} of $\cE$ is defined to be the quantity
\begin{equation*}
p_{\mathrm{e}}(\cE) \coloneqq \Prb[\bxhat \neq \bx] = \sum_{\substack{\hat{x},x \in \cX : \\ \hat{x} \neq x}}{\pi_x\Tr(\sigma_x M_{\hat{x}})}.
\end{equation*}
For binary ensembles $\cE_\alpha = \{(1-\alpha,\sigma_0),(\alpha,\sigma_1)\}$, we write $p_{\mathrm{e},\alpha}(\sigma_0,\sigma_1) \coloneqq p_{\mathrm{e}}(\cE_\alpha)$. In the more special uniform-prior case when $\alpha = \frac{1}{2}$, we write $p_{\mathrm{e}}(\sigma_0,\sigma_1) \coloneqq p_{\mathrm{e},\frac{1}{2}}(\sigma_0,\sigma_1)$ and call it the \textbf{PGM bit error rate} of the output state pair $(\sigma_0,\sigma_1)$.
\end{definition}

We next record the following central property of PGM, which was previously exploited in~\cite{KT08,BJL23,LdW25}. Its proof appears in \cref{appsubsec:proof-coarse-graining}.

\begin{lemma}[PGM is closed under coarse-graining]
\label{lem:coarse-graining}
Fix any classical map $f: \cX \to \cY$. Let $\cF \coloneqq \{(\nu_y,\omega_y) : y \in \cY\}$ be the pushforward ensemble of $\cE$ under $f$, where $\nu_y \coloneqq \Prb[f(\bx) = y]$ and $\omega_y \coloneqq \Eb[\sigma_{\bx} \mid f(\bx) = y]$. For a sample $\bx$ from $\cE$, if $\bxhat \in \cX$ is the outcome of the PGM of $\cE$ when applied to $\sigma_{\bx}$, then $f(\bxhat)$ is the outcome of the PGM of $\cF$ when applied to $\omega_{f(\bx)}$. As a result, we have the identity
\begin{equation*}
\Prb[f(\bxhat) \neq f(\bx)] = p_{\mathrm{e}}(\cF).
\end{equation*}
\end{lemma}

In~\cref{sec:main-theorem-proof}, the ensemble of interest will be $\cE = \{(\abs{\cC}^{-1}, \sigma_c) : c \in \cC\}$, which we will benignly conflate with $\cC$. For such ensembles, we will crucially rely on the following specialization of~\cref{lem:coarse-graining} applied to the projection map $f_i(c) = c_i$. This identity also explains why the PGM is preferable here to a hard maximum-a-posteriori (MAP) decoder.

\begin{lemma}[Blockwise PGM is simultaneously bitwise PGM]
\label{lem:block-pgm-bitwise}
For a uniform sample $\bc$ from $\cC$, let $\bchat \in \cC$ be the outcome of the PGM applied to $\sigma_{\bc}$. For any $i \in \{1, \ldots , n\}$ and $b \in \bits$, let $\pi_i(b) \coloneqq \Prb[\bc_i = b]$, and let $\sigma_i(b) \coloneqq \Eb[\sigma_\bc \mid \bc_i = b]$. Then $\bchat_i$ is the outcome of the PGM of the induced binary ensemble $\cE_i \coloneqq \{(\pi_i(0),\sigma_i(0)), (\pi_i(1),\sigma_i(1))\}$ when applied to $\sigma_i(\bc_i)$. Consequently,
\begin{equation*}
\Prb[\bchat_i \neq \bc_i] = p_{\mathrm{e}}(\cE_i).
\end{equation*}
\end{lemma}

We next record the monotonicity of the PGM error under post-processing the quantum outputs through any quantum channel. Note that this is a crucial step in establishing~\cref{thm:expected-hamming}. Its proof appears in \cref{appsubsec:proof-pgm-data-processing}.

\begin{lemma}[{PGM error is nondecreasing under quantum channels --- \cite[Section~11.5]{Ren22}}]
\label{lem:pgm-data-processing}
Let $\N$ be a completely positive trace-preserving map. Applying $\N$ to every state in $\cE$ cannot lower the PGM error. That is, 
\begin{equation*}
p_{\mathrm{e}}(\cE) \leq p_{\mathrm{e}}\left(\{(\pi_x,\N(\sigma_x)) : x \in \cX\}\right).
\end{equation*}
\end{lemma}

In~\cref{thm:expected-hamming}, we apply the preceding lemma to the product ensemble $\cE = \{(\abs{\cC}^{-1}, \sigma_c) : c \in \cC\}$. The partial trace $\T_i$ then discards every output except the $i$'th one.

Next, to handle the bit errors in the proof of~\cref{thm:main-result} after reducing to them, we will need the following prior-independent bound on each bit. Its proof appears in \cref{appsubsec:proof-uniform-prior}.

\begin{lemma}[Uniform prior maximizes the PGM bit error rate]
\label{lem:uniform-prior-worst}
Let $\sigma_0,\sigma_1 \in \Dens(\C^d)$ be an output-symmetric pair. For any $\gamma \in [0,1]$, let $p_{\mathrm{e},\gamma}(\sigma_0,\sigma_1)$ denote the PGM bit error rate for the binary ensemble $\cE_\gamma = \{(1-\gamma,\sigma_0),(\gamma,\sigma_1)\}$. Then
\begin{equation*}
p_{\mathrm{e},\gamma}(\sigma_0,\sigma_1) \leq p_{\mathrm{e},1/2}(\sigma_0,\sigma_1) = p_{\mathrm{e}}(\sigma_0,\sigma_1).
\end{equation*}
\end{lemma}

\subsection{Closed formulas for the PGM bit error rate computations}
\label{subsec:pgm-bit-error-formulas}

In this subsection, we record closed formulas needed for our PGM bit error rate computations. For output states $\sigma_0,\sigma_1 \in \Dens(\C^d)$ and bias parameter $\alpha \in [0,1]$, define the binary ensemble $\cE_\alpha \coloneqq \{(1-\alpha,\sigma_0),(\alpha,\sigma_1)\}$. We adopt the shorthand $\bar{\sigma}_\alpha \coloneqq (1-\alpha)\sigma_0+\alpha \sigma_1$ for their average state throughout this work. In the uniform-prior case when $\alpha = \frac{1}{2}$, we simply write $\bar{\sigma} \coloneqq \bar{\sigma}_{\frac{1}{2}}$. First, we record the following PGM trace identity. Its proof appears in \cref{appsubsec:proof-binary-pgm-formulas}.

\begin{prop}[Binary PGM error formulas]
\label{prop:binary-pgm-error}
The PGM error of the binary ensemble $\cE_\alpha$ satisfies the identity
\begin{equation}
\label{eq:binary-pgm-error}
p_{\mathrm{e},\alpha}(\sigma_0,\sigma_1)
= 2\alpha(1-\alpha)\Tr\left(\sigma_0\bar{\sigma}_\alpha^{-1/2}\sigma_1\bar{\sigma}_\alpha^{-1/2}\right).
\end{equation}
Moreover, if both output states are pure, i.e., $\sigma_b = \ketbra{\psi_b}{\psi_b}$ for unit vectors $\ket{\psi_0},\ket{\psi_1} \in \C^d$ and $f \coloneqq \abs{\braket{\psi_0|\psi_1}}^2$, then
\begin{equation}
\label{eq:non-uniform-pure-pgm-error}
p_{\mathrm{e},\alpha}(\sigma_0,\sigma_1)
= \frac{2\alpha(1-\alpha)f}{1+2\sqrt{\alpha(1-\alpha)(1-f)}}.
\end{equation}
In particular, in the uniform-prior case when $\alpha = \frac{1}{2}$, this becomes $p_{\mathrm{e}}(\sigma_0,\sigma_1) = \frac{1}{2}(1-\sqrt{1-f})$.
\end{prop}

Finally, we record how the PGM behaves when a classical binary channel precedes a cq preparation. This composition rule will be used when computing PGM bit error rates of MQC and 2MQC below. We present the proof in~\cref{appsubsec:proof-binary-pgm-formulas}.

\begin{prop}[Binary PGM error under classical pre-processing]
\label{prop:binary-pgm-preprocessing}
Fix $\eta \in [0,1]$ and $\be \sim \Ber(\eta)$. Consider the binary-input cq channel 
\begin{equation}
\label{eq:composed-cq-channel}
b
\longmapsto b \oplus \be
\longmapsto \sigma_{b \oplus \be}.
\end{equation}
Fix $\alpha \in [0,1]$, and let $\bbhat$ be the outcome of the PGM for the resulting cq channel in~\eqref{eq:composed-cq-channel} when given a non-uniform bit $\bb \sim \Ber(\alpha)$ as input. Set $\by \coloneqq \bb \oplus \be$ and $\theta \coloneqq \Prb[\by = 1] = \alpha + \eta - 2\alpha\eta$. Let $\byhat$ be the outcome of the PGM for the intermediate ensemble $\{(1-\theta,\sigma_0),(\theta,\sigma_1)\}$. Then the PGM outcome $\bbhat$ can be generated by first producing $\byhat$ and then drawing $\bbhat$ from the posterior law
\begin{equation*}
\Prb[\bbhat = b \mid \byhat = y] = \Prb[\bb = b \mid \by = y].
\end{equation*}
Moreover, if we let $\tau_b \coloneqq (1-\eta)\sigma_b + \eta\sigma_{b \oplus 1}$ for $b \in \bits$ be the output states of the composed cq channel in~\eqref{eq:composed-cq-channel}, then the binary PGM error for the composed channel satisfies
\begin{equation}
\label{eq:pgm-preprocessing-error}
p_{\mathrm{e},\alpha}(\tau_0,\tau_1) = 2\alpha(1-\alpha) - \left[
\frac{\alpha(1-\alpha)(1-2\eta)}
{\theta(1-\theta)}
\right]^2 \cdot \left[
2\theta(1-\theta)-p_{\mathrm{e},\theta}(\sigma_0,\sigma_1)
\right].
\end{equation}
In particular, for the uniform prior $\alpha = \frac{1}{2}$, the PGM error is\footnote{A more intuitive way to think about the uniform-prior formula is the following: the map $\bb \mapsto \bbhat$ can be decomposed as the composition of the maps $\bb \mapsto \bb \oplus \be = \by$ followed by $\by \mapsto \sigma_\by \mapsto \byhat$, then followed by $\byhat \mapsto \bbhat$. Since $\bb$ has a uniform prior, these three maps can therefore be recast as passing through the independent bit channels $\BSC(\eta), \BSC(p_{\mathrm{e}}), \BSC(\eta)$, respectively.}
\begin{equation}
\label{eq:pgm-preprocessing-error-uniform-prior}
p_{\mathrm{e}}(\tau_0,\tau_1) = \frac{1}{2}\left[1 - (1-2\eta)^2(1-2p_{\mathrm{e}}(\sigma_0,\sigma_1))\right].
\end{equation}
\end{prop}

\subsection{Closed formulas for the Holevo information computations}
\label{subsec:holevo-formulas}

In this subsection, we record the closed entropy formula used in the Holevo-information calculations of \cref{sec:known-rate-bounds,sec:new-cq-channels}. Every qubit appearing there is real in the computational basis and therefore lies in the plane spanned by the Pauli matrices
\begin{equation}
\label{eq:pauli-matrices}
X \coloneqq \begin{pmatrix} 0 & 1 \\ 1 & 0 \end{pmatrix}
\quad
\text{ and }
\quad
Z \coloneqq \begin{pmatrix}1 & 0 \\ 0 & -1 \end{pmatrix}.
\end{equation}
The following proposition reads the von Neumann entropy directly from these Pauli coefficients. Its proof appears in~\cref{appsubsec:proof-qubit-entropy}.

\begin{prop}[Qubit entropy in Pauli coordinates]
\label{prop:qubit-entropy}
For $D \in [0,1]$, define
\begin{equation}
\label{eq:qubit-entropy}
s(D) \coloneqq h\left(\frac{1-\sqrt{1-D}}{2}\right).
\end{equation}
If $x,z \in \mathbb{R}$ satisfy $x^2+z^2 \leq 1$, then $\varrho \coloneqq \frac{1}{2}(I+xX+zZ)$ is a density matrix with entropy
\begin{equation*}
S\left(\frac{1}{2}(I+xX+zZ)\right)
= h\left(\frac{1-\sqrt{x^2+z^2}}{2}\right)
= s(1-x^2-z^2).
\end{equation*}
Note that $4\det(\varrho) = 1-x^2-z^2$.
\end{prop}

\subsection{Channel masking}
\label{subsec:masking}

This subsection introduces the masking operation used in defining the channels in~\cref{subsec:mpsc,subsec:2mqc}, which converts a non-uniform binary ensemble $\cG_\alpha \coloneqq \{(1-\alpha,\tau_0),(\alpha,\tau_1)\}$ into a binary-input output-symmetric cq channel. Similar to the previous subsection, we adopt the shorthand $\bar{\tau}_\alpha \coloneqq (1-\alpha)\tau_0+\alpha \tau_1$ for their average state.

The idea is as follows. When transmitting an input bit $b \in \bits$, the channel draws a hidden mask $m \sim \Ber(\alpha)$ and sends the receiver the pair $(b \oplus m, \tau_m)$. In this situation, the quantum output $\tau_m$ serves as a ``quantum hint'' for the mask $m$. Moreover, decoding $b$ is equivalent to decoding $m$. Thus, decoding becomes equivalent to the task of distinguishing between $\tau_0$ and $\tau_1$ under the non-uniform prior $(1-\alpha,\alpha)$.

We now proceed with its formal definition.

\begin{definition}[Masking]
\label{def:masking}
Consider density matrices $\tau_0,\tau_1 \in \Dens(\C^d)$ and parameter $\alpha \in [0,1]$. Their \textbf{$\alpha$-masking} is the binary output pair
\begin{align*}
\Sigma_0 &\coloneqq (1-\alpha)\tau_0 \oplus \alpha\tau_1, \\
\Sigma_1 &\coloneqq \alpha\tau_1 \oplus (1-\alpha)\tau_0.
\end{align*}
Throughout, we denote this operation by $\operatorname{Mask}_\alpha(\tau_0,\tau_1) \coloneqq (\Sigma_0,\Sigma_1)$.
\end{definition}

Next, we state two propositions for deriving the rate upper bounds of the masked PSC (mPSC) in \cref{subsec:mpsc} and the masked mixed-qubit channel (2MQC) in \cref{subsec:2mqc}. The first shows that the PGM bit error rate of the masked pair equals that of the binary $\{(1-\alpha,\tau_0),(\alpha,\tau_1)\}$ ensemble. Intuitively, at the bitwise level, the PGM of the masked pair reads the block label $b \oplus m$, applies the PGM of this non-uniform binary ensemble to the ``quantum hint'' $\tau_m$ to produce a mask estimate $\widehat{m}$, and returns $\widehat{b} \coloneqq (b \oplus m) \oplus \widehat{m}$. It therefore errs exactly when $\widehat{m} \neq m$, which is why the two error rates coincide. Its proof appears in \cref{appsubsec:proof-masking}.

\begin{prop}[PGM bit error rate of a masked pair]
\label{prop:masking-pgm-error}
Let $(\Sigma_0,\Sigma_1) \coloneqq \operatorname{Mask}_\alpha(\tau_0,\tau_1)$. Then we have the PGM error identity
\begin{equation}
\label{eq:masking-pgm-error}
p_{\mathrm{e}}(\Sigma_0,\Sigma_1)
= p_{\mathrm{e}}\left(\cG_\alpha\right) = 2\alpha(1-\alpha)\Tr\left(\tau_0\bar{\tau}_\alpha^{-1/2}\tau_1\bar{\tau}_\alpha^{-1/2}\right).
\end{equation}
\end{prop}

Next, we record the Holevo information of a masked pair, which readily follows from considering the description given in this subsection's second paragraph and the observation that $\chi(\Sigma_0,\Sigma_1) = I(b;(b \oplus m,\tau_m)) = I(b;b \oplus m) + I(b;\tau_m \mid b \oplus m)$. Its proof appears in \cref{appsubsec:proof-masking}.

\begin{prop}[Holevo information of a masked pair]
\label{prop:masking-holevo}
Let $(\Sigma_0, \Sigma_1) \coloneqq \operatorname{Mask}_\alpha(\tau_0,\tau_1)$. Then we have the Holevo information identity
\begin{equation}
\label{eq:masking-holevo}
\chi(\Sigma_0,\Sigma_1) = 1 - h(\alpha) + S(\bar{\tau}_\alpha) - (1-\alpha)S(\tau_0)-\alpha S(\tau_1).
\end{equation}
\end{prop}

We close with the equal-weight specialization that will relate the MQC and 2MQC rate bounds.

\begin{remark}[Masking with a uniform prior]
\label{rem:masking-uniform-prior}
Suppose $(\tau_0,\tau_1)$ is output-symmetric, and let $U$ be a unitary that exchanges its two states. If $(\Sigma_0,\Sigma_1) \coloneqq \operatorname{Mask}_{\frac{1}{2}}(\tau_0,\tau_1)$ and $W \coloneqq I \oplus U$, then
\begin{equation*}
W\Sigma_bW^\dagger = \frac{1}{2}\tau_b \oplus \frac{1}{2}\tau_b \cong \tau_b \otimes \frac{I}{2},
\qquad b \in \bits.
\end{equation*}
Thus, a uniform-prior mask recovers the original channel together with an input-independent uniform block label. In particular, it preserves both the PGM bit error rate and the uniform-input Holevo information, so optimizing over a masked channel family cannot give a weaker rate upper bound than optimizing over the original family.
\end{remark}

\section{Proof of the main theorem}
\label{sec:main-theorem-proof}

Throughout this section, $\sigma_0$ and $\sigma_1$ denote the outputs of a finite-dimensional binary-input output-symmetric cq channel $\Phi_\sigma$ (\cref{def:binary-cq-channel}), and we write $\bar{\sigma} \coloneqq (\sigma_0+\sigma_1)/2$ for their average state. We will also fix a binary code $\cC \subseteq \bits^n$ throughout this section. For ease of exposition, we will conflate $\cC$ with the uniform ensemble $\{(\abs{\cC}^{-1},\sigma_c) : c \in \cC\}$.

We prove the pretty good criterion (\cref{thm:main-result}) in two stages. First, in \cref{thm:expected-hamming}, we upper bound the expected Hamming distance between the transmitted codeword $\bc$ and the PGM outcome $\bchat$ by the channel's PGM bit error rate $p_{\mathrm{e}}(\sigma_0,\sigma_1)$, from which then the constraint $p_{\mathrm{e}}(\sigma_0,\sigma_1) < \delta$ converts that expectation bound into a constant block decoding success (\cref{cor:block-error}).\footnote{Note that the code's distance assumption is invoked only once in the argument, which is at~\cref{cor:block-error}.} By coupling this with the cq channel coding strong converse (\cref{thm:cq-strong-converse}) and setting parameters appropriately, we conclude \cref{thm:main-result}.

We begin with the expected Hamming distance upper bound.

\begin{theorem}[PGM expected Hamming distance]
\label{thm:expected-hamming}
Consider any binary-input output-symmetric cq channel $\Phi_\sigma$. Suppose $\bc$ is a uniform sample from $\cC$. Let $\sigma_{\bc} = \Phi_\sigma^{\otimes n}(\bc)$ be its output state, and let $\bchat$ be the outcome of the PGM of $\cC$ applied to $\sigma_{\bc}$. Then
\begin{equation*}
\Eb \left[d(\bchat,\bc)\right] \leq n p_{\mathrm{e}}(\sigma_0,\sigma_1).
\end{equation*}
\end{theorem}

\begin{proof}
First, observe that
\begin{equation*}
\Eb[d(\bchat,\bc)] = \Prb[\bchat_1 \neq \bc_1] + \cdots + \Prb[\bchat_n \neq \bc_n].
\end{equation*}
Thus, it suffices to show that
\begin{equation*}
\Prb[\bchat_i \neq \bc_i] \leq p_{\mathrm{e}}(\sigma_0,\sigma_1)
\end{equation*}
for each $i \in \{1, \ldots , n\}$. Indeed, let $\pi_i(b) \coloneqq \Prb[\bc_i = b]$ for $b \in \bits$, and define the binary ensemble $\cC_i \coloneqq \{(\pi_i(0), \Eb[\sigma_{\bc} \mid \bc_i = 0]), (\pi_i(1), \Eb[\sigma_{\bc} \mid \bc_i = 1])\}$. By the blockwise-to-bitwise PGM lemma (\cref{lem:block-pgm-bitwise}), we see that
\begin{equation*}
\Prb[\bchat_i \neq \bc_i] = p_{\mathrm{e}}(\cC_i).
\end{equation*}
For any $b \in \bits$, linearity of the $i$'th partial trace channel $\T_i$ and \eqref{eq:partial-trace-product} give
\begin{equation*}
\T_i\left(\Eb[\sigma_\bc \mid \bc_i = b]\right)
= \Eb[\T_i(\sigma_\bc) \mid \bc_i = b]
= \Eb[\sigma_{\bc_i} \mid \bc_i = b]
= \sigma_b.
\end{equation*}
Applying PGM data processing (\cref{lem:pgm-data-processing}) to this partial trace maps $\cC_i$ to $\{(\pi_i(0),\sigma_0),(\pi_i(1),\sigma_1)\}$. The uniform-prior bound (\cref{lem:uniform-prior-worst}) then gives
\begin{equation*}
p_{\mathrm{e}}(\cC_i) \leq p_{\mathrm{e},\pi_i(1)}(\sigma_0,\sigma_1) \leq p_{\mathrm{e}}(\sigma_0,\sigma_1).
\end{equation*}
Combining the bitwise bounds therefore concludes the theorem.
\end{proof}

By combining this Hamming distance upper bound with the code's minimum distance, we therefore obtain a constant chance of successfully decoding over $\Phi_\sigma^{\otimes n}$.

\begin{corollary}[Constant successful block decoding]
\label{cor:block-error}
Under the hypotheses of \cref{thm:expected-hamming}, suppose in addition that $d_{\min}(\cC) \geq \delta n$ for some $\delta > 0$. Then
\begin{equation*}
\Prb[\bchat = \bc] \geq 1-\frac{p_{\mathrm{e}}(\sigma_0,\sigma_1)}{\delta}.
\end{equation*}
In particular, if $p_{\mathrm{e}}(\sigma_0,\sigma_1) < \delta$, then $\Prb[\bchat = \bc] \geq \Omega(1)$.
\end{corollary}
\begin{proof}
On the event $\{\bchat \neq \bc\}$, the two random codewords are distinct, so $d(\bchat,\bc) \geq d_{\min}(\cC)$. Thus, $d_{\min}(\cC)\Prb[\bchat \neq \bc] \leq \Eb[d(\bchat,\bc)]$. Combining this inequality with~\cref{thm:expected-hamming} proves the claim.\qedhere
\end{proof}

Now, armed with~\cref{cor:block-error}, we can proceed to the proof of our main result.

\begin{proof}[Proof of \cref{thm:main-result}]
Set $\eta \coloneqq 1-p_{\mathrm{e}}(\sigma_0,\sigma_1)/\delta > 0$, and let $L_\sigma$ be the channel-dependent constant in \cref{thm:cq-strong-converse}. Observe that \cref{cor:block-error} gives us $\Prb[\bchat = \bc] \geq \eta$. Now, apply the cq channel coding strong converse (\cref{thm:cq-strong-converse}) with
\begin{equation*}
\eps \coloneqq L_\sigma\sqrt{\frac{\log(8/\eta)}{n\log{\mathrm{e}}}}.
\end{equation*}
This choice of $\eps$ is designed so that the second term on the right-hand side of \eqref{eq:fixed-channel-exponential-converse} is then $\eta/2$. This leads us to the inequality
\begin{equation*}
\eta \leq \Prb[\bchat = \bc] \leq 2^{-n[R(\cC)-\chi(\sigma_0,\sigma_1)-2\eps]} + \frac{\eta}{2}.
\end{equation*}
Taking logarithms and rearranging gives
\begin{align*}
R(\cC) &\leq \chi(\sigma_0,\sigma_1) + 2\eps + \frac{1}{n}\log{\frac{2}{\eta}} \\
&= \chi(\sigma_0,\sigma_1) + 2L_\sigma\sqrt{\frac{\log(8/\eta)}{n\log{\mathrm{e}}}} + \frac{1}{n}\log{\frac{2}{\eta}} \\
&= \chi(\sigma_0,\sigma_1) + O(n^{-1/2}).
\end{align*}
Maximizing over all choices of $\cC$ and letting $n \to \infty$, we conclude $R_2(\delta) \leq \chi(\sigma_0,\sigma_1)$.
\end{proof}

\section{Rederiving known rate bounds}
\label{sec:known-rate-bounds}

Recall from~\eqref{eq:rate-function} that $R_2(\delta) \coloneqq \limsup_{n \to \infty} n^{-1}\log{A_2(n,\ceil{\delta n})}$ is the asymptotic rate--distance function for binary codes. This section applies \cref{thm:main-result} to rederive, in order, the Plotkin~\cite{Plo60}, Elias--Bassalygo~\cite{Eli55,Bas65}, first MRRW, and second MRRW bounds~\cite{MRRW77}. To help bridge the transition from the classical to the quantum setup for the reader, we will present the calculations of the classical channels in~\cref{subsec:bec,subsec:bsc} in the quantum formalism. When deriving the Holevo information for each channel (excluding BEC), we will recast our qubits in terms of Pauli matrices $X$ and $Z$, defined at~\eqref{eq:pauli-matrices}.

The four subsections follow the same template. We first give the two output states $\sigma_0$ and $\sigma_1$, their average state $\bar{\sigma} \coloneqq (\sigma_0+\sigma_1)/2$, and the unitary $U$ that witnesses the output-symmetry of their corresponding binary-input cq channel $\Phi_\sigma$. We then compute the PGM bit error rate $p_{\mathrm{e}}(\sigma_0,\sigma_1)$ using the tools presented in~\cref{subsec:pgm-bit-error-formulas}, followed by the computation of the Holevo information $\chi(\sigma_0, \sigma_1)$ using~\cref{eq:holevo-formula} and the tools in~\cref{subsec:holevo-formulas}. We then apply the pretty good criterion (\cref{thm:main-result}) to conclude the rate--distance tradeoff obtained from $\Phi_\sigma$.

\subsection{Binary Erasure Channel (BEC)}
\label{subsec:bec}

We begin with the Binary Erasure Channel (BEC), whose diagonal output states rederive the Plotkin bound (\cref{cor:bec-rate}). Recall that BEC transmits the input bit $\ket{b}$ with probability $1-p$ and otherwise outputs the erasure symbol $\ket{?}$ that is orthogonal to the output bits $\{\ket{0}, \ket{1}\}$. On the span of $\ket{0},\ket{1},\ket{?}$, its output states are, for $b \in \bits$,
\begin{equation*}
\sigma_b^{\BEC} \coloneqq (1-p)\ketbra{b}{b}+p\ketbra{?}{?}.
\end{equation*}
The unitary $U^{\BEC} \coloneqq \ketbra{0}{1}+\ketbra{1}{0}+\ketbra{?}{?}$ exchanges the two output states, establishing output-symmetry. Moreover, the resulting average state is
\begin{equation*}
\bar{\sigma}^{\BEC} = \frac{1}{2}(1-p)(\ketbra{0}{0}+\ketbra{1}{1})+p\ketbra{?}{?}.
\end{equation*}

\subsubsection{PGM bit error rate}
\label{subsubsec:bec-pgm-error}

Next, we compute the PGM bit error rate, which we abbreviate as $p_{\mathrm{e}}^{\BEC}(p) \coloneqq p_{\mathrm{e}}(\sigma_0^{\BEC}(p),\sigma_1^{\BEC}(p))$. We specialize the binary PGM formula (\cref{prop:binary-pgm-error}) to the uniform prior to derive the PGM bit error rate. To compute the trace in~\cref{eq:binary-pgm-error}, we first multiply each output state by the inverse square root of the average state. This gives us
\begin{align*}
\sigma_0^{\BEC}(\bar{\sigma}^{\BEC})^{-1/2} &= \sqrt{2(1-p)}\ketbra{0}{0} + \sqrt{p}\ketbra{?}{?}, \\
\sigma_1^{\BEC}(\bar{\sigma}^{\BEC})^{-1/2} &= \sqrt{2(1-p)}\ketbra{1}{1} + \sqrt{p}\ketbra{?}{?}.
\end{align*}
Plugging these two quantities into \eqref{eq:binary-pgm-error}, we find that
\begin{equation*}
p_{\mathrm{e}}^{\BEC}(p) = \frac{1}{2}\Tr(\sigma_0^{\BEC}(\bar{\sigma}^{\BEC})^{-1/2}\sigma_1^{\BEC}(\bar{\sigma}^{\BEC})^{-1/2}) = \frac{1}{2}\Tr(p\ketbra{?}{?}) = \frac{p}{2}.
\end{equation*}
Thus, $p_{\mathrm{e}}^{\BEC}(p) = p/2$.

\begin{remark}
The formula $p_{\mathrm{e}}^{\BEC}(p) = p/2$ can, of course, be derived directly from classical considerations. We invoke~\cref{prop:binary-pgm-error} here primarily to illustrate its application and provide a useful perspective for the reader.
\end{remark}

\subsubsection{Holevo information}
\label{subsubsec:bec-holevo}

Next, we compute the Holevo information, which we abbreviate as $\chi_\BEC(p) \coloneqq \chi(\sigma_0^{\BEC},\sigma_1^{\BEC})$. Since each output state $\sigma_b^{\BEC}$ has spectrum $\{1-p,p,0\}$, we see that
\begin{equation*}
S(\sigma_0^{\BEC}) = S(\sigma_1^{\BEC}) = h(p).
\end{equation*}
In addition, the average state $\bar{\sigma}^{\BEC}$ has spectrum $\{(1-p)/2,(1-p)/2,p\}$. Its entropy is
\begin{equation*}
S(\bar{\sigma}^{\BEC}) = -p\log{p}-2\left(\frac{1-p}{2}\right)\log{\left(\frac{1-p}{2}\right)} = h(p)+1-p.
\end{equation*}
Applying the Holevo identity (\cref{eq:holevo-formula}) now gives
\begin{equation*}
\chi_{\BEC}(p) = S(\bar{\sigma}^{\BEC}) - \frac{1}{2}\left[S(\sigma_0^{\BEC})+S(\sigma_1^{\BEC})\right] = (h(p)+1-p) -\frac{1}{2}\left[h(p)+h(p)\right] = 1-p.
\end{equation*}
Thus, the Holevo information is exactly the probability that the input survives the erasure.

\subsubsection{Rate upper bound}
\label{subsubsec:bec-rate}

Since $p_{\mathrm{e}}^{\BEC}(p) = p/2$, the pretty good criterion (\cref{thm:main-result}) applies when $p < 2\delta$.
Letting $p \to (2\delta)^-$ therefore leads to the Plotkin bound.

\begin{corollary}[{Plotkin bound --- \cite[pp.~445--450]{Plo60}}]
\label{cor:bec-rate}
For every $\delta \in (0,\frac{1}{2})$, applying the pretty good criterion (\cref{thm:main-result}) to the BEC yields
\begin{equation*}
R_2(\delta) \leq 1-2\delta.
\end{equation*}
\end{corollary}

\subsection{Binary Symmetric Channel (BSC)}
\label{subsec:bsc}

We next consider the Binary Symmetric Channel (BSC), whose diagonal output states rederive the Elias--Bassalygo bound (\cref{cor:bsc-rate}). Recall that the BSC flips the input bit $b$ with probability $p$. In the Pauli notation introduced at~\eqref{eq:pauli-matrices}, we can cast the output states, for $b \in \bits$, as
\begin{equation}
 \label{eq:bsc-pauli-form}
\sigma_b^{\BSC} \coloneqq \frac{1}{2}\left(I+(-1)^b(1-2p)Z\right).
\end{equation}
Note that both states are diagonal in the computational basis, hence classical. Since $XZX = -Z$, the unitary $U^{\BSC} \coloneqq X$ exchanges the outputs, establishing output-symmetry. Moreover, the average state is
\begin{equation*}
\bar{\sigma}^{\BSC}
= \frac{1}{2}\left(\sigma_0^{\BSC}+\sigma_1^{\BSC}\right)
= \frac{I}{2}.
\end{equation*}
Note that this Pauli presentation will be reused for the mixed-qubit channel (MQC) in \cref{subsec:mqc}.

\subsubsection{PGM bit error rate}
\label{subsubsec:bsc-pgm-error}

Next, we compute the PGM bit error rate, which we abbreviate as $p_{\mathrm{e}}^{\BSC}(p) \coloneqq p_{\mathrm{e}}(\sigma_0^{\BSC},\sigma_1^{\BSC})$. Since $\bar{\sigma}^{\BSC} = I/2$, we see that $\left(\bar{\sigma}^{\BSC}\right)^{-1/2} = \sqrt{2} I$. First, note the simple identities $\Tr(Z) = 0$ and $Z^2 = I$; by~\eqref{eq:binary-pgm-error} from~\cref{prop:binary-pgm-error} and~\eqref{eq:bsc-pauli-form}, we find that
\begin{align*}
p_{\mathrm{e}}^{\BSC}(p)
&= \Tr\left(\sigma_0^{\BSC}\sigma_1^{\BSC}\right)
= \frac{1}{4}\Tr\left[\left(I+(1-2p)Z\right)\left(I-(1-2p)Z\right)\right]\\
&= \frac{1}{4}\Tr\left[\left(1-(1-2p)^2\right)I\right]
= \frac{1-(1-2p)^2}{2}
= 2p(1-p).
\end{align*}

\begin{remark}[BSC posterior sampling]
\label{rem:bsc-posterior}
A more elegant way to derive the formula $p_{\mathrm{e}}^{\BSC}(p) = 2p(1-p)$ is to observe that posterior sampling---the classical specialization of the PGM---for $\BSC(p)$ is again $\BSC(p)$. Since $\Prb[\be_1 \oplus \be_2 = 1] = 2p(1-p)$ for independent $\be_1,\be_2 \sim \Ber(p)$, the PGM bit error rate formula follows.
\end{remark}

\subsubsection{Holevo information}
\label{subsubsec:bsc-holevo}

Next, we compute the Holevo information, which we abbreviate as $\chi_\BSC(p) \coloneqq \chi(\sigma_0^{\BSC},\sigma_1^{\BSC})$. Each output state has spectrum $\{1-p,p\}$, so
\begin{equation*}
S(\sigma_0^{\BSC}) = S(\sigma_1^{\BSC}) = h(p).
\end{equation*}
As for the average state, the spectrum is $\{\frac{1}{2},\frac{1}{2}\}$, and hence
\begin{equation*}
S(\bar{\sigma}^{\BSC}) = 1.
\end{equation*}
Applying the Holevo identity (\cref{eq:holevo-formula}) now gives
\begin{align*}
\chi_{\BSC}(p)
&= S(\bar{\sigma}^{\BSC})-\frac{1}{2}\left[S(\sigma_0^{\BSC})+S(\sigma_1^{\BSC})\right]\\
&= 1-\frac{1}{2}\left[h(p)+h(p)\right] \\
&= 1-h(p).
\end{align*}
Thus, the Holevo information is the usual mutual information of the BSC.

\subsubsection{Rate upper bound}
\label{subsubsec:bsc-rate}

Since $p_{\mathrm{e}}^{\BSC}(p) = 2p(1-p)$, the pretty good criterion (\cref{thm:main-result}) applies when $2p(1-p) < \delta$. This condition is met when $p < J(\delta) \coloneqq (1-\sqrt{1-2\delta})/2$ (i.e., below the Johnson radius). Letting $p \to J(\delta)^-$ leads to the Elias--Bassalygo bound.

\begin{corollary}[{Elias--Bassalygo bound for binary codes --- \cite[Appendix A, Equations~(A.15)--(A.18)]{Eli55};\cite[Equations~(5)~and~(8)]{Bas65}}]
\label{cor:bsc-rate}
For every $\delta \in (0,\frac{1}{2})$, applying the pretty good criterion (\cref{thm:main-result}) to the BSC yields
\begin{equation*}
R_2(\delta) \leq R_{\mathrm{EB}}(\delta) \coloneqq 1-h\left(\frac{1-\sqrt{1-2\delta}}{2}\right).
\end{equation*}
\end{corollary}

\subsection{Pure-State Channel (PSC)}
\label{subsec:psc}

We next consider the Pure-State Channel (PSC), whose pure qubit outputs rederive the first MRRW bound (\cref{cor:psc-rate}). For a parameter $p \in [0,\frac{1}{2}]$, the $\PSC(p)$ maps a bit $b \in \bits$ to one of two coherent Bernoulli states. Here, we define $\ket{\Ber(p)} \coloneqq \sqrt{1-p}\ket{0}+\sqrt{p}\ket{1}$ to be the coherent Bernoulli state with parameter $p$. The output states of $\PSC(p)$ are, for $b \in \bits$,
\begin{equation*}
\sigma_b^{\PSC} \coloneqq X^b\ketbra{\Ber(p)}{\Ber(p)}X^b.
\end{equation*}
Since $X\ket{\Ber(p)} = \ket{\Ber(1-p)}$, the two outputs are $\ketbra{\Ber(p)}{\Ber(p)}$ and $\ketbra{\Ber(1-p)}{\Ber(1-p)}$, which is the promised amplitude swap.
In the Pauli notation fixed in \cref{subsec:holevo-formulas}, the output states can be recast as
\begin{equation}
\label{eq:psc-pauli-form}
\sigma_b^{\PSC} = \frac{1}{2}\left(I+2\sqrt{p(1-p)}X+(-1)^b(1-2p)Z\right).
\end{equation}
As a sanity check, a state $\varrho$ is pure if and only if $\Tr(\varrho^2) = 1$. Since $\Tr(X) = \Tr(Z) = 0$ and $\frac{1}{2}(1+(2\sqrt{p(1-p)})^2+(1-2p)^2) = 1$, this therefore confirms the purity of the output states. Since $XZX = -Z$, the unitary $U^{\PSC} \coloneqq X$ exchanges the outputs, establishing output-symmetry. Moreover, their average state is
\begin{equation}
\label{eq:psc-average-state}
\bar{\sigma}^{\PSC} = \frac{1}{2}\left(I+2\sqrt{p(1-p)}X\right).
\end{equation}

\subsubsection{PGM bit error rate}
\label{subsubsec:psc-pgm-error}

Next, we compute the PGM bit error rate, which we abbreviate as $p_{\mathrm{e}}^{\PSC}(p) \coloneqq p_{\mathrm{e}}(\sigma_0^{\PSC},\sigma_1^{\PSC})$. First, observe that the output states $\sigma_b^{\PSC}$ for $b \in \bits$ are pure states with fidelity
\begin{equation*}
\abs{\braket{\Ber(p)|\Ber(1-p)}}^2 = 4p(1-p).
\end{equation*}
Applying \cref{prop:binary-pgm-error}, specifically~\eqref{eq:non-uniform-pure-pgm-error}, with the uniform prior $\alpha=\frac{1}{2}$ gives us
\begin{equation}
p_{\mathrm{e}}^{\PSC}(p) = \frac{2p(1-p)}{1+\sqrt{1-4p(1-p)}} = p,
\end{equation}
where the last equality uses $\sqrt{1-4p(1-p)} = 1-2p$ for $p \in [0,\frac{1}{2}]$.

\begin{remark}[Computational-basis PGM for the PSC]
\label{rem:psc-pgm-optimal}
For $p \in [0,\frac{1}{2})$, the bit-level PGM for the PSC is the computational-basis measurement $\{\ketbra{0}{0}, \ketbra{1}{1}\}$. This measurement is additionally optimal (see \cref{fact:psc-bitwise-optimal}), so the value $p_{\mathrm{e}}^{\PSC}(p) = p$ computed above is the minimum possible bit error rate.
\end{remark}

\subsubsection{Holevo information}
\label{subsubsec:psc-holevo}

Next, we compute the Holevo information, which we abbreviate as $\chi_{\PSC}(p) \coloneqq \chi(\sigma_0^{\PSC},\sigma_1^{\PSC})$. Since the output states $\sigma_b^{\PSC}$ are pure, we immediately see that\footnote{This can also be seen by applying~\cref{prop:qubit-entropy} along with the Pauli form given in~\cref{eq:psc-pauli-form}.}
\begin{equation*}
S(\sigma_0^{\PSC}) = S(\sigma_1^{\PSC}) = 0.
\end{equation*}
As for the average state $\bar{\sigma}^{\PSC}$, applying \cref{prop:qubit-entropy} to the Pauli form in~\eqref{eq:psc-average-state} gives
\begin{equation*}
S(\bar{\sigma}^{\PSC})
= s\left(1-4p(1-p)\right)
= s\left((1-2p)^2\right)
= h\left(\frac{1}{2}-\sqrt{p(1-p)}\right).
\end{equation*}
Applying the Holevo identity (\cref{eq:holevo-formula}) now gives
\begin{align*}
\chi_{\PSC}(p)
&= S(\bar{\sigma}^{\PSC})-\frac{1}{2}\left[S(\sigma_0^{\PSC})+S(\sigma_1^{\PSC})\right] \\
&= h\left(\frac{1}{2}-\sqrt{p(1-p)}\right).
\end{align*}

\subsubsection{Rate upper bound}
\label{subsubsec:psc-rate}

Since $p_{\mathrm{e}}^{\PSC}(p) = p$, the pretty good criterion (\cref{thm:main-result}) applies whenever $p < \delta$. Applying \cref{thm:main-result}, we get $R_2(\delta) \leq \chi_{\PSC}(p) = h\left(\frac{1}{2}-\sqrt{p(1-p)}\right)$ for all $p < \delta$. Letting $p \to \delta^-$ leads to the first MRRW bound.

\begin{corollary}[{First MRRW bound for binary codes --- \cite[Equation~(1.5)]{MRRW77}}]
\label{cor:psc-rate}
For every $\delta \in (0,\frac{1}{2})$, applying the pretty good criterion (\cref{thm:main-result}) to the PSC gives
\begin{equation*}
R_2(\delta) \leq \MRRW(\delta)
= h\left(\frac{1}{2}-\sqrt{\delta(1-\delta)}\right).
\end{equation*}
\end{corollary}

\subsection{Masked PSC (mPSC)}
\label{subsec:mpsc}

We conclude this section with the masked PSC, whose four-dimensional outputs rederive the second MRRW bound (\cref{cor:mpsc-rate}). Fix $r,\alpha \in [0,\frac{1}{2}]$. Recalling the coherent Bernoulli state $\ket{\Ber(r)} = \sqrt{1-r}\ket{0}+\sqrt{r}\ket{1}$ from \cref{subsec:psc}, write $\tau_m(r) \coloneqq X^m\ketbra{\Ber(r)}{\Ber(r)}X^m$ for $m \in \bits$, and let $\bar{\tau}_\alpha(r) \coloneqq (1-\alpha)\tau_0(r)+\alpha\tau_1(r)$ be their average state. Following the $\alpha$-masking definition (\cref{def:masking}), the two output states are
\begin{align*}
\sigma_0^{\mPSC}(r,\alpha) &\coloneqq (1-\alpha)\tau_0(r) \oplus \alpha\tau_1(r), \\
\sigma_1^{\mPSC}(r,\alpha) &\coloneqq \alpha\tau_1(r) \oplus (1-\alpha)\tau_0(r).
\end{align*}
The block-swap unitary $U^{\mPSC}(\xi \oplus \zeta) \coloneqq \zeta \oplus \xi$ exchanges the outputs, establishing output-symmetry. Moreover, their average state is
\begin{equation*}
\bar{\sigma}^{\mPSC}
=\frac{1}{2}\left(\sigma_0^{\mPSC}+\sigma_1^{\mPSC}\right)
=\frac{1}{2}\left(\bar{\tau}_\alpha(r) \oplus \bar{\tau}_\alpha(r)\right).
\end{equation*}

\begin{remark}[mPSC includes both PSC and BSC]
\label{rem:mpsc-endpoints}
Since the second MRRW bound was initially conceived as interpolating between the first MRRW bound and the Elias--Bassalygo bound, one might wonder if that can be seen directly from the masked PSC construction. Indeed, when $\alpha = \frac{1}{2}$, since masking with a uniform prior restores the original channel (\cref{rem:masking-uniform-prior}), we see that $\mPSC(r,\frac{1}{2})$ is equivalent to $\PSC(r)$ and therefore recovers the first MRRW bound (\cref{cor:psc-rate}). When $r = \frac{1}{2}$, one has $\tau_m(\frac{1}{2}) = \tau_{m \oplus 1}(\frac{1}{2})$, so the only information about the mask $m$ is its prior $\Ber(\alpha)$. Thus, $\mPSC(\tfrac{1}{2},\alpha)$ is equivalent to $\BSC(\alpha)$ and recovers the Elias--Bassalygo bound (\cref{cor:bsc-rate}).
\end{remark}

\subsubsection{PGM bit error rate}
\label{subsubsec:mpsc-pgm-error}

Next, we compute the PGM bit error rate, which we abbreviate in this subsection as $p_{\mathrm{e}}^{\mPSC}(r,\alpha) \coloneqq p_{\mathrm{e}}(\sigma_0^{\mPSC}(r,\alpha),\sigma_1^{\mPSC}(r,\alpha))$. By the PGM error masking identity (\cref{prop:masking-pgm-error}), the mPSC error is the PGM error of the non-uniform pure-state ensemble $\{(1-\alpha,\tau_0(r)),(\alpha,\tau_1(r))\}$. Its squared overlap is
\begin{equation*}
\abs{\braket{\Ber(r)|\Ber(1-r)}}^2 = 4r(1-r),
\end{equation*}
and because $r \in [0,\frac{1}{2}]$, we have that $\sqrt{1-4r(1-r)} = 1-2r$. Substituting these expressions into the non-uniform pure-state formula \eqref{eq:non-uniform-pure-pgm-error}, as recorded in \cref{prop:binary-pgm-error}, gives
\begin{equation}
\label{eq:mpsc-pgm-error}
p_{\mathrm{e}}^{\mPSC}(r,\alpha)
\coloneqq p_{\mathrm{e}}\left(\sigma_0^{\mPSC},\sigma_1^{\mPSC}\right)
= \frac{8\alpha(1-\alpha)r(1-r)}
{1+2(1-2r)\sqrt{\alpha(1-\alpha)}}.
\end{equation}

\subsubsection{Holevo information}
\label{subsubsec:mpsc-holevo}

Next, we compute the Holevo information, which we abbreviate as $\chi_{\mPSC}(r,\alpha) \coloneqq \chi(\sigma_0^{\mPSC},\sigma_1^{\mPSC})$. Both states in the underlying PSC pair are pure, so
\begin{equation*}
S(\tau_0(r)) = S(\tau_1(r)) = 0.
\end{equation*}
Their average state has the Pauli form
\begin{equation*}
\bar{\tau}_\alpha(r)
=\frac{1}{2}\left(I+2\sqrt{r(1-r)}X
+(1-2\alpha)(1-2r)Z\right).
\end{equation*}
By \cref{prop:qubit-entropy},
\begin{align*}
S(\bar{\tau}_\alpha(r))
&= s\left(1-4r(1-r)-(1-2\alpha)^2(1-2r)^2\right) \\
&= s\left(4\alpha(1-\alpha)(1-2r)^2\right).
\end{align*}
The masked Holevo identity (\cref{prop:masking-holevo}) therefore gives
\begin{align}
\chi_{\mPSC}(r,\alpha)
&\coloneqq \chi\left(\sigma_0^{\mPSC},\sigma_1^{\mPSC}\right) \notag \\
&=1-h(\alpha)
+s\left(4\alpha(1-\alpha)(1-2r)^2\right) \notag \\
&=1+s\left(4\alpha(1-\alpha)(1-2r)^2\right)
-s\left(4\alpha(1-\alpha)\right).
\label{eq:mpsc-holevo}
\end{align}
Here the last step uses $h(\alpha) = s\left(4\alpha(1-\alpha)\right)$.

\subsubsection{Rate upper bound}
\label{subsubsec:mpsc-rate}

To see how the resulting rate bound from mPSC attains exactly the second MRRW in its original form~\cite[Equation~(1.4)]{MRRW77},\footnote{Note that the function $s$ in this work coincides with the function $g$ in the original work~\cite{MRRW77}. In this work, the function $s$ plays a similar role to von Neumann entropy $S$ as the binary entropy function $h$ does to Shannon entropy $H$.} we introduce the shorthands
\begin{equation*}
u \coloneqq 2(1-2r)\sqrt{\alpha(1-\alpha)}, \qquad p \coloneqq \frac{8\alpha(1-\alpha)r(1-r)}{1+2(1-2r)\sqrt{\alpha(1-\alpha)}}.
\end{equation*}
Then \eqref{eq:mpsc-holevo} can be rewritten as 
\begin{equation*}
\chi_{\mPSC}(r,\alpha) = 1+s(u^2)-s\left(u^2+2pu+2p\right),
\end{equation*}
which is close in form to the bound presented in~\cite[Equation~(1.4)]{MRRW77}. Moreover, by~\eqref{eq:mpsc-pgm-error}, we see that in fact
\begin{equation}
\label{eq:mpsc-same-p}
p_{\mathrm{e}}^{\mPSC}(r,\alpha) = p.
\end{equation}
Now, via our shorthands, we can solve for $\alpha$ and $r$ in terms of $u$ and $p$. This gives us, picking the solutions in the range $[0,1/2]$,
\begin{equation*}
(r,\alpha) = \left(\frac{1}{2}-\frac{u}{2\sqrt{u^2+2p(1+u)}}, \frac{1 - \sqrt{(1+u)(1-2p-u)}}{2} \right),
\end{equation*}
assuming $u \in [0,1-2p]$ (and $p \geq 0$). Now, by~\eqref{eq:mpsc-same-p}, the pretty good criterion (\cref{thm:main-result}) applies whenever $p < \delta$. Therefore, letting $p \to \delta^-$ and minimizing over $u$ yields the second MRRW bound.
\begin{corollary}[{Second MRRW bound for binary codes --- \cite[Equation~(1.4)]{MRRW77}}]
\label{cor:mpsc-rate}
For every $\delta \in (0,\frac{1}{2})$, applying the pretty good criterion (\cref{thm:main-result}) to the masked PSC gives
\begin{equation}
\label{eq:second-mrrw}
R_2(\delta) \leq \MRRW^{(2)}(\delta)
\coloneqq
\min_{0 \leq u \leq 1-2\delta}
\left[1+s(u^2)-s\left(u^2+2\delta u+2\delta\right)\right].
\end{equation}
\end{corollary}

The endpoints of this family recover the BSC and PSC bounds already encountered.

\begin{remark}[Qubit cq channels]
\label{rem:second-mrrw-qubit-channels}
Naturally, one might wonder if the second MRRW bound could instead be captured by a cq channel with qubit outputs as opposed to a $4$-dimensional one. As it turns out, by~\cref{rem:mqc-has-all-qubits}, the mixed-qubit channel (MQC) is the binary-input output-symmetric cq channel with qubit outputs that yields the best rate bound among all such cq channels. Since it does not outperform the \emph{second} MRRW bound throughout the entire interval, any cq channel corresponding to the second MRRW bound necessarily has to go to higher dimensional qudits.
\end{remark}

\section{New classical--quantum (cq) channel constructions}
\label{sec:new-cq-channels}

We now turn to channel designs introduced in this work. The first one is the mixed-qubit channel (MQC), which we will later show has a resulting rate upper bound that lies strictly below the first MRRW bound (\cref{thm:first-mrrw-improvement}). The second construction, the masked MQC (2MQC), follows a similar approach to the derivation of the second MRRW bound but with MQC in place of PSC. As a result, its ultimate rate upper bound will lie strictly below the second MRRW bound (\cref{thm:second-mrrw-improvement}).

\subsection{Mixed-qubit channel (MQC)}
\label{subsec:mqc}

The mixed-qubit channel (MQC) is defined to be the channel that is obtained by passing the output of $\PSC(r)$ through the X-Pauli channel $\mathcal{B}_\eta$ from \eqref{eq:pauli-x-channel}. Fix $r,\eta \in [0,\frac{1}{2}]$. Since $X\sigma_b^{\PSC}(r)X = \sigma_{b \oplus 1}^{\PSC}(r)$, following the description given in~\cref{subsec:how-mixing}, its output states are, for $b \in \bits$,
\begin{equation}
\label{eq:mqc-pauli-form}
\sigma_b^{\MQC}(r,\eta)
\coloneqq \mathcal{B}_\eta\left(\sigma_b^{\PSC}(r)\right)
=\frac{1}{2}\left(I+2\sqrt{r(1-r)}X+(-1)^b(1-2r)(1-2\eta)Z\right).
\end{equation}
Put differently, each output is the mixture $(1-\eta)\sigma_b^{\PSC}(r)+\eta\sigma_{b \oplus 1}^{\PSC}(r)$. The unitary $U^{\PSC} \coloneqq X$ exchanges the outputs, establishing output-symmetry. Moreover, their average state is
\begin{equation}
\label{eq:mqc-average-state}
\bar{\sigma}^{\MQC}(r,\eta)
=\frac{1}{2}I+\sqrt{r(1-r)}X.
\end{equation}

\begin{remark}[MQC includes both PSC and BSC]
\label{rem:mqc-endpoints}
Almost by construction, the BSC and PSC are special cases of the MQC. Indeed, at $\eta = 0$, the $X$-Pauli noise is no longer there, and so $\MQC(r,0)$ is exactly $\PSC(r)$. Moreover, at $r = 0$, the underlying PSC outputs are the orthogonal computational-basis states, so $\MQC(0,\eta)$ is exactly $\BSC(\eta)$ (cf.~\cref{fn:bsc-psc-composition}).
\end{remark}

\begin{remark}[The MQC encompasses all output-symmetric qubit states]
\label{rem:mqc-has-all-qubits}
One can in fact observe that, up to symmetry, the MQC encompasses all the possible output-symmetric qubit states. Indeed, write the Bloch vectors of the two output states as $v_0,v_1 \in \mathbb{R}^3$. From this Bloch sphere lens, we see that output-symmetry is equivalent to $\norm{v_0}_2 = \norm{v_1}_2$. This is moreover equivalent to the orthogonality between their midpoint $m \coloneqq (v_0+v_1)/2$ and half-difference $d \coloneqq (v_0-v_1)/2$. Since every rotation in $\text{SO}(3)$ on the Bloch sphere is realized by a unitary $U \in \text{SU}(2)$ applied to the qubit ($\varrho \mapsto U\varrho U^\dagger$), by orthogonally mapping $m \mapsto (\norm{m}_2,0,0)$ and $d \mapsto (0,0,\norm{d}_2)$, we can therefore rotate the two output states to $v_b=(x,0,(-1)^bz)$ with $x,z \geq 0$ and $x^2+z^2 \leq 1$. Set $r \coloneqq (1-\sqrt{1-x^2})/2$ and $\eta \coloneqq (1-z/\sqrt{1-x^2})/2$. Then $2\sqrt{r(1-r)}=x$ and $(1-2r)(1-2\eta)=z$, which is exactly the MQC pair in~\eqref{eq:mqc-pauli-form}.
\end{remark}

\subsubsection{PGM bit error rate}
\label{subsubsec:mqc-pgm-error}

Next, we compute the PGM bit error rate, which we abbreviate in this subsection as $p_{\mathrm{e}}^{\MQC}(r,\eta) \coloneqq p_{\mathrm{e}}(\sigma_0^{\MQC}(r,\eta),\sigma_1^{\MQC}(r,\eta))$. From the definition of MQC in~\eqref{eq:mqc-pauli-form}, observe that
\begin{equation*}
    \Phi^{\MQC} = \mathcal{B}_\eta \circ \Phi^{\PSC} = \Phi^{\PSC} \circ \Phi^{\BSC}.
\end{equation*}
Indeed, the commutativity follows from observing that both $\mathcal{B}_\eta$ and $\Phi^{\PSC}$ are just applications of the Pauli $X$ (cf.~\cref{fn:bsc-psc-composition}). Thus, applying~\eqref{eq:pgm-preprocessing-error-uniform-prior} from~\cref{prop:binary-pgm-preprocessing} gives us that
\begin{equation}
\label{eq:mqc-pgm-error-formula}
p_{\mathrm{e}}^{\MQC}(r,\eta) = \frac{1}{2}\left[1-(1-2r)(1-2\eta)^2\right].
\end{equation}
Equivalently, the MQC PGM first applies the PSC PGM and then adds an independent posterior $\Ber(\eta)$ flip (\cref{rem:bsc-posterior}).\footnote{This is perhaps best observed via the Petz contravariant functoriality (\cref{fact:petz-properties}\ref{item:petz-functoriality}), concretely packaged in \cref{prop:binary-pgm-preprocessing}.} Combined with the prior $\Ber(\eta)$ flip, it gives the PGM bit error rate above.

\subsubsection{Holevo information}
\label{subsubsec:mqc-holevo}

Next, we compute the Holevo information, which we abbreviate in this subsection as $\chi_{\MQC}(r,\eta) \coloneqq \chi(\sigma_0^{\MQC}(r,\eta),\sigma_1^{\MQC}(r,\eta))$. Applying \cref{prop:qubit-entropy} to the output-state Pauli forms in~\eqref{eq:mqc-pauli-form} gives
\begin{align}
S\left(\sigma_b^{\MQC}(r,\eta)\right) &= s\left(1-4r(1-r)-(1-2r)^2(1-2\eta)^2\right) \notag \\
&= s\left(4\eta(1-\eta)(1-2r)^2\right) \label{eq:mqc-output-state-holevo}
\end{align}
Similarly, applying \cref{prop:qubit-entropy} to the average-state Pauli form in~\eqref{eq:mqc-average-state} gives
\begin{equation*}
S\left(\bar{\sigma}^{\MQC}(r,\eta)\right) = s\left(1-4r(1-r)\right) = s\left((1-2r)^2\right).
\end{equation*}
Plugging the two values into the Holevo identity (\cref{eq:holevo-formula}) therefore gives
\begin{equation}
\label{eq:mqc-holevo-formula}
\chi_{\MQC}(r,\eta) = s\left((1-2r)^2\right)-s\left(4\eta(1-\eta)(1-2r)^2\right).
\end{equation}

\begin{remark}[Alternative derivation via mutual information]
\label{rem:mqc-mutual-information}
Let us give an alternative and more conceptual derivation of~\cref{eq:mqc-holevo-formula}. Let $\bb \sim \Ber(\frac{1}{2})$ and $\be \sim \Ber(\eta)$ be independent random bits. Set $\by \coloneqq \bb \oplus \be \sim \Ber(\frac{1}{2})$. Since $\bb \longmapsto \bb \oplus \be \longmapsto \sigma_{\bb \oplus \be}^{\PSC}$ is a Markov chain, using the mutual information chain rule, the identity $\sigma_{\bb \oplus \be}^{\PSC} = X^\bb\sigma_{\be}^{\PSC}X^\bb$, the observation in~\cref{fn:bsc-psc-composition}, and the Holevo identity (\cref{eq:holevo-formula}), notice that
\begin{align*}
I(\by;\sigma_\by^{\PSC}) = I(\bb,\be;\sigma_{\bb \oplus \be}^{\PSC})
&= I(\bb;\sigma_{\bb \oplus \be}^{\PSC}) + I(\be;\sigma_{\bb \oplus \be}^{\PSC} \mid \bb) \\
&= \chi_{\MQC}(r,\eta) + I(\be;\sigma_\be^{\PSC}).
\end{align*}
Thus, we find that
\begin{equation*}
\chi_{\MQC}(r,\eta) = I(\by;\sigma_{\by}^{\PSC}) - I(\be;\sigma_\be^{\PSC}).
\end{equation*}
Because the bit $\by$ remains uniform for every $\eta$, the first term is always $I(\by;\sigma_{\by}^{\PSC}) = s((1-2r)^2)$. Thus, the dependence on $\eta$ lies in the second term, which equals $I(\be;\sigma_\be^{\PSC}) = s(4\eta(1-\eta)(1-2r)^2)$. Thus, mixing starts with the information carried by the underlying PSC and subtracts exactly the information that its output state $\sigma_\by$ retains about the discarded flip $\be$ once $\bb$ is known.
\end{remark}

\subsubsection{Rate upper bound}
\label{subsubsec:mqc-rate}

Based on~\cref{eq:mqc-pgm-error-formula}, define the function
\begin{equation}
\label{eq:mqc-error-function}
e(r,\eta) \coloneqq \frac{1}{2}\left[1-(1-2r)(1-2\eta)^2\right].
\end{equation}
Then the pretty good criterion (\cref{thm:main-result}) applies whenever $e(r,\eta) < \delta$. Now, fix any value $r \in [0,\delta]$. Notice that increasing $\eta$ along the interval $[0,\frac{1}{2}]$ decreases the Holevo information, which only sharpens the resulting rate upper bound. On the other hand, it also increases the PGM bit error rate, which is $e(r,\eta)$. Since $\delta$ is the supremum value of $e(r,\eta)$ under which the pretty good criterion (\cref{thm:main-result}) still applies, then setting $\eta = \eta_\delta(r)$, where $\eta_\delta(r)$ is the unique solution to the equation $e(r,\eta_\delta(r)) = \delta$ that lies in the interval $[0,\frac{1}{2}]$, would achieve the best possible rate upper bound for MQC. Now, solving for $\eta_\delta(r)$ in the equation $e(r,\eta_\delta(r)) = \delta$, we find that 
\begin{equation*}
\eta_\delta(r) = \frac{1}{2} - \frac{1}{2} \cdot \sqrt{\frac{1-2\delta}{1-2r}}.
\end{equation*}
Thus, by~\cref{eq:mqc-holevo-formula}, for any fixed $r \in [0,\delta]$, the lowest Holevo information for which the pretty good criterion (\cref{thm:main-result}) still applies at this specified value $r$ is
\begin{align*}
\chi_{\MQC}(r,\eta_\delta(r)) &= s\left((1-2r)^2\right)-s\left(4\eta_\delta(r)(1-\eta_\delta(r))(1-2r)^2\right) \\
&= s\left((1-2r)^2\right)-s\left(2(\delta-r)(1-2r)\right).
\end{align*}
Minimizing over all values $r \in [0,\delta]$ gives us the following rate bound.
\begin{corollary}[MQC rate upper bound]
\label{cor:mqc-rate}
For every $\delta \in (0,\frac{1}{2})$, applying the pretty good criterion (\cref{thm:main-result}) to the MQC gives us the rate bound
\begin{equation}
\label{eq:mixed-rate-bound}
\begin{aligned}
R_2(\delta)
\leq R_{\MQC}(\delta)
&\coloneqq
\min_{0 \leq r \leq \delta}
\left[s\left((1-2r)^2\right)-s\left(2(\delta-r)(1-2r)\right)\right] \\
&=
\min_{1-2\delta \leq u \leq 1}
\left[s(u^2)-s\left(u(u-(1-2\delta))\right)\right],
\end{aligned}
\end{equation}
where the second line is the change of variables $u \coloneqq 1-2r$.
\end{corollary}

This is the rate upper bound stated in \cref{thm:first-mrrw-improvement}. Its strict improvement over the first MRRW bound is proved later in~\cref{prop:mqc-strict-improvement}.

\begin{remark}[Holevo information and PGM error data processing]
\label{rem:mqc-data-processing}
As was noted in~\cref{subsec:how-mixing}, as $\eta$ increases along the interval $[0,\frac{1}{2}]$, we expect the Holevo information to decay while the PGM error to grow. Indeed, by inspecting~\cref{eq:mqc-holevo-formula,eq:mqc-pgm-error-formula}, one can numerically confirm that this is indeed the case here.
\end{remark}

\subsection{Masked mixed-qubit channel (2MQC)}
\label{subsec:2mqc}

We next apply the masking operation of \cref{subsec:masking} to the MQC pair. Fix $r,\eta,\alpha \in [0,\frac{1}{2}]$. Throughout, we will recurrently use the convenient shorthand 
\begin{equation}
\label{eq:theta-shorthand}
\theta \coloneqq \alpha+\eta-2\alpha\eta.
\end{equation}
For $m \in \bits$, write $\tau_m(r,\eta) \coloneqq \sigma_m^{\MQC}(r,\eta)$, and let $\bar{\tau}_\alpha(r,\eta) \coloneqq (1-\alpha)\tau_0(r,\eta)+\alpha\tau_1(r,\eta)$. Here, $r$ is the noise parameter of the underlying PSC, $\eta$ is the probability of the discarded inner $X$-flip, and $\alpha$ is the prior probability of the mask. The 2MQC output states are
\begin{align*}
\sigma_0^{\mMQC}(r,\eta,\alpha) &\coloneqq (1-\alpha)\tau_0(r,\eta) \oplus \alpha\tau_1(r,\eta), \\
\sigma_1^{\mMQC}(r,\eta,\alpha) &\coloneqq \alpha\tau_1(r,\eta) \oplus (1-\alpha)\tau_0(r,\eta).
\end{align*}
The block-swap unitary $U^{\mMQC}(\varrho_0 \oplus \varrho_1) \coloneqq \varrho_1 \oplus \varrho_0$ exchanges the outputs, establishing output-symmetry. Moreover, their average state is
\begin{equation*}
\bar{\sigma}^{\mMQC}
=\frac{1}{2}\left(\sigma_0^{\mMQC}+\sigma_1^{\mMQC}\right)
=\frac{1}{2}\left(\bar{\tau}_\alpha(r,\eta) \oplus \bar{\tau}_\alpha(r,\eta)\right).
\end{equation*}

\begin{remark}[2MQC includes both mPSC and MQC]
\label{rem:2mqc-endpoints}
Almost by construction, the mPSC and MQC are both special cases of the 2MQC. Indeed, at $\eta = 0$, the $X$-Pauli noise is no longer there, and so $\mMQC(r,0,\alpha)$ is exactly $\mPSC(r,\alpha)$. Moreover, when $\alpha = \frac{1}{2}$, since masking with a uniform prior restores the original channel (\cref{rem:masking-uniform-prior}), we see that $\mMQC(r,\eta,\frac{1}{2})$ is equivalent to $\MQC(r,\eta)$.
\end{remark}

\subsubsection{PGM bit error rate}
\label{subsubsec:2mqc-pgm-error}

Next, we compute the PGM bit error rate, which we abbreviate in this subsection as $p_{\mathrm{e}}^{\mMQC}(r,\eta,\alpha) \coloneqq p_{\mathrm{e}}(\sigma_0^{\mMQC}(r,\eta,\alpha),\sigma_1^{\mMQC}(r,\eta,\alpha))$. Following the same initial reasoning as in~\cref{subsubsec:mqc-pgm-error}, from the definition of MQC in~\eqref{eq:mqc-pauli-form}, observe that
\begin{equation*}
\Phi^{\MQC} = \mathcal{B}_\eta \circ \Phi^{\PSC} = \Phi^{\PSC} \circ \Phi^{\BSC}.
\end{equation*}
Indeed, the commutativity follows from observing that both $\mathcal{B}_\eta$ and $\Phi^{\PSC}$ are just applications of the Pauli $X$ (cf.~\cref{fn:bsc-psc-composition}). Thus, by the PGM error masking identity (\cref{prop:masking-pgm-error}) and~\cref{eq:pgm-preprocessing-error} from~\cref{prop:binary-pgm-preprocessing}, we find that
\begin{equation*}
p_{\mathrm{e}}^{\mMQC}(r,\eta,\alpha) = p_{\mathrm{e},\alpha}(\tau_0(r,\eta),\tau_1(r,\eta)) 
= 2\alpha(1-\alpha) - \left[
\frac{\alpha(1-\alpha)(1-2\eta)}
{\theta(1-\theta)}
\right]^2 \cdot \left[
2\theta(1-\theta)-p_{\mathrm{e},\theta}(\sigma_0,\sigma_1)
\right].
\end{equation*}
On a more conceptual level, this is the PGM error of a bit $\bb \sim \Ber(\alpha)$ sent through the composition $\bb \mapsto \by \coloneqq \bb \oplus \be \mapsto \sigma_{\by}^{\PSC}(r)$, where $\be \sim \Ber(\eta)$ and hence $\by \sim \Ber(\theta)$. Now, by applying the non-uniform pure-state formula \eqref{eq:non-uniform-pure-pgm-error} to the intermediate PSC ensemble, we see that
\begin{equation*}
p_{\mathrm{e},\theta}\left(\sigma_0^{\PSC}(r),\sigma_1^{\PSC}(r)\right)
= \frac{8\theta(1-\theta)r(1-r)}{1+2(1-2r)\sqrt{\theta(1-\theta)}}.
\end{equation*}
Substituting this PGM error directly into the composition identity \eqref{eq:pgm-preprocessing-error} of \cref{prop:binary-pgm-preprocessing} gives
\begin{align}
\label{eq:2mqc-pgm-error-formula}
p_{\mathrm{e}}^{\mMQC}(r,\eta,\alpha)
&= 2\alpha(1-\alpha) - \left[
\frac{\alpha(1-\alpha)(1-2\eta)}
{\theta(1-\theta)}
\right]^2 \cdot \left[
2\theta(1-\theta)-\frac{8\theta(1-\theta)r(1-r)}{1+2(1-2r)\sqrt{\theta(1-\theta)}}
\right] \notag \\
&= 2\alpha(1-\alpha)\left[1 -
\frac{\alpha(1-\alpha)(1-2\eta)^2}
{\theta(1-\theta)} \cdot \left[
1-\frac{4r(1-r)}{1+2(1-2r)\sqrt{\theta(1-\theta)}}
\right]\right] \notag \\
&= 2\alpha(1-\alpha)\left[1 -
\frac{\alpha(1-\alpha)(1-2\eta)^2}
{\theta(1-\theta)}
+\frac{4\alpha(1-\alpha)r(1-r)(1-2\eta)^2}{\theta(1-\theta)\left[1+2(1-2r)\sqrt{\theta(1-\theta)}\right]}
\right] \notag \\
&= 2\alpha(1-\alpha)
\left[
\frac{\eta(1-\eta)}{\theta(1-\theta)}
+\frac{4\alpha(1-\alpha)r(1-r)(1-2\eta)^2}
{\theta(1-\theta)
\left[1+2(1-2r)\sqrt{\theta(1-\theta)}\right]}
\right].
\end{align}
Here, we used the identity $4\theta(1-\theta) = 1-(2\alpha-1)^2(1-2\eta)^2$ at the final step, which holds by definition of $\theta$ in~\eqref{eq:theta-shorthand}. Note that when $\alpha = \frac{1}{2}$, we have $\theta = \frac{1}{2}$, and so this recovers the formula given in~\cref{eq:mqc-pgm-error-formula}, which is consistent with~\cref{rem:masking-uniform-prior}.

\subsubsection{Holevo information}
\label{subsubsec:2mqc-holevo}

Next, we compute the Holevo information, which we abbreviate in this subsection as $\chi_{\mMQC}(r,\eta,\alpha) \coloneqq \chi(\sigma_0^{\mMQC}(r,\eta,\alpha),\sigma_1^{\mMQC}(r,\eta,\alpha))$. Recall from \cref{subsubsec:mqc-holevo}, specifically~\eqref{eq:mqc-output-state-holevo}, that
\begin{equation*}
S(\tau_0(r,\eta)) = S(\tau_1(r,\eta)) = s\left(4\eta(1-\eta)(1-2r)^2\right).
\end{equation*}
Now, the average state admits the Pauli form
\begin{equation*}
\bar{\tau}_\alpha(r,\eta)
=\frac{1}{2}\left(I+2\sqrt{r(1-r)}X
+(1-2\theta)(1-2r)Z\right),
\end{equation*}
so \cref{prop:qubit-entropy} gives us
\begin{align*}
S(\bar{\tau}_\alpha(r,\eta))
&= s\left(1-4r(1-r)-(1-2\theta)^2(1-2r)^2\right) \\
&= s\left(4\theta(1-\theta)(1-2r)^2\right).
\end{align*}
Thus, by the masked Holevo information identity (\cref{prop:masking-holevo}), we now find that
\begin{equation}
\label{eq:2mqc-holevo}
\chi_{\mMQC}(r,\eta,\alpha)
= 1-h(\alpha)
+s\left(4\theta(1-\theta)(1-2r)^2\right) - s\left(4\eta(1-\eta)(1-2r)^2\right).
\end{equation}

\subsubsection{Rate upper bound}
\label{subsubsec:2mqc-rate}

Based on~\cref{eq:2mqc-pgm-error-formula}, define the function
\begin{equation}
\label{eq:2mqc-error-function}
E(r,\eta,\alpha) \coloneqq 
2\alpha(1-\alpha)
\left[
\frac{\eta(1-\eta)}{\theta(1-\theta)}
+\frac{4\alpha(1-\alpha)r(1-r)(1-2\eta)^2}
{\theta(1-\theta)
\left[1+2(1-2r)\sqrt{\theta(1-\theta)}\right]}
\right].
\end{equation}
Building on the shorthands for the mPSC in \cref{subsubsec:mpsc-rate} and the MQC in \cref{subsubsec:mqc-rate}, define
\begin{align}
u &\coloneqq 2(1-2r)\sqrt{\theta(1-\theta)}, \notag \\
v &\coloneqq 2(1-2r)\sqrt{\eta(1-\eta)}, \notag \\
A &\coloneqq 4\alpha(1-\alpha).
\label{eq:2mqc-entropy-coordinates}
\end{align}
Here $u^2$ and $v^2$ are exactly the two arguments of $s$ in \eqref{eq:2mqc-holevo}, while $h(\alpha)=s(A)$. Thus, \cref{eq:2mqc-pgm-error-formula,eq:2mqc-holevo} become
\begin{equation*}
E(u,v,A) = \frac{A}{2}
\left[
\frac{v^2}{u^2}
+\frac{4A r(1-r)(1-2r)^2 (1-2\eta)^2}
{u^2(1+u)}
\right] 
= \frac{A\left[u^2(1-v^2)+v^2(u+v^2)\right]-(u^2-v^2)^2}{2u^2(1+u)},
\end{equation*}
and
\begin{equation}
\label{eq:2mqc-reduced-formulas}
\chi_{\mMQC}(u,v,A) = 1+s(u^2)-s(A)-s(v^2).
\end{equation}
The pretty good criterion (\cref{thm:main-result}) applies whenever $E(u,v,A)<\delta$. Fix a pair $(u,v)$ with $0<u\leq 1$ and $0\leq v<u$. At this fixed pair, the displayed formula for $E(u,v,A)$ is increasing and affine in $A$, whereas \eqref{eq:2mqc-reduced-formulas} decreases with $A$. As in the MQC calculation (\cref{subsubsec:mqc-rate}), the lowest admissible Holevo information is therefore obtained at $A = \lambda_\delta(u,v)$, where $\lambda_\delta(u,v)$ is the unique solution to the equation $E(u,v,\lambda_\delta(u,v)) = \delta$. Solving for $\lambda_\delta(u,v)$ in this equation, we find that 
\begin{equation}
\label{eq:2mqc-exact-error-coordinates}
\lambda_\delta(u,v) = 
\frac{(u^2-v^2)^2+2\delta u^2(1+u)}
{u^2(1-v^2)+v^2(u+v^2)}.
\end{equation}
Since we must have $A \in [0,1]$, we therefore see that the optimal value for $A$ must be
\begin{equation*}
A_\delta(u,v) \coloneqq \min\left\{1,\frac{(u^2-v^2)^2+2\delta u^2(1+u)}
{u^2(1-v^2)+v^2(u+v^2)}\right\}.
\end{equation*}
It remains to identify the admissible values of the triple $(u,v,A)$. We claim that a triple is admissible if and only if it satisfies the inequalities
\begin{equation*}
0\leq v \leq u \leq 1
\qquad
\text{ and }
\qquad
\frac{u^2-v^2}{1-v^2}\leq A\leq 1.
\end{equation*}
Indeed, both inequalities can be deduced by observing the identity $u^2-v^2=A(1-2r)^2(1-2\eta)^2$. As for the converse, it follows by setting $r,\eta,\alpha$ equal to
\begin{align*}
\alpha &= \frac{1-\sqrt{1-A}}{2}, \\
(1-2r)^2 &= v^2+\frac{u^2-v^2}{A}, \\
(1-2\eta)^2 &= 1-\frac{v^2}{(1-2r)^2}.
\end{align*}
Finally, we inspect the values of $(u,v)$ for which there exists an $A$ satisfying $E(u,v,A)<\delta$. Now, we know that $A=(u^2-v^2)/(1-v^2)$ is the smallest possible value of $A$. In such a case, the PGM bit error rate equals
\begin{equation*}
E\left(u,v,\frac{u^2-v^2}{1-v^2}\right)
=\frac{v^2(u^2-v^2)}{2u^2(1-v^2)}.
\end{equation*}
Because $E(u,v,A)$ increases with $A$ (it is affine in $A$), the range over which there exists an $A$ satisfying $E(u,v,A)<\delta$ is specified precisely by the inequality
\begin{equation*}
v^2(u^2-v^2)<2\delta u^2(1-v^2).
\end{equation*}
Now, we take stock of everything and conclude our rate bound. By applying the pretty good criterion (\cref{thm:main-result}) and a simple continuity argument, we find that
\begin{equation*}
R_2(\delta) \leq 1+s(u^2)-s\left(A_\delta(u,v)\right)-s(v^2).
\end{equation*}
Taking the infimum over all feasible pairs $(u,v)$ proves the following corollary.

\begin{corollary}[2MQC rate upper bound]
\label{cor:2mqc-rate}
For every $\delta \in (0,\frac{1}{2})$, applying the pretty good criterion (\cref{thm:main-result}) to the 2MQC gives
\begin{equation}
\label{eq:2mqc-rate}
R_2(\delta) \leq R_{\mMQC}(\delta)
\coloneqq
\inf_{\substack{0\leq v \leq u \leq 1, \\
v^2(u^2-v^2)<2\delta u^2(1-v^2)}}
\left[
1+s(u^2)-s\left(A_\delta(u,v)\right)-s(v^2)
\right],
\end{equation}
where we define $A_\delta(u,v) \coloneqq \min\left\{1,\frac{(u^2-v^2)^2+2\delta u^2(1+u)}
{u^2(1-v^2)+v^2(u+v^2)}\right\}$.
\end{corollary}

\section{Comparing the rate bounds}
\label{sec:comparing-rate-bounds}

In \cref{sec:known-rate-bounds}, we recovered all four previously known rate upper bounds---Plotkin, Elias--Bassalygo, the first and second MRRW bounds. In \cref{sec:new-cq-channels}, we produced two new rate upper bounds from the newly introduced cq channels MQC and 2MQC, which are stated in~\cref{cor:mqc-rate,cor:2mqc-rate}. In this section, we analytically and numerically compare the rate bounds of MQC and 2MQC with the first and second MRRW bounds, respectively, and establish~\cref{thm:first-mrrw-improvement,thm:second-mrrw-improvement}. For the reader's convenience, we recap the parameters and formulas of all six channels in~\cref{tab:channel-summary} below.

\begin{table}[!htbp]
\centering
\caption{PGM bit error rates, Holevo information, and optimized rate bounds for the six channel families from \cref{sec:known-rate-bounds,sec:new-cq-channels}. Here, $h$ is the binary entropy function, $J(\delta) \coloneqq (1-\sqrt{1-2\delta})/2$ is the Johnson radius, while $s(D) \coloneqq h((1-\sqrt{1-D})/2)$ is the von Neumann entropy of a qubit with determinant $D/4$. The variables $u,v,A$ are functions of $r,\eta,\alpha$ and are specified at~\cref{subsubsec:mpsc-rate,subsubsec:mqc-rate,subsubsec:2mqc-rate}. Here, the functions $e$ and $E$ are defined at~\cref{eq:mqc-error-function,eq:2mqc-error-function}, respectively.}
\label{tab:channel-summary}
\footnotesize
\renewcommand{\arraystretch}{1.4}
\begin{tabularx}{\textwidth}{@{} > {\centering\arraybackslash}p{0.15\textwidth} > {\centering\arraybackslash}p{0.17\textwidth} > {\centering\arraybackslash}p{0.34\textwidth} > {\centering\arraybackslash}X@{}}
\toprule
\textbf{Channel}
& \textbf{PGM bit error rate}
& \textbf{Uniform-input $\chi$}
& \textbf{Upper bound on $R_2(\delta)$} \\
\midrule
$\BEC(p)$
& $p/2$
& $1-p$
& $1-2\delta$ \\
$\BSC(p)$
& $2p(1-p)$
& $1-h(p)$
& $1-h(J(\delta))$ \\
$\PSC(p)$
& $p$
& $h(\frac{1}{2}-\sqrt{p(1-p)})$
& $\MRRW(\delta)$ \\
$\mPSC(r,\alpha)$
& $p$
& $1+s(u^2)-s(u^2+2pu+2p)$
& $\MRRW^{(2)}(\delta)$ \\
$\MQC(r,\eta)$
& $e(r,\eta)$
& $s\left(u^2\right)-s\left(4\eta(1-\eta)u^2\right)$
& $R_{\MQC}(\delta)$ \\
$\mMQC(r,\eta,\alpha)$
& $E(u,v,A)$
& $1+s(u^2)-s(A)-s(v^2)$
& $R_{\mMQC}(\delta)$ \\
\bottomrule
\end{tabularx}
\end{table}
\FloatBarrier

\subsection{MQC improvement over the first MRRW bound}
\label{subsec:mqc-first-mrrw-improvement}

In this subsection, we prove that slightly deviating away from PSC by considering a small $X$-Pauli noise in the MQC will make the MQC rate bound in~\cref{cor:mqc-rate} strictly smaller than the first MRRW bound. This is stated in~\cref{prop:mqc-strict-improvement}. For ease of exposition, we adopt the shorthand
\begin{equation}
\label{eq:mqc-holevo-function}
F_\delta(u) \coloneqq s(u^2)-s\left(u(u-(1-2\delta))\right),
\end{equation}
for the second objective in \eqref{eq:mixed-rate-bound}, where $u$ ranges over $[1-2\delta,1]$. Note that by~\cref{rem:mqc-endpoints}, we know that
\begin{equation*}
R_{\MQC}(\delta) \leq \min\{\MRRW(\delta),R_{\mathrm{EB}}(\delta)\}.
\end{equation*}
Now, we proceed to show that this is indeed strict for the first MRRW bound.\footnote{One can similarly show that for the Elias--Bassalygo bound, the MQC bound strictly improves upon it, but only when $\delta \gtrsim 0.07441947$.} For the reader's convenience, \cref{fig:mixed-versus-mrrw} provides a visual illustration of the argument used in the proof of the next proposition.

\begin{prop}[Strict improvement over the first MRRW expression]
\label{prop:mqc-strict-improvement}
For every $\delta \in (0,\frac{1}{2})$, we have the strict inequality
\begin{equation*}
R_{\MQC}(\delta) < \MRRW(\delta).
\end{equation*}
\end{prop}
\begin{proof}
By~\cref{rem:mqc-endpoints} and~\cref{eq:mqc-holevo-function,eq:mixed-rate-bound}, notice that $F_\delta(1-2\delta) = \MRRW(\delta)$. Now, fix $\delta$ and choose a perturbation $\zeta \to 0^+$ of $1-2\delta$. Write the difference
\begin{equation*}
F_\delta(1-2\delta+\zeta)-F_\delta(1-2\delta)
= \left[s((1-2\delta+\zeta)^2)-s((1-2\delta)^2)\right]
-s\left((1-2\delta+\zeta)\zeta\right).
\end{equation*}
Now, for $D \to 0^+$, the definition of $s$ given at~\eqref{eq:qubit-entropy} gives
\begin{equation*}
s(D)=\frac{D}{4}\log{\frac{1}{D}}+O(D).
\end{equation*}
Since $(1-2\delta)^2\in(0,1)$, the mean value theorem thus gives us
\begin{equation*}
s((1-2\delta+\zeta)^2)-s((1-2\delta)^2)=O(\zeta).
\end{equation*}
For the second term, which is where the primary gain in improved rate bounds stems from, its argument tends to zero, so the preceding expansion gives
\begin{align*}
s((1-2\delta+\zeta)\zeta)
&=\frac{(1-2\delta+\zeta)\zeta}{4}
  \log{\frac{1}{(1-2\delta+\zeta)\zeta}}+O(\zeta)\\
&=\frac{1-2\delta}{4}\zeta\log{\frac{1}{\zeta}}+O(\zeta).
\end{align*}
Therefore,
\begin{equation*}
\frac{F_\delta(1-2\delta+\zeta)-F_\delta(1-2\delta)}{\zeta}
=-\frac{1-2\delta}{4}\log{\frac{1}{\zeta}}+O(1),
\end{equation*}
which tends to $-\infty$ as $\zeta \to 0^+$, as $1-2\delta>0$. Hence $F_\delta(1-2\delta+\zeta)< F_\delta(1-2\delta)$ for all sufficiently small $\zeta>0$. Thus, we conclude that
\begin{equation*}
R_{\MQC}(\delta)
\leq F_\delta(1-2\delta+\zeta)
<F_\delta(1-2\delta)
=\MRRW(\delta).\qedhere
\end{equation*}
\end{proof}

Now, we have established everything we need to prove~\cref{thm:first-mrrw-improvement}, which we now show.

\begin{proof}[Proof of \cref{thm:first-mrrw-improvement}]
The MQC rate upper bound in \cref{cor:mqc-rate} gives $R_2(\delta) \leq R_{\MQC}(\delta)$,
while \cref{prop:mqc-strict-improvement} gives $R_{\MQC}(\delta) < \MRRW(\delta)$ for every $\delta \in (0,\frac{1}{2})$. Combining the two inequalities proves the theorem.\qedhere
\end{proof}

\paragraph{Comparing the curves numerically.}
\cref{fig:mixed-versus-mrrw} plots the rate bounds resulting from BEC (\cref{cor:bec-rate}), BSC (\cref{cor:bsc-rate}), PSC (\cref{cor:psc-rate}), mPSC (\cref{cor:mpsc-rate}), and the new rate bound from MQC (\cref{cor:mqc-rate}). We omit the 2MQC rate bound (\cref{cor:2mqc-rate}) from the plot, since the MQC bound alone already outperforms all previously known rate bounds for $\delta > 0.1951$. The inset provides a magnified view of the region around $\delta = \frac{1}{4}$, making the improvement more readily visible.

\begin{figure}[!htbp]
\centering
\begin{tikzpicture}
\begin{axis}[
    name=mainaxis,
    declare function={
        hb(\x)=-(max(\x,0.000000000001)*ln(max(\x,0.000000000001))
        +max(1-\x,0.000000000001)*ln(max(1-\x,0.000000000001)))/ln(2);
    },
    width=0.90\textwidth,
    height=0.42\textwidth,
    xmin=0,
    xmax=0.5,
    ymin=0,
    ymax=1.02,
    xtick={0,0.1,0.2,0.3,0.4,0.5},
    ytick={0,0.2,0.4,0.6,0.8,1},
    xlabel={Relative distance $\delta$},
    ylabel={Rate upper bound (bits)},
    grid=major,
    legend style={
        at={(0.985,0.985)},
        anchor=north east,
        draw=black!35,
        fill=white,
        fill opacity=0.94,
        text opacity=1,
        font=\scriptsize,
        legend columns=1
    },
    legend cell align={left},
    tick align=outside,
    clip=false
]
\addplot[thick,densely dotted,black,domain=0:0.5,samples=201]
    {1-2*x};
\addlegendentry{Plotkin}
\addplot[thick,dashdotted,orange!85!black,domain=0:0.5,samples=201]
    {1-hb((1-sqrt(max(0,1-2*x)))/2)};
\addlegendentry{Elias--Bassalygo}
\addplot[thick,dashed,magenta,domain=0:0.5,samples=201]
    {hb(0.5-sqrt(x*(1-x)))};
\addlegendentry{First MRRW}
\addplot[thick,teal!75!black,dash pattern=on 5pt off 2pt on 1pt off 2pt] coordinates {
    (0.000,1.0000000000)
    (0.005,0.9746882819)
    (0.010,0.9542335551)
    (0.015,0.9355485985)
    (0.020,0.9179845740)
    (0.025,0.9012357922)
    (0.030,0.8851235078)
    (0.035,0.8695304056)
    (0.040,0.8543736319)
    (0.045,0.8395916081)
    (0.050,0.8251368080)
    (0.055,0.8109714699)
    (0.060,0.7970648883)
    (0.065,0.7833916203)
    (0.070,0.7699302507)
    (0.075,0.7566625140)
    (0.080,0.7435726530)
    (0.085,0.7306469400)
    (0.090,0.7178733119)
    (0.095,0.7052410871)
    (0.100,0.6927407431)
    (0.105,0.6803637388)
    (0.110,0.6681023712)
    (0.115,0.6559496586)
    (0.120,0.6438992441)
    (0.125,0.6319453157)
    (0.130,0.6200825391)
    (0.135,0.6083060012)
    (0.140,0.5966111620)
    (0.145,0.5849938134)
    (0.150,0.5734500437)
    (0.155,0.5619762074)
    (0.160,0.5505688982)
    (0.165,0.5392249260)
    (0.170,0.5279412966)
    (0.175,0.5167151937)
    (0.180,0.5055439630)
    (0.185,0.4944250982)
    (0.190,0.4833562286)
    (0.195,0.4723351078)
    (0.200,0.4613596038)
    (0.205,0.4504276898)
    (0.210,0.4395374364)
    (0.215,0.4286870042)
    (0.220,0.4178746369)
    (0.225,0.4070986557)
    (0.230,0.3963574536)
    (0.235,0.3856494904)
    (0.240,0.3749732884)
    (0.245,0.3643274277)
    (0.250,0.3537105431)
    (0.255,0.3431213199)
    (0.260,0.3325584911)
    (0.265,0.3220208344)
    (0.270,0.3115071689)
    (0.275,0.3010234624)
    (0.280,0.2906344902)
    (0.285,0.2803578450)
    (0.290,0.2701957917)
    (0.295,0.2601506728)
    (0.300,0.2502249116)
    (0.305,0.2404210173)
    (0.310,0.2307415890)
    (0.315,0.2211893209)
    (0.320,0.2117670079)
    (0.325,0.2024775511)
    (0.330,0.1933239643)
    (0.335,0.1843093807)
    (0.340,0.1754370607)
    (0.345,0.1667103995)
    (0.350,0.1581329366)
    (0.355,0.1497083650)
    (0.360,0.1414405425)
    (0.365,0.1333335033)
    (0.370,0.1253914711)
    (0.375,0.1176188738)
    (0.380,0.1100203598)
    (0.385,0.1026008163)
    (0.390,0.0953653900)
    (0.395,0.0883195101)
    (0.400,0.0814689150)
    (0.405,0.0748196827)
    (0.410,0.0683782653)
    (0.415,0.0621515301)
    (0.420,0.0561468064)
    (0.425,0.0503719417)
    (0.430,0.0448353682)
    (0.435,0.0395461828)
    (0.440,0.0345142456)
    (0.445,0.0297503003)
    (0.450,0.0252661277)
    (0.455,0.0210747419)
    (0.460,0.0171906483)
    (0.465,0.0136301938)
    (0.470,0.0104120583)
    (0.475,0.0075579751)
    (0.480,0.0050938549)
    (0.485,0.0030516857)
    (0.490,0.0014731664)
    (0.495,0.0004182693)
    (0.500,0.0000000000)
};
\addlegendentry{Second MRRW}
\addplot[very thick,blue] coordinates {
    (0.000,1.0000000000)
    (0.005,0.9747338723)
    (0.010,0.9543925510)
    (0.015,0.9358765649)
    (0.020,0.9185310850)
    (0.025,0.9020465016)
    (0.030,0.8862412140)
    (0.035,0.8709956996)
    (0.040,0.8562253307)
    (0.045,0.8418670634)
    (0.050,0.8278721372)
    (0.055,0.8142017337)
    (0.060,0.8008242322)
    (0.065,0.7877133892)
    (0.070,0.7748470833)
    (0.075,0.7622052066)
    (0.080,0.7496660740)
    (0.085,0.7371581226)
    (0.090,0.7246824453)
    (0.095,0.7122401567)
    (0.100,0.6998323937)
    (0.105,0.6874603162)
    (0.110,0.6751251072)
    (0.115,0.6628279736)
    (0.120,0.6505701470)
    (0.125,0.6383528838)
    (0.130,0.6261774661)
    (0.135,0.6140452023)
    (0.140,0.6019574278)
    (0.145,0.5899155056)
    (0.150,0.5779208269)
    (0.155,0.5659748124)
    (0.160,0.5540789122)
    (0.165,0.5422346075)
    (0.170,0.5304434107)
    (0.175,0.5187068668)
    (0.180,0.5070265542)
    (0.185,0.4954040856)
    (0.190,0.4838411090)
    (0.195,0.4723393088)
    (0.200,0.4609004073)
    (0.205,0.4495261650)
    (0.210,0.4382183828)
    (0.215,0.4269789028)
    (0.220,0.4158096096)
    (0.225,0.4047124322)
    (0.230,0.3936893451)
    (0.235,0.3827423702)
    (0.240,0.3718735784)
    (0.245,0.3610850914)
    (0.250,0.3503790838)
    (0.255,0.3397577850)
    (0.260,0.3292234814)
    (0.265,0.3187785187)
    (0.270,0.3084253046)
    (0.275,0.2981663112)
    (0.280,0.2880040779)
    (0.285,0.2779412144)
    (0.290,0.2679804041)
    (0.295,0.2581244072)
    (0.300,0.2483760649)
    (0.305,0.2387383031)
    (0.310,0.2292141369)
    (0.315,0.2198066753)
    (0.320,0.2105191258)
    (0.325,0.2013548009)
    (0.330,0.1923171229)
    (0.335,0.1834096313)
    (0.340,0.1746359895)
    (0.345,0.1659999925)
    (0.350,0.1575055758)
    (0.355,0.1491568244)
    (0.360,0.1409579831)
    (0.365,0.1329134685)
    (0.370,0.1250278811)
    (0.375,0.1173060196)
    (0.380,0.1097528969)
    (0.385,0.1023737580)
    (0.390,0.0951740996)
    (0.395,0.0881596934)
    (0.400,0.0813366115)
    (0.405,0.0747112569)
    (0.410,0.0682903969)
    (0.415,0.0620812040)
    (0.420,0.0560913018)
    (0.425,0.0503288208)
    (0.430,0.0448024640)
    (0.435,0.0395215870)
    (0.440,0.0344962943)
    (0.445,0.0297375610)
    (0.450,0.0252573841)
    (0.455,0.0210689784)
    (0.460,0.0171870342)
    (0.465,0.0136280665)
    (0.470,0.0104109056)
    (0.475,0.0075574173)
    (0.480,0.0050936257)
    (0.485,0.0030516130)
    (0.490,0.0014731520)
    (0.495,0.0004182684)
    (0.500,0.0000000000)
};
\addlegendentry{MQC (new)}
\coordinate (zoomboxsw) at (axis cs:0.246,0.341);
\coordinate (zoomboxse) at (axis cs:0.254,0.341);
\draw[thin,black!60] (axis cs:0.246,0.341) rectangle (axis cs:0.254,0.365);
\end{axis}
\begin{axis}[
    name=insetaxis,
    at={(mainaxis.south west)},
    anchor=south west,
    xshift=1.30cm,
    yshift=0.70cm,
    declare function={
        hb(\x)=-(max(\x,0.000000000001)*ln(max(\x,0.000000000001))
        +max(1-\x,0.000000000001)*ln(max(1-\x,0.000000000001)))/ln(2);
    },
    width=0.30\textwidth,
    height=0.16\textwidth,
    xmin=0.246,
    xmax=0.254,
    ymin=0.341,
    ymax=0.365,
    xtick={0.248,0.250,0.252},
    xticklabels={$0.248$,$0.250$,$0.252$},
    ytick={0.345,0.350,0.355,0.360},
    yticklabels={$0.345$,$0.350$,$0.355$,$0.360$},
    scaled ticks=false,
    tick label style={font=\tiny},
    title={Near $\delta = \frac{1}{4}$},
    title style={font=\scriptsize,yshift=-1pt},
    grid=major,
    grid style={black!12},
    axis background/.style={fill=white},
    axis line style={black!55},
    tick align=outside,
    clip=true
]
\addplot[thick,dashed,magenta,domain=0.246:0.254,samples=101]
    {hb(0.5-sqrt(x*(1-x)))};
\addplot[thick,teal!75!black,dash pattern=on 5pt off 2pt on 1pt off 2pt] coordinates {
    (0.245,0.3643274277)
    (0.250,0.3537105431)
    (0.255,0.3431213199)
};
\addplot[very thick,blue] coordinates {
    (0.245,0.3610850914)
    (0.250,0.3503790838)
    (0.255,0.3397577850)
};
\addplot[thin,densely dotted,black!70] coordinates {
    (0.250,0.341)
    (0.250,0.365)
};
\end{axis}
\draw[thin,black!45] (zoomboxsw) -- (insetaxis.north east);
\draw[thin,black!45] (zoomboxse) -- (insetaxis.south east);
\end{tikzpicture}
\caption{Five benchmark rate curves, shown for $0 \leq \delta \leq \frac{1}{2}$. The Plotkin~\cite{Plo60}, Elias--Bassalygo~\cite{Eli55,Bas65}, and first MRRW~\cite{MRRW77} curves are closed form, while the second MRRW~\cite{MRRW77} and MQC curves are sampled one-dimensional minima. The two-variable 2MQC infimum from \eqref{eq:2mqc-rate} is not plotted; analytically, it lies strictly below the second MRRW throughout $0 < \delta < \frac{1}{2}$ by \cref{prop:2mqc-strict-mrrw2}. The inset magnifies the marked window around $\delta = \frac{1}{4}$, where the sampled MQC value, approximately $0.3503791$, lies below the sampled second MRRW bound $0.3537105$ and the closed-form first MRRW bound $0.3545789$. The MQC improvement over the first MRRW curve is analytic throughout the open interval (\cref{prop:mqc-strict-improvement}); its relation to the second MRRW curve at this point is numerical. All displayed curves use the common endpoint conventions rate $1$ at $\delta = 0$ and rate $0$ at $\delta = \frac{1}{2}$.}
\label{fig:mixed-versus-mrrw}
\end{figure}
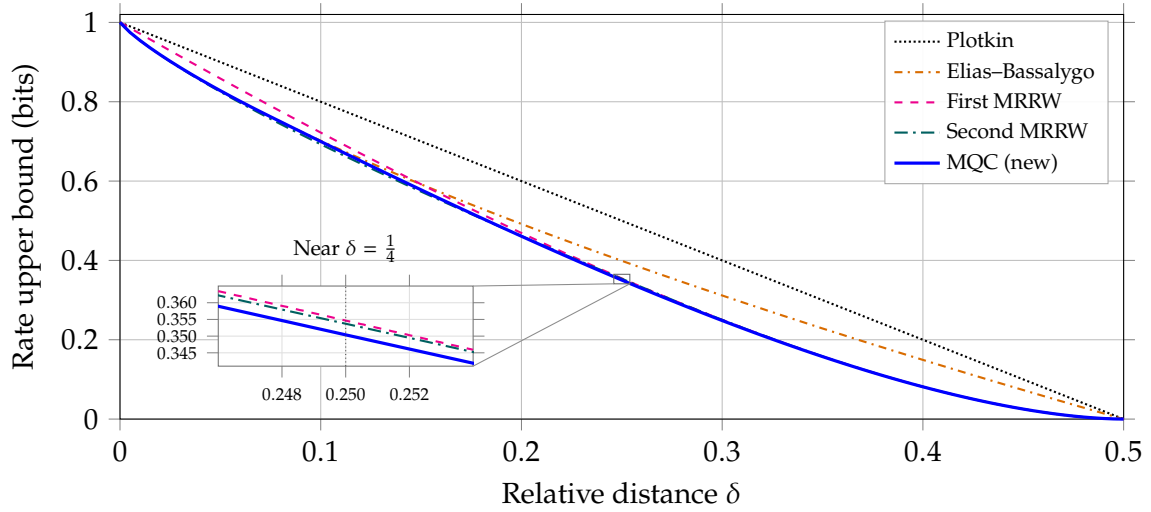

In~\cref{fig:mixed-optimization}, we give numerical substantiation of the proof of~\cref{prop:mqc-strict-improvement} for $\delta = 1/4$. In particular, \cref{fig:mixed-optimization} shows a numerical evaluation of the full function $u \mapsto F_{1/4}(u)$ on its admissible interval. At the point $u = 1-2\delta = 1/2$, the function $F_{1/4}(u)$ equals the first MRRW bound, as predicted by~\cref{rem:mqc-endpoints}. The curve initially falls, reaches the marked minimum, and then rises.

\begin{figure}[!htbp]
\centering
\begin{tikzpicture}
\begin{axis}[
    declare function={
        hb(\x)=-(max(\x,0.000000000001)*ln(max(\x,0.000000000001))
        +max(1-\x,0.000000000001)*ln(max(1-\x,0.000000000001)))/ln(2);
        ss(\y)=hb((1-sqrt(max(0,1-\y)))/2);
    },
    width=0.84\textwidth,
    height=0.34\textwidth,
    xmin=0.5,
    xmax=1,
    ymin=0.348,
    ymax=0.402,
    xtick={0.5,0.6,0.7,0.8,0.9,1},
    ytick={0.35,0.36,0.37,0.38,0.39,0.40},
    xlabel={$u$},
    ylabel={Holevo objective (bits)},
    grid=major,
    legend style={at={(0.98,0.04)},anchor=south east,draw=none,fill=white},
    tick align=outside,
    clip=false
]
\addplot[very thick,blue,domain=0.500001:1,samples=180]
    {ss(x^2)-ss(x*(x-0.5))};
\addlegendentry{$F_{1/4}(u)$}
\addplot[thick,dashed,magenta,domain=0.5:1,samples=2]
    {0.3545789027};
\addlegendentry{$\MRRW(1/4)$}
\addplot[only marks,black,mark=square*,mark size=2.8pt]
    coordinates {(0.5,0.3545789027)};
\addplot[only marks,red,mark=star,mark size=4pt]
    coordinates {(0.5293948668,0.3503790838)};
\draw[thin,black] (axis cs:0.5,0.3545789027)
    -- (axis cs:0.512,0.363)
    -- (axis cs:0.522,0.363);
\node[anchor=west,fill=white,inner sep=1pt]
    at (axis cs:0.522,0.363)
    {\footnotesize PSC endpoint};
\draw[thin,red] (axis cs:0.5293948668,0.3503790838)
    -- (axis cs:0.553,0.3498)
    -- (axis cs:0.592,0.3498);
\node[anchor=south west,text=black,fill=white,inner sep=1pt]
    at (axis cs:0.592,0.3490)
    {\footnotesize $u_\ast \approx 0.5294$};
\end{axis}
\end{tikzpicture}
\caption{A plot of the MQC Holevo information function $F_{1/4}(u) = s(u^2)-s(u(u-\frac{1}{2}))$ from~\cref{eq:mqc-holevo-function} over the range $u \in [\frac{1}{2},1]$, where $s(D) \coloneqq h((1-\sqrt{1-D})/2)$ as in \cref{prop:qubit-entropy}. The black square marks the PSC endpoint when $u = 1-2\delta$ (cf.~\cref{rem:mqc-endpoints}), where the MQC Holevo information equals the first MRRW bound~\cite{MRRW77}, which is approximately $\MRRW(1/4) \approx 0.3545789$ in this case. The red star marks the numerically obtained minimizer $u_\ast \approx 0.5293949$ of $F_{1/4}$, whose value is $F_{1/4}(u_\ast) \approx 0.3503791$.}
\label{fig:mixed-optimization}
\end{figure}
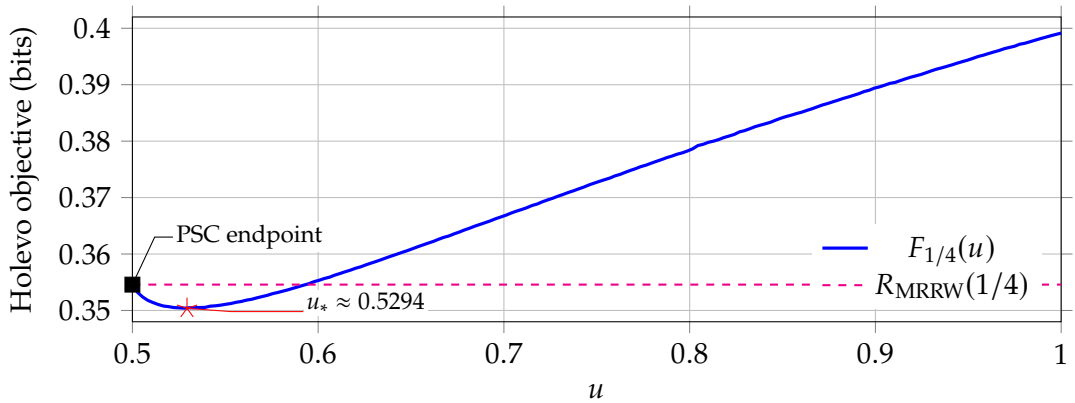

\FloatBarrier

\subsection{Low-rate regime}
\label{subsec:mqc-low-rate}

In this subsection, we examine how rapidly the MQC gain over the first MRRW bound decays in the low-rate regime. Define the gap between the two bounds as 
\begin{equation*}
\Delta(\delta) \coloneqq \MRRW(\delta)-R_{\MQC}(\delta).
\end{equation*}
Near the endpoint, this gap is difficult to resolve in \cref{fig:mixed-versus-mrrw}. We provide instead \cref{fig:mixed-gap}, which displays the gap on a logarithmic scale for greater clarity. The formal expansion below characterizes the rate at which the gap decays and identifies the leading order at which the two bounds differ.

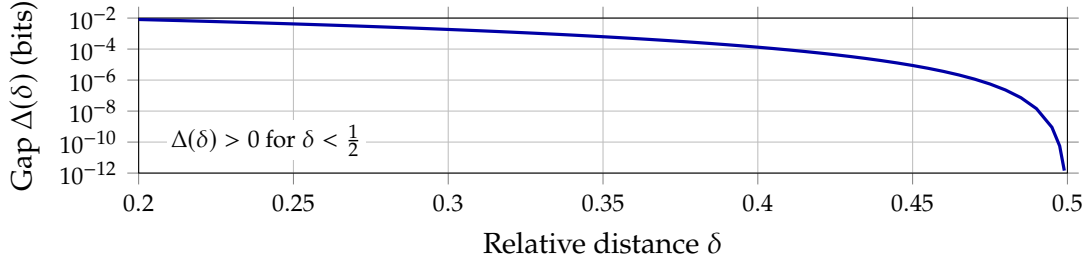
\begin{figure}[!htbp]
\centering
\begin{tikzpicture}
\begin{semilogyaxis}[
    width=0.84\textwidth,
    height=0.22\textwidth,
    xmin=0.2,
    xmax=0.5,
    ymin=1e-12,
    ymax=1e-2,
    xtick={0.2,0.25,0.3,0.35,0.4,0.45,0.5},
    ytick={1e-2,1e-4,1e-6,1e-8,1e-10,1e-12},
    xlabel={Relative distance $\delta$},
    ylabel={Gap $\Delta(\delta)$ (bits)},
    grid=major,
    tick align=outside,
    ticklabel style={font=\footnotesize}
]
\addplot[very thick,blue!65!black] coordinates {
    (0.200,8.095186317740e-3)
    (0.205,7.624514072615e-3)
    (0.210,7.172936037871e-3)
    (0.215,6.740057807526e-3)
    (0.220,6.325482148272e-3)
    (0.225,5.928809144341e-3)
    (0.230,5.549636349189e-3)
    (0.235,5.187558944221e-3)
    (0.240,4.842169904739e-3)
    (0.245,4.513060173322e-3)
    (0.250,4.199818840856e-3)
    (0.255,3.902033335414e-3)
    (0.260,3.619289619236e-3)
    (0.265,3.351172394012e-3)
    (0.270,3.097265314728e-3)
    (0.275,2.857151212299e-3)
    (0.280,2.630412325255e-3)
    (0.285,2.416630540736e-3)
    (0.290,2.215387645050e-3)
    (0.295,2.026265584080e-3)
    (0.300,1.848846733822e-3)
    (0.305,1.682714181324e-3)
    (0.310,1.527452016345e-3)
    (0.315,1.382645634021e-3)
    (0.320,1.247882048853e-3)
    (0.325,1.122750220332e-3)
    (0.330,1.006841390529e-3)
    (0.335,8.997494339711e-4)
    (0.340,8.010712201385e-4)
    (0.345,7.104069889229e-4)
    (0.350,6.273607393813e-4)
    (0.355,5.515406321251e-4)
    (0.360,4.825594056822e-4)
    (0.365,4.200348071628e-4)
    (0.370,3.635900375567e-4)
    (0.375,3.128542119714e-4)
    (0.380,2.674628351094e-4)
    (0.385,2.270582922580e-4)
    (0.390,1.912903560371e-4)
    (0.395,1.598167091094e-4)
    (0.400,1.323034830145e-4)
    (0.405,1.084258132212e-4)
    (0.410,8.786841042240e-5)
    (0.415,7.032614799914e-5)
    (0.420,5.550466546818e-5)
    (0.425,4.312098757688e-5)
    (0.430,3.290415854705e-5)
    (0.435,2.459589073580e-5)
    (0.440,1.795122673276e-5)
    (0.445,1.273921357349e-5)
    (0.450,8.743587349796e-6)
    (0.455,5.763465997809e-6)
    (0.460,3.614047414561e-6)
    (0.465,2.127309269037e-6)
    (0.470,1.152705881783e-6)
    (0.475,5.578562470845e-7)
    (0.480,2.292156093689e-7)
    (0.485,7.272071757313e-8)
    (0.490,1.439573506169e-8)
    (0.495,9.011030164200e-10)
    (0.4975,5.634472639700e-11)
    (0.499,1.442643903400e-12)
};
\node[anchor=south west,fill=white,inner sep=1.5pt] at (rel axis cs:0.03,0.08)
    {\footnotesize $\Delta(\delta) > 0$ for $\delta < \frac{1}{2}$};
\end{semilogyaxis}
\end{tikzpicture}
\caption{Numerically minimized values of $\Delta(\delta) = \MRRW(\delta)-R_{\MQC}(\delta)$ for $0.2 \leq \delta \leq 0.499$, shown on a logarithmic vertical scale. The strict-improvement proposition (\cref{prop:mqc-strict-improvement}) proves $\Delta(\delta) > 0$ for every $\delta \in (0,\frac{1}{2})$. The formal calculation below suggests that $\Delta(\frac{1}{2} - \eps) = \eps^4\log{\mathrm{e}}+O(\eps^6\log(1/\eps))$.}
\label{fig:mixed-gap}
\end{figure}

\paragraph{Formal expansion.} We now heuristically examine the scale visible in \cref{fig:mixed-gap}. Put $\delta = \frac{1}{2}-\eps$ with $\eps \to 0^+$. The identity $\MRRW(\frac{1}{2}-\eps) = s(4\eps^2)$ reduces the MRRW expansion to the following small-$D$ expansion of $s$, where $D \to 0^+$:
\begin{equation*}
s(D) = \frac{D}{4}\left[\log{\frac{4}{D}}+\log{\mathrm{e}}\right]
+\frac{D^2}{16}\left[\log{\frac{4}{D}}-\frac{\log{\mathrm{e}}}{2}\right]
+O\left(D^3\log{\frac{1}{D}}\right).
\end{equation*}
Now, note that the admissible interval for $F_{\frac{1}{2}-\eps}$ starts at $u = 2\eps$. Numerical minimizers suggest a displacement of order $\eps^3$ from this endpoint, so we set $u = 2\eps+8\xi\eps^3$ for a fixed value $\xi > 0$. Substitution into the expansion of $s$ then gives
\begin{equation*}
F_{\frac{1}{2}-\eps}(2\eps+8\xi\eps^3)-s(4\eps^2)
= 4\xi\eps^4\left[\log(4\xi)-\log{\mathrm{e}}\right]
+O\left(\eps^6\log{\frac{1}{\eps}}\right).
\end{equation*}
Now, to minimize the leading $\eps^4$ coefficient, differentiate with respect to $\xi$. Because all logarithms are base $2$, $\frac{d}{d\xi}\log(4\xi) = (\log{\mathrm{e}})/\xi$. The $\log{\mathrm{e}}$ terms cancel, so the derivative of $4\xi[\log(4\xi)-\log{\mathrm{e}}]$ is $4\log(4\xi)$. Thus, the displayed coefficient is minimized at $\xi = 1/4$, corresponding to the candidate value $u = 2\eps+2\eps^3$. Plugging this value for both the first MRRW bound and the MQC rate and expanding gives us
\begin{align*}
\MRRW\left(\frac{1}{2}-\eps\right)
&= \eps^2\left[2\log{\frac{1}{\eps}}+\log{\mathrm{e}}\right]
+\eps^4\left[2\log{\frac{1}{\eps}}-\frac{\log{\mathrm{e}}}{2}\right]
+O\left(\eps^6\log{\frac{1}{\eps}}\right), \\
F_{\frac{1}{2}-\eps}(2\eps+2\eps^3)
&= \eps^2\left[2\log{\frac{1}{\eps}}+\log{\mathrm{e}}\right]
+\eps^4\left[2\log{\frac{1}{\eps}}-\frac{3\log{\mathrm{e}}}{2}\right]
+O\left(\eps^6\log{\frac{1}{\eps}}\right).
\end{align*}
In particular, taking their difference thus gives us
\begin{equation*}
\Delta\left(\frac{1}{2} - \eps\right)
\geq
\MRRW\left(\frac{1}{2}-\eps\right)
-F_{\frac{1}{2}-\eps}(2\eps+2\eps^3)
= \eps^4\log{\mathrm{e}}
+O\left(\eps^6\log{\frac{1}{\eps}}\right).
\end{equation*}
First, observe that the positive coefficient is consistent with the strict-improvement proposition (\cref{prop:mqc-strict-improvement}). Moreover, the two bounds agree at the leading $\eps^2\log(1/\eps)$ scale, the $\eps^2$ term, and in the logarithmic fourth-order term. Their first difference occurs at the term $\eps^4\log{\mathrm{e}}$. We remark that this analysis is primarily based on what numerical computations seem to suggest.

\subsection{Comparing with the second MRRW bound}
\label{subsec:comparing-second-mrrw}

In this subsection, we prove the strict 2MQC improvement over the second MRRW bound (\cref{prop:2mqc-strict-mrrw2}) and complete the proof of \cref{thm:second-mrrw-improvement}. For ease of exposition, we adopt the shorthand
\begin{equation}
\label{eq:second-mrrw-holevo-function}
H_\delta(u) \coloneqq 1+s(u^2)-s(u^2+2\delta u+2\delta),
\end{equation}
for the objective in~\eqref{eq:second-mrrw}, where $u$ ranges over $[0,1-2\delta]$. By~\cref{rem:2mqc-endpoints}, the 2MQC contains the MQC and mPSC as special cases. Thus, we see that
\begin{equation*}
R_{\mMQC}(\delta) \leq \min\left\{R_{\MQC}(\delta),\MRRW^{(2)}(\delta)\right\}.
\end{equation*}
We now show that this inequality for the second MRRW bound is always strict.

\begin{prop}[Strict improvement over the second MRRW bound]
\label{prop:2mqc-strict-mrrw2}
For every $\delta \in (0,\frac{1}{2})$, we have the strict inequality
\begin{equation*}
R_{\mMQC}(\delta) < \MRRW^{(2)}(\delta).
\end{equation*}
\end{prop}

\begin{proof}
Suppose first that $\MRRW^{(2)}(\delta) = \MRRW(\delta)$. By~\cref{rem:2mqc-endpoints} and~\cref{prop:mqc-strict-improvement}, we have that 
\begin{equation*}
R_{\mMQC}(\delta)
\leq R_{\MQC}(\delta)
< \MRRW(\delta)
= \MRRW^{(2)}(\delta).
\end{equation*}
For the remainder of the proof, assume $\MRRW^{(2)}(\delta) < \MRRW(\delta)$. 
By~\eqref{eq:second-mrrw-holevo-function} and the variational formula of the second MRRW bound~\eqref{eq:second-mrrw}, we have
\begin{equation*}
\MRRW^{(2)}(\delta) = \min_{0 \leq u \leq 1-2\delta}{H_\delta(u)}.
\end{equation*}
The continuous function $H_\delta$ attains its minimum on $[0,1-2\delta]$. The right endpoint is not a minimizer because
\begin{equation*}
H_\delta(1-2\delta) = \MRRW(\delta) > \MRRW^{(2)}(\delta).
\end{equation*}
The left endpoint is not a minimizer either. Indeed, $s(u^2) = O(u^2\log(1/u))$ has vanishing one-sided derivative at $u = 0$, and hence
\begin{equation*}
H_\delta'(0^+) = -2\delta s'(2\delta) < 0.
\end{equation*}
Thus, we can fix a minimizer $u_* \in (0,1-2\delta)$.

Now, observe that the mPSC endpoint is when $v = 0$. Indeed, by~\eqref{eq:2mqc-entropy-coordinates}, we have $v = 0$ if and only if $\eta = 0$, which recovers mPSC by~\cref{rem:2mqc-endpoints}. Now, define the shorthand
\begin{equation}
\label{eq:2mqc-exact-error-objective}
G_\delta(u,v)
\coloneqq
1+s(u^2)-s\left(\lambda_\delta(u,v)\right)-s(v^2),
\end{equation}
where $\lambda_\delta(u,v)$ is as defined in~\cref{eq:2mqc-exact-error-coordinates}. Note that when $v = 0$, we have that $\lambda_\delta(u_*,0) = u_*^2 + 2\delta u_* + 2\delta$. Thus
\begin{equation}
\label{eq:j-to-mrrw}
G_\delta(u_*,0)
= H_\delta(u_*)
= \MRRW^{(2)}(\delta),
\end{equation}
as expected by~\cref{rem:2mqc-endpoints}. 

We will now perturb $v \to 0^+$ slightly away from $0$ on the positive side while keeping the PGM bit error rate fixed. Now, from \cref{eq:2mqc-exact-error-coordinates}, we see that $\lambda_\delta(u_*,v) = \lambda_0+O(v^2)$ for $\lambda_0 \coloneqq \lambda_\delta(u_*,0) \in (0,1)$. Thus, for all sufficiently small $v > 0$, we have that $A_\delta(u_*,v) = \lambda_\delta(u_*,v)$. Consequently, by~\cref{eq:2mqc-rate}, we have 
\begin{equation}
\label{eq:2mqc-to-j}
R_{\mMQC}(\delta) \leq G_\delta(u_*,v)
\end{equation}
for $v \to 0^+$. Now, write the difference
\begin{equation*}
G_\delta(u_*,v)-G_\delta(u_*,0)
= -\left[s\left(\lambda_\delta(u_*,v)\right)-s(\lambda_0)\right]-s(v^2).
\end{equation*}
By the mean value theorem, we have that
\begin{equation*}
s\left(\lambda_\delta(u_*,v)\right)-s(\lambda_0) = O(v^2).
\end{equation*}
For the second term, the expansion used in the proof of~\cref{prop:mqc-strict-improvement} gives
\begin{equation*}
s(v^2) = \frac{v^2}{2}\log{\frac{1}{v}}+O(v^2).
\end{equation*}
Therefore,
\begin{equation*}
\frac{G_\delta(u_*,v)-G_\delta(u_*,0)}{v^2}
= -\frac{1}{2}\log{\frac{1}{v}}+O(1),
\end{equation*}
which tends to $-\infty$ as $v \to 0^+$. Hence 
\begin{equation}
\label{eq:j-strict-improvement}
G_\delta(u_*,v) < G_\delta(u_*,0)
\end{equation}
for all sufficiently small $v > 0$. Combining~\cref{eq:2mqc-to-j,eq:j-strict-improvement,eq:j-to-mrrw}, we therefore conclude that
\begin{equation*}
R_{\mMQC}(\delta)
\leq G_\delta(u_*,v)
< G_\delta(u_*,0)
= \MRRW^{(2)}(\delta).\qedhere
\end{equation*}
\end{proof}

With this proposition at hand, we can now establish~\cref{thm:second-mrrw-improvement}.

\begin{proof}[Proof of \cref{thm:second-mrrw-improvement}]
The theorem follows by combining the 2MQC rate upper bound $R_2(\delta) \leq R_{\mMQC}(\delta)$ (\cref{cor:2mqc-rate}) together with the strict inequality $R_{\mMQC}(\delta) < \MRRW^{(2)}(\delta)$ (\cref{prop:2mqc-strict-mrrw2}).\qedhere
\end{proof}

We conclude this section with numerical estimates of our rate bound improvement obtained from the 2MQC (\cref{cor:2mqc-rate}) relative to the previous state-of-the-art second MRRW bound (\cref{cor:mpsc-rate}). This is presented in \cref{tab:2mqc-mrrw-numerics}. The second MRRW bound column minimizes the one-dimensional objective in \eqref{eq:second-mrrw}. The 2MQC value column records reproducible floating-point estimates of the minimal value of the optimization problem in~\cref{eq:2mqc-rate}. For the first four rows, the 2MQC column evaluates the exact-error objective $G_\delta(u,v)$ from \eqref{eq:2mqc-exact-error-objective} at the displayed interior coordinates, as $\lambda_\delta(u,v) < 1$ in that regime. For the remaining rows, $A_\delta(u,v) = 1$ instead of $\lambda_\delta(u,v)$.

\begin{table}[!htbp]
\centering
\caption{The second MRRW bound compared numerically with floating-point estimates of $R_{\mMQC}(\delta)$, based on estimating the minimal value of the optimization problem in~\cref{eq:2mqc-rate}. The displayed gap in the last column is $\MRRW^{(2)}(\delta)$ minus an estimate of the 2MQC rate bound $R_{\mMQC}(\delta)$. The displayed coordinates $u,v$ are floating-point estimates witnessing the numerical upper bound on $R_{\mMQC}(\delta)$ displayed in the table. Note that all entries are floating-point evaluations, not certified global optima.}
\label{tab:2mqc-mrrw-numerics}
\footnotesize
\renewcommand{\arraystretch}{1.4}
\begin{tabular}{@{}cccccc@{}}
\toprule
$\boldsymbol{\delta}$
& $\boldsymbol{u}$
& $\boldsymbol{v}$
& $\boldsymbol{\MRRW^{(2)}(\delta)}$
& \textbf{2MQC value}
& \textbf{Gap} \\
\midrule
$0.05$ & $0.0544815$ & $0.00135647$ & $0.825136808$ & $0.825136151$ & $0.000000657$ \\
$0.10$ & $0.1193920$ & $0.00618142$ & $0.692740743$ & $0.692727379$ & $0.000013364$ \\
$0.15$ & $0.1997723$ & $0.01681725$ & $0.573450044$ & $0.573356705$ & $0.000093339$ \\
$0.20$ & $0.3160463$ & $0.04188496$ & $0.461359604$ & $0.460883553$ & $0.000476051$ \\
$0.25$ & $0.5293949$ & $0.12474570$ & $0.353710543$ & $0.350379084$ & $0.003331459$ \\
$0.30$ & $0.4151219$ & $0.07923030$ & $0.250224912$ & $0.248376065$ & $0.001848847$ \\
$0.35$ & $0.3064440$ & $0.04443793$ & $0.158132937$ & $0.157505576$ & $0.000627361$ \\
$0.40$ & $0.2019380$ & $0.01978282$ & $0.081468915$ & $0.081336612$ & $0.000132303$ \\
$0.45$ & $0.1002467$ & $0.00497293$ & $0.025266128$ & $0.025257384$ & $0.000008744$ \\
\bottomrule
\end{tabular}
\end{table}
\FloatBarrier

\section{Extending to $q$-ary codes}
\label{sec:qary-codes}

In this section, we sketch how one might extend our framework to $q$-ary codes. Fix an integer $q \geq 2$, write $[q] \coloneqq \{0,\ldots,q-1\}$, and let $A_q(n,d)$ be the largest cardinality of a code in $[q]^n$ with minimum distance at least $d$. We normalize rates in base $q$ and write
\begin{equation*}
R_q(\delta)
\coloneqq
\limsup_{n \to \infty}\frac{1}{n}\log_q{A_q(n,\ceil{\delta n})}.
\end{equation*}
For $\delta \geq 1-\frac{1}{q}$, the $q$-ary Plotkin bound shows that $R_q(\delta) = 0$. Thus, the nontrivial interval in the $q$-ary setting is $\delta \in (0,1-\frac{1}{q})$. Since the pretty good criterion (\cref{thm:main-result}) is stated for binary codes, we first extend its output-symmetry hypothesis to the cyclic $q$-ary alphabet.

\begin{definition}[$q$-ary output-symmetry]
\label{def:qary-output-symmetry}
For a positive integer $d$, a $q$-ary-input cq channel $\Phi_\sigma : [q] \to \Dens(\C^d)$, given by $\Phi_\sigma(x)=\sigma_x$, is \textbf{$q$-ary output-symmetric} if there is a unitary $U$ such that
\begin{equation*}
U\sigma_xU^\dagger=\sigma_{x+1}
\end{equation*}
for every $x \in [q]$, where subscripts and addition are taken modulo $q$. Iterating this identity gives $U^a\sigma_x(U^\dagger)^a=\sigma_{x+a}$ for all $a,x \in [q]$.
\end{definition}

With this terminology in place, the following proposition records the precise $q$-ary extension.

\begin{prop}[$q$-ary pretty good criterion]
\label{prop:qary-main-criterion}
Let $\Phi_\sigma : x \mapsto \sigma_x$ be a $q$-ary output-symmetric cq channel. If the uniform-input PGM symbol error rate $p_{\mathrm{e}}(\{(1/q,\sigma_x) : x \in [q]\}) < \delta$, then every code $\cC \subseteq [q]^n$ with minimum distance at least $\delta n$ satisfies
\begin{equation} \label{eq:qary-main-criterion}
\frac{1}{n}\log_q{\abs{\cC}}
\leq
\frac{\chi(\{(1/q,\sigma_x)\}_{x \in [q]})}{\log q}
+O_{q,\sigma,\delta}\left(\frac{1}{\sqrt{n}}\right).
\end{equation}
Consequently,
\begin{equation*}
R_q(\delta)
\leq
\frac{\chi(\{(1/q,\sigma_x)\}_{x \in [q]})}{\log q}.
\end{equation*}
\end{prop}

\begin{proof}
Let $U$ be the unitary witnessing output-symmetry. Write
\begin{equation*}
p_{\mathrm{e}} \coloneqq p_{\mathrm{e}}\left(\{(1/q,\sigma_x) : x \in [q]\}\right),
\qquad
\chi_u \coloneqq \chi\left(\{(1/q,\sigma_x) : x \in [q]\}\right).
\end{equation*}
The PGM coarse-graining and data-processing statements, \cref{lem:coarse-graining,lem:pgm-data-processing}, are already formulated for arbitrary finite ensembles, so they apply here without modification. Since the reduction of blockwise to bitwise PGM in~\cref{lem:block-pgm-bitwise} was specialized to the binary case, we will instead directly invoke \cref{lem:coarse-graining} for its $q$-ary analog. The remaining ingredient in the proof of the $q$-ary analog of \cref{thm:expected-hamming} is a $q$-ary analog of \cref{lem:uniform-prior-worst}. We first prove its $q$-ary replacement: for every distribution $\pi$ on $[q]$, the PGM error $p_{\mathrm{e},\pi}$ of $\{(\pi_x,\sigma_x) : x \in [q]\}$ satisfies $p_{\mathrm{e},\pi} \leq p_{\mathrm{e}}$.

The proof of the $q$-ary analog of~\cref{lem:uniform-prior-worst} closely follows the one presented in~\cref{appsubsec:proof-uniform-prior}. Indeed, define
\begin{equation*}
\bar{\sigma}_\pi \coloneqq \sum_{x \in [q]}{\pi_x\sigma_x},
\qquad
M_{\hat{x}}^\pi \coloneqq \pi_{\hat{x}}\bar{\sigma}_\pi^{-1/2}\sigma_{\hat{x}}\bar{\sigma}_\pi^{-1/2},
\end{equation*}
and introduce a classical system with basis $\{\ket{a} : a \in [q]\}$. For each $x \in [q]$, define
\begin{equation*}
\widehat{\sigma}_x
\coloneqq
\sum_{a \in [q]}
\pi_{x-a}\ketbra{a}{a} \otimes \sigma_x.
\end{equation*}
The translated probabilities sum to one, so each $\widehat{\sigma}_x$ is a bona fide density state. Endow the lifted states $\widehat{\sigma}_x$ with the uniform prior. Their average state and PGM operators are, respectively,
\begin{align*}
\overline{\widehat{\sigma}}
&=\frac{1}{q}\sum_{a \in [q]}{\ketbra{a}{a} \otimes U^a\bar{\sigma}_\pi (U^\dagger)^a}, \\
\widehat{M_{\hat{x}}}
&=\sum_{a \in [q]}{\ketbra{a}{a} \otimes U^aM_{\hat{x}-a}^\pi (U^\dagger)^a}.
\end{align*}
Indeed, conditioned on block $a$, the translated symbol $x-a$ has law $\pi$ and $\sigma_x = U^a\sigma_{x-a}(U^\dagger)^a$. The lifted PGM therefore has success probability
\begin{equation*}
\frac{1}{q}\sum_{x \in [q]}{\Tr(\widehat{\sigma}_x\widehat{M_x})}
=\frac{1}{q}\sum_{a,x \in [q]}{\pi_{x-a}\Tr(\sigma_{x-a}M_{x-a}^\pi)} =\sum_{x \in [q]}{\pi_x\Tr(\sigma_xM_x^\pi)}.
\end{equation*}
Thus, the lifted uniform ensemble has PGM error exactly $p_{\mathrm{e},\pi}$.

Tracing out the classical system sends $\widehat{\sigma}_x$ to $\sigma_x$. Moreover, the PGM data-processing lemma (\cref{lem:pgm-data-processing}) says that this trace-out channel cannot decrease PGM error, and hence
\begin{equation}
\label{eq:qary-uniform-prior-worst}
p_{\mathrm{e},\pi} \leq p_{\mathrm{e}}.
\end{equation}
This proves the $q$-ary counterpart of \cref{lem:uniform-prior-worst}.

We now follow the proof of \cref{thm:expected-hamming}. Draw $\bc$ uniformly from $\cC$, apply the block PGM to $\sigma_{\bc} \coloneqq \sigma_{\bc_1} \otimes \cdots \otimes \sigma_{\bc_n}$, and denote its codeword-valued outcome by $\bchat$. For each coordinate $i$, let $\pi_i(x) \coloneqq \Prb[\bc_i = x]$. By \cref{lem:coarse-graining}, which applies as-is with the coordinate map $c \mapsto c_i$, we have
\begin{equation*}
\Prb[\bchat_i \neq \bc_i]
= p_{\mathrm{e}}\left(\{(\pi_i(x),\Eb[\sigma_{\bc} \mid \bc_i = x]) : x \in [q]\}\right).
\end{equation*}
The partial trace channel similarly satisfies $\T_i(\Eb[\sigma_{\bc} \mid \bc_i = x]) = \sigma_x$. Applying \cref{lem:pgm-data-processing} as-is and then the $q$-ary uniform-prior bound \eqref{eq:qary-uniform-prior-worst} gives
\begin{equation*}
\Prb[\bchat_i \neq \bc_i]
\leq p_{\mathrm{e}}\left(\{(\pi_i(x),\sigma_x) : x \in [q]\}\right)
\leq p_{\mathrm{e}}.
\end{equation*}
Summing over the coordinates yields the $q$-ary counterpart of \cref{thm:expected-hamming}:
\begin{equation*}
\Eb[d(\bchat,\bc)] = \sum_{i = 1}^n{\Prb[\bchat_i \neq \bc_i]} \leq np_{\mathrm{e}}.
\end{equation*}
The argument of \cref{cor:block-error} uses only Hamming distance and therefore extends unchanged. Since distinct codewords are separated by at least $\delta n$, it gives
\begin{equation*}
\Prb[\bchat = \bc] \geq 1-\frac{p_{\mathrm{e}}}{\delta} \eqqcolon \eta > 0.
\end{equation*}

It remains to reproduce the final converse step in the proof of \cref{thm:main-result}. Here we use the $q$-ary output-symmetric counterparts of \cref{prop:output-symmetric-capacity,thm:cq-strong-converse}. The same averaging over input shifts and entropy-subadditivity proof as in \cref{prop:output-symmetric-capacity} gives
\begin{equation*}
C(\Phi_\sigma) = \chi_u.
\end{equation*}
Likewise, the proof of \cref{thm:cq-strong-converse} applies using the unitary $U$ from \cref{def:qary-output-symmetry} and setting $V_c \coloneqq U^{c_1} \otimes \cdots \otimes U^{c_n}$. Indeed, $V_c\sigma_0^{\otimes n}V_c^\dagger = \sigma_c$, while every power $U^a$ fixes the uniform average state. The decoder-averaging step and both spectral truncations therefore remain unchanged. Thus, there is a channel-dependent constant $L_\sigma \geq 1$ such that, for every $\eps > 0$,
\begin{equation*}
\Prb[\bchat = \bc]
\leq 2^{-n[n^{-1}\log{\abs{\cC}}-\chi_u-2\eps]}
+4\exp\left(-\eps^2n/L_\sigma^2\right).
\end{equation*}
Exactly as in the proof of \cref{thm:main-result}, choose
\begin{equation*}
\eps \coloneqq L_\sigma\sqrt{\frac{\log(8/\eta)}{n\log{\mathrm{e}}}}.
\end{equation*}
The exponential term is then $\eta/2$, and hence
\begin{equation*}
\eta
\leq \Prb[\bchat = \bc]
\leq 2^{-n[n^{-1}\log{\abs{\cC}}-\chi_u-2\eps]}
+\frac{\eta}{2}.
\end{equation*}
Taking logarithms and rearranging gives
\begin{align*}
\frac{1}{n}\log{\abs{\cC}}
&\leq \chi_u + 2\eps + \frac{1}{n}\log{\frac{2}{\eta}} \\
&= \chi_u + 2L_\sigma\sqrt{\frac{\log(8/\eta)}{n\log{\mathrm{e}}}} + \frac{1}{n}\log{\frac{2}{\eta}} \\
&= \chi_u + O_{q,\sigma,\delta}\left(\frac{1}{\sqrt{n}}\right).
\end{align*}
Dividing by $\log{q}$ proves \eqref{eq:qary-main-criterion}. Maximizing over $\cC$ and taking $n \to \infty$ gives
\begin{equation*}
R_q(\delta) \leq \frac{\chi_u}{\log{q}}.\qedhere
\end{equation*}
\end{proof}

Thus, as in the binary setting, the $q$-ary pretty good criterion reduces to channel selection. Indeed, choose a $q$-ary output-symmetric cq channel, read off its uniform-input PGM symbol error and Holevo information, and substitute them into \cref{prop:qary-main-criterion}. To state the resulting bounds compactly, define
\begin{align*}
h_q(x)
&\coloneqq \frac{h(x)+x\log(q-1)}{\log{q}}, \\
J_q(x)
&\coloneqq \frac{q-1}{q}\left(1-\sqrt{1-\frac{qx}{q-1}}\right), \\
\gamma_q(x)
&\coloneqq \frac{q-1-(q-2)x-2\sqrt{(q-1)x(1-x)}}{q}.
\end{align*}
We now specify only the channels. Their PGM symbol errors, Holevo information, and induced rate bounds are collected in \cref{tab:qary-channel-summary}.

\paragraph{$q$-ary erasure channel ($q$-EC).} For $0 \leq p \leq 1$, let $\ket{?}$ be orthogonal to $\{\ket{x}\}_{x \in [q]}$, and define
\begin{equation*}
\sigma_x^{q\mathrm{EC}}(p)
\coloneqq
(1-p)\ketbra{x}{x}+p\ketbra{?}{?}.
\end{equation*}

\paragraph{$q$-ary symmetric channel ($q$-SC).} For $0 \leq p \leq \frac{q-1}{q}$, define the diagonal output states
\begin{equation*}
\sigma_x^{q\mathrm{SC}}(p)
\coloneqq
(1-p)\ketbra{x}{x}
+\frac{p}{q-1}\sum_{y \in [q] \setminus \{x\}}{\ketbra{y}{y}}.
\end{equation*}

\paragraph{$q$-ary pure-state channel ($q$-PSC).} For $0 \leq p \leq \frac{q-1}{q}$, define
\begin{equation*}
\sigma_x^{q\mathrm{PSC}}(p)
\coloneqq
\ketbra{\psi_x^{(q)}(p)}{\psi_x^{(q)}(p)},
\qquad
\ket{\psi_x^{(q)}(p)}
\coloneqq
\sqrt{1-p}\ket{x}
+\sqrt{\frac{p}{q-1}}\sum_{y \in [q] \setminus \{x\}}{\ket{y}}.
\end{equation*}

\paragraph{$q$-ary mixed channel ($q$-MC).} Following the binary MQC construction in \cref{subsec:how-mixing,subsec:mqc}, we place a $q$-ary symmetric channel before a $q$-PSC. Let $\mathcal{T}^{(q)}_\eta$ send $x$ to itself with probability $1-\eta$ and to each $y \neq x$ with probability $\eta/(q-1)$, and let $\mathcal{W}^{(q)}_r$ be the preparation $y \mapsto \sigma_y^{q\mathrm{PSC}}(r)$. Their composition $\mathcal{W}^{(q)}_r \circ \mathcal{T}^{(q)}_\eta$ is the $q$-MC: it prepares the pure state indexed by the perturbed symbol and discards the classical error.

For $0 \leq a \leq \frac{q-1}{q}$, define $\nu_q(a) \coloneqq 1-qa/(q-1)$ as the centered correlation of the $q$-SC with crossover probability $a$. Also denote its inverse by $r_q(v) \coloneqq (q-1)(1-v)/q$. Fix a desired PGM symbol error $0 \leq p < \frac{q-1}{q}$ and $\nu_q(p) \leq v \leq 1$, and set
\begin{equation*}
t_q(p,v) \coloneqq \sqrt{\frac{\nu_q(p)}{v}},
\qquad
\eta_q(p,v) \coloneqq r_q\left(t_q(p,v)\right).
\end{equation*}
Here $r_q(v)$ is the parameter of the underlying $q$-PSC, while $\eta_q(p,v)$ is the crossover probability of the discarded $q$-SC. Abbreviate
\begin{equation*}
\tau_x(v) \coloneqq \sigma_x^{q\mathrm{PSC}}\left(r_q(v)\right),
\qquad
\bar{\tau}_v \coloneqq \frac{1}{q}\sum_{y \in [q]}{\tau_y(v)}.
\end{equation*}
The composition therefore has output
\begin{align}
\sigma_x^{q\mathrm{MC}}(p,v)
&\coloneqq
\left(\mathcal{W}^{(q)}_{r_q(v)} \circ
\mathcal{T}^{(q)}_{\eta_q(p,v)}\right)(x) \notag \\
&=\left(1-\eta_q(p,v)\right)\tau_x(v)
+\frac{\eta_q(p,v)}{q-1}
\sum_{y \in [q]\setminus\{x\}}{\tau_y(v)} \notag \\
&=t_q(p,v)\tau_x(v)
+\left(1-t_q(p,v)\right)\bar{\tau}_v.
\label{eq:qary-mc}
\end{align}
The last equality uses $\sum_{y \in [q] \setminus \{x\}}{\tau_y(v)} = q\bar{\tau}_v-\tau_x(v)$ and $t_q(p,v) = \nu_q(\eta_q(p,v))$.

For the uniform $q$-PSC ensemble, the PGM is the computational-basis measurement. As in the binary MQC, the PGM reverses the classical pre-processing by applying the same $q$-SC once more. As correlations multiply under composition, the final symbol error $p_{\mathrm{e}}$ satisfies $\nu_q(p_{\mathrm{e}})=t_q(p,v)^2v=\nu_q(p)$ and hence equals $p$. At $v = \nu_q(p)$, the pre-processing is noiseless, and the channel is exactly $q$-PSC$(p)$; at $v = 1$, the pure states are orthogonal and the channel is exactly $q$-SC$(J_q(p))$. For alphabet size $2$, \eqref{eq:qary-mc} becomes
\begin{equation*}
\MQC\left(\frac{1-v}{2},\frac{1-t_2(p,v)}{2}\right),
\end{equation*}
with the first argument the underlying binary PSC parameter and the second the discarded BSC crossover probability.

To state its Holevo information compactly, define
\begin{align*}
g_q(v)
&\coloneqq \gamma_q\left(r_q(v)\right), \\
\ell_q(p,v)
&\coloneqq \frac{\left(1-t_q(p,v)\right)g_q(v)}{q-1}, \\
\lambda_\pm^{(q)}(p,v)
&\coloneqq \frac{1}{2}\left(1-(q-2)\ell_q(p,v)
\pm \left[\left(1-(q-2)\ell_q(p,v)\right)^2\right.\right. \\
&\hspace{5.5em}\left.\left.{}-\frac{4\left(1-t_q(p,v)\right)\left(1+(q-1)t_q(p,v)\right)}{q-1}
g_q(v)\left(1-g_q(v)\right)\right]^{1/2}\right), \\
\Xi_q(p,v)
&\coloneqq h_q\left(g_q(v)\right)
+\frac{1}{\log{q}}\left[\lambda_+^{(q)}(p,v)\log{\lambda_+^{(q)}(p,v)}
+\lambda_-^{(q)}(p,v)\log{\lambda_-^{(q)}(p,v)}\right. \\
&\hspace{8.5em}\left.{}+(q-2)\ell_q(p,v)\log{\ell_q(p,v)}\right].
\end{align*}
Each $q$-MC output has spectrum $\lambda_+^{(q)}(p,v),\lambda_-^{(q)}(p,v)$, and $\ell_q(p,v)$ with multiplicity $q-2$, while the average state has spectrum $1-g_q(v)$ and $g_q(v)/(q-1)$ with multiplicity $q-1$. These spectra give the normalized Holevo information $\Xi_q(p,v)$, and the uniform-input PGM symbol error is exactly $p$.

The unitary that sends $\ket{x}$ to $\ket{x+1}$ (and fixes $\ket{?}$ for the $q$-EC) shows that all four channels are $q$-ary output-symmetric. For the open channel-design problem, write
\begin{equation*}
R_{q,\mathrm{cq}}^\star(\delta)
\coloneqq
\inf_{\substack{\{\sigma_x\}_{x \in [q]}\ \text{$q$-ary output-symmetric}\\
p_{\mathrm{e}}(\{(1/q,\sigma_x)\}_{x \in [q]}) < \delta}}
\frac{\chi(\{(1/q,\sigma_x)\}_{x \in [q]})}{\log{q}}.
\end{equation*}
The final column below follows by choosing a strictly feasible channel in \cref{prop:qary-main-criterion} and then approaching the displayed error threshold by continuity.

\begin{table}[!htbp]
\centering
\caption{$q$-ary channel blueprint. All rates and Holevo quantities are normalized by $\log{q}$. The first four rows give explicit channels used in \cref{prop:qary-main-criterion}; the last row records the open optimization over all $q$-ary output-symmetric cq channels.}
\label{tab:qary-channel-summary}
\footnotesize
\renewcommand{\arraystretch}{1.4}
\begin{tabularx}{\textwidth}{@{} >{\centering\arraybackslash}p{0.17\textwidth} >{\centering\arraybackslash}p{0.20\textwidth} >{\centering\arraybackslash}p{0.25\textwidth} >{\centering\arraybackslash}X@{}}
\toprule
\textbf{Channel}
& \textbf{PGM symbol error}
& \textbf{Uniform-input $\chi/\log{q}$}
& \textbf{Upper bound on $R_q(\delta)$} \\
\midrule
$q$-EC$(p)$
& $\frac{q-1}{q}p$
& $1-p$
& $1-\frac{q}{q-1}\delta$ \\
\addlinespace[0.2em]
$q$-SC$(p)$
& $2p-\frac{q}{q-1}p^2$
& $1-h_q(p)$
& $1-h_q\left(J_q(\delta)\right)$ \\
\addlinespace[0.2em]
$q$-PSC$(p)$
& $p$
& $h_q\left(\gamma_q(p)\right)$
& $h_q\left(\gamma_q(\delta)\right)$ \\
\addlinespace[0.2em]
$q$-MC$(p,v)$
& $p$
& $\Xi_q(p,v)$
& $\displaystyle\min_{\nu_q(\delta) \leq v \leq 1}{\Xi_q(\delta,v)}$ \\
\addlinespace[0.2em]
\shortstack{Open\\output-symmetric\\channel search}
& $p_{\mathrm{e}} < \delta$
& $\chi(\{(1/q,\sigma_x)\}_x)/\log{q}$
& $R_{q,\mathrm{cq}}^\star(\delta)$ \\
\bottomrule
\end{tabularx}
\end{table}
\FloatBarrier

The first three rows recover the $q$-ary Plotkin bound~\cite{Plo60}, the $q$-ary Elias--Bassalygo bound~\cite{Eli55,Bas65}, and the first Delsarte--Levenshtein linear-programming bound~\cite{Del73,Lev95}, respectively. The general $q$-ary linear-programming framework originates with Delsarte~\cite{Del72,Del73}, was developed systematically by Levenshtein~\cite{Lev95} and Delsarte and Levenshtein~\cite{DL98}, and admits more recent explicit dual constructions~\cite{CDA25}. Kaipa's anticode-based hybrid, using the optimal-anticode theorem of Ahlswede and Khachatrian, improves both the $q$-ary Elias and Plotkin bounds~\cite{AK98,Kai18}. At $q = 2$, the $q$-PSC row is the first MRRW bound~\cite{MRRW77}, while the $q$-MC row becomes the binary MQC bound.

The $q$-MC supplies the desired $q$-ary output-symmetric fixed-error path from $q$-SC$(J_q(p))$ to $q$-PSC$(p)$. We do not claim here that it improves the best known $q$-ary bounds when $q > 2$. Its optimized curve should first be compared with Aaltonen's nonbinary linear-programming bound~\cite{Aal90}, the shortening and straight-line bound of Laihonen and Litsyn~\cite{LL98}, and the product-association-scheme bounds of Ben-Haim and Litsyn~\cite{BHL07}. Searching beyond the $q$-MC family is represented by the last row of the table.

\section{PSC rate bounds for LDPC codes}
\label{sec:psc-ldpc}

This section proves \cref{thm:psc-ldpc-improvement}. It is primarily intended as a proof of concept showing how the pretty good criterion (\cref{thm:main-result}) can incorporate structure specific to a code family and thereby improve a specialized rate--distance bound. For simplicity, we specialize throughout to parity checks of weight at most $3$.

We begin with the benchmark bound of Shangguan and Yang~\cite{SY26} in the notation of this paper. Set $r_\delta \coloneqq \frac{1}{2}-\sqrt{\delta(1-\delta)}$. The $w = 3$ bound in their Theorem~1.5, with their base-four entropy expanded into base-two logarithms, is 
\begin{equation}
\label{eq:sy-ldpc-base}
B_3^{\mathrm{SY}}(\delta)
\coloneqq
\begin{cases}
\frac{2}{3}, & r_\delta \geq \frac{1}{4}, \\
-r_\delta\log{r_\delta}
-\frac{1-3r_\delta}{3}\log(1-3r_\delta),
& 0 < r_\delta < \frac{1}{4},
\end{cases}
\end{equation}
extended continuously by $B_3^{\mathrm{SY}}(\frac{1}{2}) \coloneqq 0$. Throughout $\frac{2}{5} \leq \delta < \frac{1}{2}$, this is the $w = 3$ bound obtained in~\cite{SY26}. In proving~\cref{thm:psc-ldpc-improvement}, we will compare the rate bound obtained from the PSC with this benchmark for $\delta \in [\frac{1}{6},\frac{1}{2})$.

Recall the two exact optimality statements for PGM decoding over PSC from~\cref{subsec:main-results}. Namely, for binary linear codes, the PGM achieves optimal error at both the block and bit levels. In this section, we will rely on PGM's bitwise optimality, stated below.

\begin{fact}[{Bitwise PGM optimality for linear codes ---~\cite[Proposition~5.1]{PR25}}]
\label{fact:psc-bitwise-optimal}
Let $\cC \leq \F_2^n$ be linear and $p \in (0,\frac{1}{2})$. For the uniform ensemble $\{(\abs{\cC}^{-1},\sigma_c^{\PSC(p)}) : c \in \cC\}$, let $\bchat$ be the outcome of the PGM. For any coordinate $i \in \{1, \ldots , n\}$, the $i$'th outcome bit $\bchat_i$ is the Helstrom-optimal measurement~\cite{Hel69} for decoding $\bc_i$. In particular, the error of any explicit decoder for $\bc_i$ upper-bounds the bitwise error $p_{\mathrm{e},i}$ of the block PGM.
\end{fact}

The explicit decoder we use will apply the equal-prior Helstrom measurement directly to the PSC outputs indexed by one weight-$3$ parity check containing the desired bit. Define the weight-three error function
\begin{equation}
\label{eq:ldpc-local-error}
e_3^{\mathrm{LDPC}}(p)
\coloneqq
\frac{1+4p(1-p)-(1-2p)\sqrt{1+12p(1-p)}}{4}.
\end{equation}

We next bound this Helstrom error and then invoke the bitwise optimality of the code PGM to bound $p_{\mathrm{e},i}$ from above.

\begin{lemma}[One-check PSC decoder]
\label{lem:psc-local-check}
Let $p \in (0,\frac{1}{2})$. Let $\cC \leq \F_2^n$ have minimum distance greater than one and let $\cC^\perp$ be generated by checks of weight at most $3$. If a uniformly random $\bc \in \cC$ is encoded into the product state $\sigma_{\bc}^{\PSC(p)}$, then for every $i \in \{1, \ldots , n\}$, the bitwise coarse-graining of the block PGM has error
\begin{equation*}
p_{\mathrm{e},i}
\leq e_3^{\mathrm{LDPC}}(p)
< p.
\end{equation*}
\end{lemma}

\begin{proof}
Fix $i \in \{1, \ldots , n\}$. Because $\cC^\perp$ is generated by checks of weight at most $3$, some such generating check contains $i$; otherwise, the weight-one vector supported at $i$ would lie in $\cC$, contrary to the minimum-distance assumption. Choose a minimum-weight $h \in \cC^\perp$ with $h_i = 1$. If $\wt(h) = 1$, the check fixes $\bc_i = 0$, so the bitwise error is zero. Suppose henceforth that $\wt(h) \in \{2,3\}$. The minimality of $h$ implies that the weight-one vector supported at $i$ is not in $\cC^\perp$; hence $c \mapsto c_i$ is a nonzero linear functional on $\cC$, and $\bc_i$ is uniform.

By \cref{fact:psc-bitwise-optimal}, the error of any decoder for $\bc_i$ upper-bounds $p_{\mathrm{e},i}$. Suppose first that $h$ has support $\{i,j,k\}$. Minimality implies that the projection of $\cC$ onto these coordinates is the full even-parity code $\{000,011,101,110\}$. Indeed, any additional linear relation on the projection, or its sum with $h$, would give a check of weight at most two containing $i$. Since $\bc$ is uniform on $\cC$, its projection is uniform as well. Thus, conditioned on $\bc_i = b$, the three outputs have states
\begin{equation*}
\omega_0
\coloneqq
\frac{\sigma_{000}^{\PSC(p)}+\sigma_{011}^{\PSC(p)}}{2},
\qquad
\omega_1
\coloneqq
\frac{\sigma_{101}^{\PSC(p)}+\sigma_{110}^{\PSC(p)}}{2}.
\end{equation*}
We compute the error of their equal-prior Helstrom measurement. Let $D \coloneqq 4p(1-p)$ and let $V$ be the matrix whose columns are the four pure vectors underlying the states above, ordered as $000,011,101,110$. Any two distinct words in this list have Hamming distance two, so
\begin{equation*}
V^\dagger V = (1-D)I_4+D\one\one^\top.
\end{equation*}
Writing $\Lambda \coloneqq \operatorname{diag}(1,1,-1,-1)$ gives $\omega_0-\omega_1 = \frac{1}{2}V\Lambda V^\dagger$. The nonzero eigenvalues of this difference are the eigenvalues of $\frac{1}{2}\Lambda V^\dagger V$; its characteristic polynomial is
\begin{equation*}
\left(\lambda^2-\frac{(1-D)^2}{4}\right)
\left(\lambda^2-\frac{(1-D)(1+3D)}{4}\right).
\end{equation*}
Consequently,
\begin{equation*}
\norm{\omega_0-\omega_1}_1
=1-D+\sqrt{(1-D)(1+3D)},
\end{equation*}
and the Helstrom formula~\cite[Theorem 3.2]{Wri24} gives
\begin{equation*}
\frac{1}{2}-\frac{1}{4}\norm{\omega_0-\omega_1}_1
=e_3^{\mathrm{LDPC}}(p).
\end{equation*}

If $h$ has support $\{i,j\}$, its local projection is $\{00,11\}$ and the corresponding Helstrom measurement distinguishes the two pure states $\sigma_0^{\PSC(p)} \otimes \sigma_0^{\PSC(p)}$ and $\sigma_1^{\PSC(p)} \otimes \sigma_1^{\PSC(p)}$. Its error is
\begin{equation*}
\frac{1}{2}\left(1-\sqrt{1-16p^2(1-p)^2}\right)
\leq e_3^{\mathrm{LDPC}}(p),
\end{equation*}
where the inequality is equivalent, after two squarings, to
\begin{equation*}
2\sqrt{1+D}
\geq
\sqrt{1-D}+\sqrt{1+3D}.
\end{equation*}
The weight-one case has zero error. Therefore, in every case,
\begin{equation*}
p_{\mathrm{e},i}
\leq e_3^{\mathrm{LDPC}}(p).
\end{equation*}
Finally, $e_3^{\mathrm{LDPC}}(p) < p$ is equivalent to
\begin{equation*}
1-2p+\sqrt{1+12p(1-p)} > 2,
\end{equation*}
whose square reduces to $8p(1-2p) > 0$.\qedhere
\end{proof}

\begin{remark}[Beyond one check]
\label{rem:local-decoder-optimality}
The decoder in \cref{lem:psc-local-check} is optimal among measurements restricted to the outputs in the chosen check. A joint measurement on the outputs covered by several checks containing $i$ may nevertheless decrease the bit error further.
\end{remark}

The associated PSC rate curve has a closed form. Indeed, in the coordinate $D = 4p(1-p)$, the function in \eqref{eq:ldpc-local-error} becomes
\begin{equation*}
\frac{1+D-\sqrt{(1-D)(1+3D)}}{4},
\end{equation*}
whose derivative with respect to $D$ is
\begin{equation*}
\frac{1}{4}\left(1-\frac{1-3D}{\sqrt{(1-D)(1+3D)}}\right)>0.
\end{equation*}
Thus, it increases strictly from zero to $\frac{1}{2}$ as $D$ increases from zero to one. Solving $e_3^{\mathrm{LDPC}}(p) = \delta$ gives
\begin{equation*}
4p(1-p) = \delta+\sqrt{\delta(2-3\delta)}.
\end{equation*}
Therefore,
\begin{equation}
\label{eq:psc-ldpc-rate}
U_3^{\PSC}(\delta)
\coloneqq
\inf_{\substack{0 < p < \frac{1}{2}\\ e_3^{\mathrm{LDPC}}(p) < \delta}}
h\left(\frac{1}{2}-\sqrt{p(1-p)}\right)
=
h\left(\frac{1-\sqrt{\delta+\sqrt{\delta(2-3\delta)}}}{2}\right).
\end{equation}
Because $e_3^{\mathrm{LDPC}}(p) < p$, this constraint admits parameters $p$ beyond $\delta$, and the PSC Holevo information there is strictly smaller than its value at $p = \delta$; \cref{thm:psc-ldpc-improvement} records the resulting strict improvement over the first MRRW expression. Before proving the rate bound, we compare this curve analytically with \eqref{eq:sy-ldpc-base}. The estimates below are not fully optimized. Instead, they are chosen for ease of exposition. \cref{fig:psc-ldpc-curves} displays $U_3^{\PSC}$ against both benchmarks.

\begin{lemma}[Shangguan--Yang interval separation]
\label{lem:psc-sy-interval}
The pointwise inequality
\begin{equation*}
U_3^{\PSC}(u)
< B_3^{\mathrm{SY}}(u)
\end{equation*}
holds for $\frac{1}{6} \leq u < \frac{1}{2}$.
\end{lemma}

\begin{proof}
For $u \in [\frac{1}{6},\frac{1}{2})$, set
\begin{equation*}
q
\coloneqq
\frac{1-\sqrt{u+\sqrt{u(2-3u)}}}{2},
\qquad
r
\coloneqq
\frac{1}{2}-\sqrt{u(1-u)}.
\end{equation*}
By \eqref{eq:psc-ldpc-rate}, $U_3^{\PSC}(u) = h(q)$, while $r$ is the parameter $r_\delta$ in \eqref{eq:sy-ldpc-base}. We compare both curves through $h(r)$.

The key estimate is
\begin{equation}
\label{eq:psc-w3-gain}
r-q > r^2.
\end{equation}
To prove it, put $z \coloneqq 1-2u \in (0,\frac{2}{3}]$. Then $r(1-r) = \frac{z^2}{4}$, so \eqref{eq:psc-w3-gain} is equivalent to $q < \frac{z^2}{4}$. After substituting $u = \frac{1-z}{2}$ and squaring twice, this inequality reduces to
\begin{equation*}
1+2z-3z^2
>
\left(1+z-2z^2+\frac{z^4}{2}\right)^2.
\end{equation*}
Both squarings preserve the inequality because all quantities involved are positive. Finally,
\begin{align*}
&1+2z-3z^2
-\left(1+z-2z^2+\frac{z^4}{2}\right)^2\\
&\qquad
=z^3\left(4-5z-z^2+2z^3-\frac{z^5}{4}\right).
\end{align*}
For $0 < z \leq \frac{2}{3}$, the factor in parentheses satisfies
\begin{equation*}
4-5z-z^2+2z^3-\frac{z^5}{4}
\geq
4-\frac{10}{3}-\frac{4}{9}-\frac{8}{243}
= \frac{46}{243} > 0.
\end{equation*}
This proves \eqref{eq:psc-w3-gain}.

The assumption $u \geq \frac{1}{6}$ gives
\begin{equation*}
r
\leq
\frac{3-\sqrt{5}}{6}
< \frac{1}{7}.
\end{equation*}
In particular, $B_3^{\mathrm{SY}}(u)$ lies on the entropy branch of \eqref{eq:sy-ldpc-base}. Since $h'(t) = \log\frac{1-t}{t}$ decreases on $(0,\frac{1}{2})$, \eqref{eq:psc-w3-gain} yields
\begin{equation}
\label{eq:hrq}
h(r)-h(q)
=\int_q^r h'(t)\,dt
\geq(r-q)\log\frac{1-r}{r}
>r^2\log{6}.
\end{equation}
On the other hand,
\begin{equation}
\label{eq:hSY}
h(r)-B_3^{\mathrm{SY}}(u)
=\int_0^r\log\frac{1-t}{1-3t}\,dt
\leq(\log{\mathrm e})\int_0^r\frac{2t}{1-3t}\,dt
<\frac{7}{4}(\log{\mathrm e})\,r^2.
\end{equation}
Here we used $\log(1+x) \leq (\log{\mathrm e})x$ and $1-3t > \frac{4}{7}$ on the integration interval. Since $\log{6} > \frac{7}{4}\log{\mathrm e}$, subtracting \eqref{eq:hSY} from \eqref{eq:hrq} gives $U_3^{\PSC}(u) < B_3^{\mathrm{SY}}(u)$.
\end{proof}

We now proceed to prove~\cref{thm:psc-ldpc-improvement}.

\begin{proof}[Proof of \cref{thm:psc-ldpc-improvement}]
Fix a strictly feasible $p$ in \eqref{eq:psc-ldpc-rate} and a code $\cC$ counted by $A_{2,3}(n,\ceil{\delta n})$, with $n$ sufficiently large that $\ceil{\delta n} > 1$. By \cref{lem:psc-local-check}, every coordinate error of the block PGM satisfies $p_{\mathrm{e},i} \leq e_3^{\mathrm{LDPC}}(p)$, and summing these $n$ coordinate bounds gives
\begin{equation*}
\Eb[d(\bchat,\bc)]
\leq n e_3^{\mathrm{LDPC}}(p).
\end{equation*}
Since $e_3^{\mathrm{LDPC}}(p) < \delta$, the block-success step (\cref{cor:block-error}) and the fixed-channel strong converse (\cref{thm:cq-strong-converse}), combined exactly as in the proof of \cref{thm:main-result}, give
\begin{equation*}
\frac{1}{n}\log{\abs{\cC}}
\leq
h\left(\frac{1}{2}-\sqrt{p(1-p)}\right)
+O\left(\frac{1}{\sqrt{n}}\right),
\end{equation*}
where the implicit constant depends on the fixed $p,\delta$ but not on $n$ or $\cC$.
Maximizing over $\cC$, taking $n \to \infty$, and then taking the infimum over strictly feasible $p$ proves $R_3(\delta) \leq U_3^{\PSC}(\delta)$.

By \cref{lem:psc-local-check}, $e_3^{\mathrm{LDPC}}(\delta) < \delta$. Continuity makes some $p > \delta$ feasible in \eqref{eq:psc-ldpc-rate}. The PSC Holevo information decreases strictly on $(0,\frac{1}{2})$, so
\begin{equation*}
U_3^{\PSC}(\delta)
< h\left(\frac{1}{2}-\sqrt{\delta(1-\delta)}\right)
= \MRRW(\delta).
\end{equation*}
The comparison with $B_3^{\mathrm{SY}}$ on $\frac{1}{6} \leq \delta < \frac{1}{2}$ is contained in \cref{lem:psc-sy-interval}.
\end{proof}

The comparison with the explicit Shangguan--Yang curve therefore holds throughout $\frac{1}{6} \leq \delta < \frac{1}{2}$. Since this curve is their $w = 3$ bound throughout $\frac{2}{5} \leq \delta < \frac{1}{2}$, the theorem yields an improvement on the latter interval.

\begin{figure}[!htbp]
\centering
\begin{tikzpicture}
\begin{axis}[
    name=ldpcaxis,
    declare function={
        hb(\x)=-(max(\x,0.000000000001)*ln(max(\x,0.000000000001))
        +max(1-\x,0.000000000001)*ln(max(1-\x,0.000000000001)))/ln(2);
        rr(\x)=0.5-sqrt(max(0,\x*(1-\x)));
        sy(\x)=(rr(\x)>=0.25)*(2/3)
        +(rr(\x)<0.25)*(-(rr(\x)*ln(max(rr(\x),0.000000000001))
        +((1-3*rr(\x))/3)*ln(max(1-3*rr(\x),0.000000000001)))/ln(2));
        ldpc(\x)=hb((1-sqrt(max(0,\x+sqrt(max(0,\x*(2-3*\x))))))/2);
    },
    width=0.94\textwidth,
    height=0.50\textwidth,
    xmin=0,
    xmax=0.5,
    ymin=0,
    ymax=1.02,
    xtick={0,0.1,0.2,0.3,0.4,0.5},
    ytick={0,0.2,0.4,0.6,0.8,1},
    xlabel={Relative distance $\delta$},
    ylabel={Rate upper bound (bits)},
    grid=major,
    legend style={
        at={(0.015,0.015)},
        anchor=south west,
        draw=black!35,
        fill=white,
        fill opacity=0.94,
        text opacity=1,
        font=\scriptsize,
        legend columns=1
    },
    legend cell align={left},
    tick align=outside,
    clip=false
]
\addplot[thick,dashed,magenta,domain=0:0.5,samples=201]
    {hb(0.5-sqrt(x*(1-x)))};
\addlegendentry{First MRRW}
\addplot[thick,dashdotted,green!45!black,domain=0:0.5,samples=201]
    {sy(x)};
\addlegendentry{Shangguan--Yang $B_3^{\mathrm{SY}}$}
\addplot[very thick,blue,domain=0:0.5,samples=301]
    {ldpc(x)};
\addlegendentry{One-check PSC $U_3^{\PSC}$}
\coordinate (ldpcboxnw) at (axis cs:0.2875,0.280);
\coordinate (ldpcboxne) at (axis cs:0.3125,0.280);
\draw[thin,black!60] (axis cs:0.2875,0.186) rectangle (axis cs:0.3125,0.280);
\end{axis}
\begin{axis}[
    name=ldpcinset,
    at={(ldpcaxis.south west)},
    anchor=south west,
    xshift=8.70cm,
    yshift=3.25cm,
    declare function={
        hb(\x)=-(max(\x,0.000000000001)*ln(max(\x,0.000000000001))
        +max(1-\x,0.000000000001)*ln(max(1-\x,0.000000000001)))/ln(2);
        rr(\x)=0.5-sqrt(max(0,\x*(1-\x)));
        sy(\x)=(rr(\x)>=0.25)*(2/3)
        +(rr(\x)<0.25)*(-(rr(\x)*ln(max(rr(\x),0.000000000001))
        +((1-3*rr(\x))/3)*ln(max(1-3*rr(\x),0.000000000001)))/ln(2));
        ldpc(\x)=hb((1-sqrt(max(0,\x+sqrt(max(0,\x*(2-3*\x))))))/2);
    },
    width=0.36\textwidth,
    height=0.20\textwidth,
    xmin=0.2875,
    xmax=0.3125,
    ymin=0.186,
    ymax=0.280,
    xtick={0.290,0.300,0.310},
    xticklabels={$0.290$,$0.300$,$0.310$},
    ytick={0.200,0.220,0.240,0.260},
    yticklabels={$0.200$,$0.220$,$0.240$,$0.260$},
    scaled ticks=false,
    tick label style={font=\tiny},
    title={Near $\delta = \frac{3}{10}$},
    title style={font=\scriptsize,yshift=-1pt},
    grid=major,
    grid style={black!12},
    axis background/.style={fill=white},
    axis line style={black!55},
    tick align=outside,
    clip=true
]
\addplot[thick,dashed,magenta,domain=0.2875:0.3125,samples=101]
    {hb(0.5-sqrt(x*(1-x)))};
\addplot[thick,dashdotted,green!45!black,domain=0.2875:0.3125,samples=101]
    {sy(x)};
\addplot[very thick,blue,domain=0.2875:0.3125,samples=201]
    {ldpc(x)};
\addplot[thin,densely dotted,black!70] coordinates {
    (0.300,0.186)
    (0.300,0.280)
};
\end{axis}
\draw[thin,black!45] (ldpcboxnw) -- (ldpcinset.south west);
\draw[thin,black!45] (ldpcboxne) -- (ldpcinset.south east);
\end{tikzpicture}
\caption{Rate upper bounds for binary linear codes whose duals are generated by checks of weight at most $3$. The first MRRW curve~\cite{MRRW77} applies to every binary code, whereas $B_3^{\mathrm{SY}}$ and $U_3^{\PSC}$ apply to this sparse-check class. By \cref{thm:psc-ldpc-improvement,lem:psc-sy-interval}, $U_3^{\PSC}$ lies strictly below the first MRRW bound for $0 < \delta < \frac{1}{2}$ and below $B_3^{\mathrm{SY}}$ for $\frac{1}{6} \leq \delta < \frac{1}{2}$; numerically, the latter separation begins near $\delta \approx 0.0705$. At $\delta = \frac{3}{10}$, we have $U_3^{\PSC} \approx 0.2065$, $B_3^{\mathrm{SY}} \approx 0.2476$, and first MRRW $\approx 0.2502$.}
\label{fig:psc-ldpc-curves}
\end{figure}
\FloatBarrier

\bibliographystyle{alpha}
\bibliography{cq_channels_ref}

\appendix
\crefalias{subsection}{appendix}

\section{Proofs for the preliminaries}
\label{app:preliminary-proofs}

In this appendix, for the sake of completeness, we present the non-immediate proofs deferred from \cref{sec:preliminaries}, particularly the PGM error data processing inequality (\cref{lem:pgm-data-processing}) and the cq channel coding strong converse (\cref{thm:cq-strong-converse}), as those are the two crucial inequalities needed in the proof of the pretty good criterion (\cref{thm:main-result}).

To this end, we recall some necessary quantum notions and tools. Building on the notion of a density matrix, a $d \times d$ matrix is said to be a \emph{subnormalized state} if it is a positive semidefinite matrix whose trace is at most $1$. For any finite-dimensional Hilbert space $V$, let $\cM(V)$ denote the vector space of all Hermitian matrices over $V$.

Next, we recall some standard matrix analysis tools needed in the proof of the cq channel coding strong converse in~\cref{appsubsec:proof-fixed-channel-strong-converse}. We denote the trace norm of a $d \times d$ complex-valued matrix $X$ by $\norm{X}_1 \coloneqq \Tr\sqrt{X^\dagger X}$, the Frobenius norm by $\norm{X}_2 \coloneqq \sqrt{\Tr(X^\dagger X)}$, and the spectral norm by $\norm{X}_\infty \coloneqq \sqrt{\lambda_{\max}(X^\dagger X)}$. Throughout, matrix inequalities $A \preceq B$ between Hermitian matrices mean that $B-A$ is positive semidefinite. Note that for diagonal density matrices, i.e., classical probability distributions, the quantity $\frac{1}{2}\norm{\tau-\tau'}_1$ is the total-variation distance between the probability distributions $\tau$ and $\tau'$.

We now recall a standard inequality that is needed for the proof of the cq channel coding strong converse (\cref{appsubsec:proof-fixed-channel-strong-converse}), which is the gentle measurement lemma, stated in~\cite[Lemma~9.4.2 and Exercise~9.4.1]{Wil21}. The gentle measurement lemma controls how much a test that passes with high probability disturbs such a state. Classically, conditioning on an event of probability at least $1-\eps$ moves a distribution by $O(\eps)$ in total variation. In the quantum case, that estimate becomes $O(\sqrt{\eps})$.

\begin{lemma}[Gentle measurement lemma]
\label{lem:gentle-measurement}
If $A$ is a subnormalized $d \times d$ state and $P$ is a projection $d \times d$ matrix, then
\begin{equation*}
\norm{A-PAP}_1
\leq 2\sqrt{\Tr[A(I-P)]}.
\end{equation*}
\end{lemma}
\begin{proof}
Since $I-P$ and $P$ are both projection matrices, by the triangle inequality, the matrix Cauchy--Schwarz inequality~\cite[Corollary IV.2.6]{Bha97}, and the inequalities $\Tr(A) \leq 1$ and $\max\{\Tr(PA),\Tr((I-P)A)\} \leq \Tr(A)$ (as $A$ is a subnormalized state), we find that
\begin{align*}
\norm{A-PAP}_1 &= \norm{(I-P)A+PA(I-P)}_1\\
&\leq \norm{(I-P)A}_1 +\norm{PA(I-P)}_1 \\
&\leq \norm{(I-P)A^{1/2}}_2 \cdot \norm{A^{1/2}}_2 + \norm{PA^{1/2}}_2 \cdot \norm{A^{1/2}(I-P)}_2\\
&= \sqrt{\Tr[A(I-P)]} \cdot \sqrt{\Tr(A)} + \sqrt{\Tr(PA)} \cdot \sqrt{\Tr[A(I-P)]} \\
&\leq \sqrt{\Tr[A(I-P)]} \cdot 1 + 1 \cdot \sqrt{\Tr[A(I-P)]} \\
&= 2\sqrt{\Tr[A(I-P)]}. \qedhere
\end{align*}
\end{proof}

\subsection{Petz recovery and its structural properties}
\label{appsubsec:petz-recovery-properties}

In this subsection, we develop the Petz recovery map~\cite{Pet84,Pet86,Pet88,Ren22}, a generalization of the PGM, and prove its structural properties (\cref{fact:petz-properties}), en route to proving the PGM data-processing lemma (\cref{lem:pgm-data-processing}) in~\cref{appsubsec:proof-pgm-data-processing}. A Petz map reverses an arbitrary quantum channel relative to a chosen reference state. When the channel is a cq preparation channel and the reference is its classical prior, the Petz map reduces to the PGM~\cite{BK02}. This identification explains the connection with the PGM.

As a brief digression, recent work uses Petz maps and their variants to quantify approximate recovery in strengthened data-processing inequalities~\cite{STH16,JRSWW18}, to construct explicit decoders for quantum communication~\cite{BDL16}, to reconstruct holographic entanglement wedges~\cite{CPS20}, and to design quantum algorithms for Petz channels and PGMs~\cite{GLMQW22}.

Returning to the subject matter, we now formally introduce the Petz map.

\begin{definition}[Petz recovery map]
\label{def:petz-recovery}
Let $\Phi : \C^{d_A \times d_A} \to \C^{d_B \times d_B}$ be a quantum channel, and let $\pi \in \Dens(\C^{d_A})$ be a reference state. Let $\nu \coloneqq \Phi(\pi) \in \Dens(\C^{d_B})$. The \textbf{Petz recovery map} of $\Phi$ relative to $\pi$ is the linear map $\Phi_\pi^\sharp : \C^{d_B \times d_B} \to \C^{d_A \times d_A}$ defined as
\begin{equation*}
\Phi_\pi^\sharp(Y)
\coloneqq
\pi^{1/2}\Phi^*\left(\nu^{-1/2}Y\nu^{-1/2}\right)\pi^{1/2},
\end{equation*}
where $\Phi^*$ is determined by the identity $\Tr[X^\dagger\Phi^*(Y)] = \Tr[\Phi(X)^\dagger Y]$, and inverse powers act on the indicated supports. Moreover, for $A_1,A_2 \in \cM(\supp(\pi))$, define the weighted inner product
\begin{equation*}
\inner{A_1}{A_2}_\pi
\coloneqq
\Tr\left(A_1^\dagger\pi^{-1/2}A_2\pi^{-1/2}\right).
\end{equation*}
Since $\inner{A}{A}_\pi = \norm{\pi^{-1/4}A\pi^{-1/4}}_2^2$, this weighted form is therefore positive definite over $\cM(\supp(\pi))$.
\end{definition}

\begin{remark}[Petz in the classical case]
\label{rem:classical-petz-recovery}
Let $W(y \mid x)$ be any classical channel with input and© output alphabets $\cX$ and $\cY$, let $\pi(x)$ be an input distribution, and let $\nu(y) \coloneqq \sum_{x \in \cX}{\pi(x)W(y \mid x)}$ be the distribution of $\pi$ when passed through $W$. For every output $y$ with $\nu(y) > 0$, the Petz recovery map is the Bayes reverse channel
\begin{equation*}
W_\pi^\sharp(x \mid y)
\coloneqq \frac{\pi(x)W(y \mid x)}{\nu(y)}
= \Prb[\bx = x \mid \by = y].
\end{equation*}
Thus, Petz recovery is the quantum form of reversing a channel by Bayes' rule relative to a fixed prior.
\end{remark}

The proof of PGM data processing in~\cref{appsubsec:proof-pgm-data-processing} uses the following three structural properties of this recovery map.

\begin{fact}[Petz adjointness, functoriality, and contraction]
\label{fact:petz-properties}
Following the notation given in~\cref{def:petz-recovery}, when restricted to $\cM(\supp(\nu))$, the Petz recovery map $\Phi_\pi^\sharp$ is a quantum channel with output density matrices residing in $\cM(\supp(\pi))$. Moreover, the following three properties hold:
\begin{enumerate}
\item\label{item:petz-adjointness} \emph{(Weighted adjointness).} For all $A \in \cM(\supp(\pi))$ and $B \in \cM(\supp(\nu))$, we have the identity
\begin{equation*}
\inner{\Phi(A)}{B}_\nu
= \inner{A}{\Phi_\pi^\sharp(B)}_\pi.
\end{equation*}
\item\label{item:petz-functoriality} \emph{(Contravariant functoriality).} For any quantum channel $\Psi : \C^{d_B \times d_B} \to \C^{d_C \times d_C}$, we have the composition identity
\begin{equation*}
(\Psi\circ\Phi)_\pi^\sharp
= \Phi_\pi^\sharp\circ\Psi_\nu^\sharp.
\end{equation*}
\item\label{item:petz-contraction} \emph{(Petz--Gram contraction).} For all $A_1,A_2 \in \cM(\supp(\pi))$, the quantum channel $K \coloneqq \Phi_\pi^\sharp\circ\Phi$ satisfies
\begin{equation*}
\inner{A_1}{K(A_2)}_\pi
= \inner{\Phi(A_1)}{\Phi(A_2)}_\nu.
\end{equation*}
In particular, for any $A \in \C^{d_A \times d_A}$, we have
\begin{equation}
0 \leq \inner{A}{K(A)}_\pi
\leq \inner{A}{A}_\pi.
\end{equation}
\end{enumerate}
\end{fact}

\begin{proof}[Proof of \cref{fact:petz-properties}]
For every $B \in \cM(\supp(\nu))$, from~\cref{def:density-matrices}, we have that $B\Pi_\nu = B$. Combining this observation with the cyclicity of the trace and the definition of $\Phi^*$ in~\cref{def:petz-recovery}, we therefore find that
\begin{align*}
\Tr(\Phi_\pi^\sharp(B))
&= \Tr\left[\Phi^*\left(\nu^{-1/2}B\nu^{-1/2}\right)\pi\right] \\
&= \Tr\left[\nu^{-1/2}B\nu^{-1/2}\Phi(\pi)\right] \\
&= \Tr\left[B\nu^{-1/2}\Phi(\pi)\nu^{-1/2}\right] \\
&= \Tr(B\Pi_\nu) \\
&= \Tr(B).
\end{align*}
Hence, $\Phi_\pi^\sharp$ is trace-preserving on $\cM(\supp(\nu))$. Since the maps $B \mapsto \nu^{-1/2} B \nu^{-1/2}$, $A \mapsto \pi^{1/2} A \pi^{1/2}$, and $\Phi^*$ are all completely positive,\footnote{To see the complete positivity of $\Phi^*$, adjoint the Kraus representation of $\Phi$.} the closure of completely positive maps under composition therefore shows that $\Phi_\pi^\sharp$ is completely positive. Thus, $\Phi_\pi^\sharp$ is a quantum channel when restricted to $\cM(\supp(\nu))$. The fact that its output states lie in $\cM(\supp(\pi))$ is easily seen from the definition of $\Phi_\pi^\sharp$.

For \cref{fact:petz-properties}\ref{item:petz-adjointness}, by the defining identity for $\Phi^*$ given in~\cref{def:petz-recovery} and the cyclicity property of the trace map, we have that
\begin{align*}
\inner{A}{\Phi_\pi^\sharp(B)}_\pi 
&= \Tr\left[A^\dagger \pi^{-1/2}\pi^{1/2} \Phi^*\left(\nu^{-1/2}B\nu^{-1/2}\right)\pi^{1/2}\pi^{-1/2}\right] \\
&= \Tr\left[\Pi_\pi A^\dagger \Pi_\pi \Phi^*\left(\nu^{-1/2}B\nu^{-1/2}\right)\right] \\
&= \Tr\left[\Phi(A)^\dagger\nu^{-1/2}B\nu^{-1/2}\right] \\
&= \inner{\Phi(A)}{B}_\nu.
\end{align*}

For \cref{fact:petz-properties}\ref{item:petz-functoriality}, define $\mu \coloneqq \Psi(\nu)$. For $C \in \Dens(\supp(\mu))$, applying the preceding adjointness identity twice gives us
\begin{equation*}
\inner{(\Psi\circ\Phi)(A)}{C}_\mu
= \inner{\Phi(A)}{\Psi_\nu^\sharp(C)}_\nu
= \inner{A}{\Phi_\pi^\sharp\Psi_\nu^\sharp(C)}_\pi,
\end{equation*}
which proves the functoriality formula by uniqueness of the adjoint. Note that the adjoint map, in our case, extends from $\cM(\supp(\mu))$ to $\cM(\C^{d_C})$ by zero.

It remains to prove \cref{fact:petz-properties}\ref{item:petz-contraction}. To this end, we first show that $K$ is positive semidefinite with respect to the $\pi$-weighted inner product $\inner{\cdot}{\cdot}_\pi$ from~\cref{def:petz-recovery}, which enables us to invoke the (finite-dimensional) spectral theorem. Indeed, weighted adjointness~\ref{item:petz-adjointness} tells us that
\begin{equation*}
\inner{A_1}{K(A_2)}_\pi
= \inner{\Phi(A_1)}{\Phi(A_2)}_\nu.
\end{equation*}
In particular, $\inner{A}{K(A)}_\pi = \inner{\Phi(A)}{\Phi(A)}_\nu \geq 0$. Moreover, conjugate symmetry gives us
\begin{equation*}
\inner{K(A)}{B}_\pi
= \overline{\inner{B}{K(A)}_\pi}
= \overline{\inner{\Phi(B)}{\Phi(A)}_\nu}
= \inner{\Phi(A)}{\Phi(B)}_\nu
= \inner{A}{K(B)}_\pi.
\end{equation*}
Thus, $K$ is positive semidefinite with respect to the $\pi$-weighted inner product. By the spectral theorem with respect to $\inner{\cdot}{\cdot}_\pi$, we can find a weighted-eigenbasis $\{E_j\}_{j=1}^m$ satisfying $\inner{E_j}{E_j}_\pi = 1$ and $K(E_j) = \lambda_jE_j$. Positivity gives $\lambda_j = \inner{E_j}{K(E_j)}_\pi \geq 0$.

We next show that every eigenvalue of $K$ lies in $[0,1]$. Indeed, as a composition of support-restricted quantum channels, $K$ is a quantum channel from $\cM(\supp(\pi))$ to itself. Let $K^*$ denote its adjoint with respect to the trace inner product on $\cM(\supp(\pi))$, defined by the identity $\Tr[A_1^\dagger K^*(A_2)] = \Tr[K(A_1)^\dagger A_2]$. Now, notice that the trace preservation of $K$ implies that (in fact, equivalent) $K^*$ is unital, i.e., $K^*(\Pi_\pi) = \Pi_\pi$. Indeed, for every $A \in \cM(\supp(\pi))$,
\begin{equation*}
\Tr\left[A^\dagger K^*(\Pi_\pi)\right]
= \Tr\left[K(A)^\dagger \Pi_\pi\right]
= \overline{\Tr(K(A))}
= \overline{\Tr(A)}
= \Tr(A^\dagger \Pi_\pi).
\end{equation*}
The nondegeneracy of the trace inner product over $\cM(\supp(\pi))$ thus implies that $K^*(\Pi_\pi) = \Pi_\pi$. Since $K$ is a quantum channel over $\cM(\supp(\pi))$, by the Stinespring representation~\cite[Chapter 8]{NC10}, there exist a positive integer $k$ and an isometry $V : \supp(\pi) \to \supp(\pi) \otimes \C^k$ such that we can write $K(A) = \Tr_E(V A V^\dagger)$, where $\Tr_E : \cM(\supp(\pi)) \otimes \cM(\C^k) \to \cM(\supp(\pi))$ is the partial trace map, i.e., the map $A \otimes M \mapsto A\Tr(M)$. By uniqueness of adjoints, we can therefore write $K^*(X)=V^\dagger(X \otimes I_k)V$ for all $X \in \cM(\supp(\pi))$. Since $VV^\dagger \preceq \Pi_\pi \otimes I_k$, we have
\begin{align*}
K^*(A)^\dagger K^*(A)
&= V^\dagger(A^\dagger \otimes I_k)VV^\dagger(A \otimes I_k)V \\
&\preceq V^\dagger(A^\dagger A \otimes I_k)V \\
&= K^*(A^\dagger A) \\
&\preceq K^*(\norm{A}_\infty^2\Pi_\pi) \\
&\preceq \norm{A}_\infty^2K^*(\Pi_\pi) \\
&= \norm{A}_\infty^2\Pi_\pi,
\end{align*}
where the last two inequalities use the positivity and unitality of $K^*$.\footnote{This argument closely follows the one presented in~\cite[Section~5.2, Equations~(5.2)--(5.3)]{Wol12}.} Taking spectral norms thus gives us $\norm{K^*(A)}_\infty \leq \norm{A}_\infty$. Hence, if $K^*(A) = \mu A$ with $A \neq 0$, then $|\mu| \leq 1$. Since the eigenvalues of $K$ are simply conjugates of those of $K^*$, the preceding spectral bound then gives us $\lambda_j \leq 1$. Thus, every eigenvalue of $K$ lies in $[0,1]$.

Now, we can conclude our proof. Writing $A = \sum_{j=1}^m{a_jE_j}$, we thus see that
\begin{equation*}
0 \leq \inner{A}{K(A)}_\pi
= \sum_{j = 1}^m{\lambda_j\abs{a_j}^2}
\leq \sum_{j = 1}^m{\abs{a_j}^2}
= \inner{A}{A}_\pi.\qedhere
\end{equation*}
\end{proof}

\subsection{\texorpdfstring{\cref{lem:pgm-data-processing}}{Lemma \getrefnumber{lem:pgm-data-processing}}: PGM data processing}
\label{appsubsec:proof-pgm-data-processing}

In this subsection, using the tools we developed in~\cref{appsubsec:petz-recovery-properties}, we prove \cref{lem:pgm-data-processing}. Throughout this subsection, for an ensemble $\cE = \{(\pi_x,\sigma_x) : x \in \cX\}$ with PGM observables $\{M_x\}_{x \in \cX}$, we adopt the shorthand $p_{\mathrm{s}}(\cE) \coloneqq \sum_{x \in \cX}{\pi_x\Tr(\sigma_xM_x)} = 1-p_{\mathrm{e}}(\cE)$ for its PGM success probability.

\begin{proof}[Proof of \cref{lem:pgm-data-processing}]
Define $A_x \coloneqq \pi_x \sigma_x$. Expanding the PGM success formula using the definition of PGM (\cref{def:pgm}), observe that
\begin{equation*}
p_{\mathrm{s}}(\cE)
= \sum_{x \in \cX}{\pi_x^2\Tr\left(\sigma_x\bar{\sigma}^{-1/2}\sigma_x\bar{\sigma}^{-1/2}\right)}
= \sum_{x \in \cX}{\inner{A_x}{A_x}_{\bar{\sigma}}}.
\end{equation*}
For $x \in \cX$, let $\tau_x \coloneqq \N(\sigma_x)$. By linearity of $\N$, we see that the average post-processed state satisfies $\bar{\tau} \coloneqq \sum_{x \in \cX}{\pi_x\tau_x} = \sum_{x \in \cX}{\pi_x\N(\sigma_x)} = \N(\bar{\sigma})$. Now, for every $x \in \cX$, by applying the weighted adjointness~\ref{item:petz-adjointness} and contraction~\ref{item:petz-contraction} of the Petz map from~\cref{fact:petz-properties}, we see that
\begin{equation*}
\inner{\N(A_x)}{\N(A_x)}_{\bar{\tau}}
= \inner{A_x}{(\N_{\bar{\sigma}}^\sharp \circ \N)(A_x)}_{\bar{\sigma}}
\leq \inner{A_x}{A_x}_{\bar{\sigma}}.
\end{equation*}
Summing this inequality and applying the same PGM expansion to the processed ensemble thus shows that
\begin{equation*}
p_{\mathrm{s}}(\cE)
= \sum_{x \in \cX}{\inner{A_x}{A_x}_{\bar{\sigma}}}
\geq \sum_{x \in \cX}{\inner{\N(A_x)}{\N(A_x)}_{\bar{\tau}}}
= p_{\mathrm{s}}\left(\{(\pi_x,\N(\sigma_x)) : x \in \cX\}\right).
\end{equation*}
Since $p_{\mathrm{e}} = 1-p_{\mathrm{s}}$, it therefore follows that
\begin{equation*}
p_{\mathrm{e}}(\cE) \leq p_{\mathrm{e}}\left(\{(\pi_x,\N(\sigma_x)) : x \in \cX\}\right).\qedhere
\end{equation*}
\end{proof}

\begin{remark}[Alternative proof via R\'enyi divergence]
\label{rem:pgm-data-processing-alt-proof}
In the classical case,~\cref{lem:pgm-data-processing} admits a more natural proof by recasting $p_{\mathrm{s}}(\cE)$ as a conditional collision entropy and invoking Jensen's inequality. Such a proof does extend to the quantum setting via the data processing inequality for the sandwiched order-$2$ quantum R\'enyi divergence~\cite{MLDSFT13,WWY14}, established independently by Frank and Lieb~\cite{FL13} and by Beigi~\cite{Bei13}. For completeness and greater self-containment, we presented this alternative proof via the Petz map. In this formulation, the Petz--Gram inequality in \cref{fact:petz-properties}\ref{item:petz-contraction} serves as the noncommutative analog of the conditional-Jensen step.
\end{remark}

\subsection{\texorpdfstring{\cref{prop:output-symmetric-capacity}}{Proposition \getrefnumber{prop:output-symmetric-capacity}}: Capacity of a binary-input output-symmetric cq channel}
\label{appsubsec:proof-output-symmetric-capacity}

In this subsection, we prove \cref{prop:output-symmetric-capacity}.

\begin{proof}[Proof of \cref{prop:output-symmetric-capacity}]
Fix any positive integer $n$, and let $\pi^{(n)}$ be any probability distribution on $\bits^n$. Write $\chi(\pi^{(n)};\Phi_\sigma^{\otimes n})$ for the Holevo information of the ensemble that assigns probability $\pi^{(n)}(x^n)$ to $\sigma_{x_1} \otimes \cdots \otimes \sigma_{x_n}$. Define the $n$-partite state
\begin{equation*}
\omega_{B^n} \coloneqq \sum_{x^n \in \bits^n}{\pi^{(n)}(x^n)\sigma_{x_1} \otimes \cdots \otimes \sigma_{x_n}},
\end{equation*}
and let $\omega_{B_i} \coloneqq \pi_i(0)\sigma_0+\pi_i(1)\sigma_1$ be its $i$'th subsystem, where $\pi_i$ is the $i$'th-coordinate marginal of $\pi^{(n)}$. By the subadditivity of entropy, we see that
\begin{align}
\chi(\pi^{(n)};\Phi_\sigma^{\otimes n})
&= S(\omega_{B^n})-\sum_{i = 1}^n{\left(\pi_i(0)S(\sigma_0)+\pi_i(1)S(\sigma_1)\right)} \notag \\
&\leq \sum_{i = 1}^n{\left(S(\omega_{B_i})-\pi_i(0)S(\sigma_0)-\pi_i(1)S(\sigma_1)\right)} \notag \\
&= \sum_{i = 1}^n{\chi(\pi_i;\Phi_\sigma)}. \label{eq:prop-6-1}
\end{align}
Now, define $f(t) \coloneqq \chi(\{(t,\sigma_0),(1-t,\sigma_1)\})$ for $t \in [0,1]$. Notice that the concavity of von Neumann entropy implies that $f$ is concave, while output-symmetry of $\Phi_\sigma$ implies that $f(t) = f(1-t)$. Hence
\begin{equation*}
\chi(\{(t,\sigma_0),(1-t,\sigma_1)\}) = f(t) = \frac{1}{2}f(t)+\frac{1}{2}f(1-t) \leq f\left(\frac{1}{2}\right) = \chi(\sigma_0,\sigma_1).
\end{equation*}
Therefore, for every $i \in \{1, \ldots , n\}$, we have that
\begin{equation}
\label{eq:prop-6-2}
\chi(\pi_i;\Phi_\sigma) \leq \chi(\sigma_0,\sigma_1).
\end{equation}
Combining~\eqref{eq:prop-6-1} and~\eqref{eq:prop-6-2} thus gives us $\chi(\pi^{(n)};\Phi_\sigma^{\otimes n}) \leq n\chi(\sigma_0,\sigma_1)$. Since the uniform product input attains equality by additivity of entropy, we thus see that
\begin{equation*}
\frac{1}{n}\max_{\pi^{(n)}}\chi(\pi^{(n)};\Phi_\sigma^{\otimes n})
= \chi(\sigma_0,\sigma_1).
\end{equation*}
As $n \to \infty$, the left-hand side becomes $C(\Phi_\sigma)$ by \cref{def:holevo-capacity}. This completes the proof.\qedhere
\end{proof}

\subsection{\texorpdfstring{\cref{thm:cq-strong-converse}}{Theorem \getrefnumber{thm:cq-strong-converse}}: The cq channel coding strong converse}
\label{appsubsec:proof-fixed-channel-strong-converse}

In this subsection, we prove the cq channel coding strong converse.\footnote{For context, the classical strong converse dates back to Wolfowitz~\cite{Wol57}, with exponential versions following later refinements~\cite{Str64,Ari73}.} To facilitate the exposition, we recall for the reader the proof for the classical BSC case. Fix $p\in(0,\frac{1}{2})$, and suppose that $\{\cD_c : c \in \cC\}$ are the decoding regions of a deterministic decoder for $\BSC(p)$. If $\bc$ is sent and the noise is $\bz \in \bits^n$, then the decoder succeeds exactly when $\bc \oplus \bz \in \cD_c$. Now, define the function
\begin{equation*}
t_n(z)
\coloneqq
\frac{1}{\abs{\cC}}
\abs{\{c\in\cC:c\oplus z\in\cD_c\}}.
\end{equation*}
If $\bz\sim\Ber(p)^n$ is the BSC noise and $\by$ is uniform on $\bits^n$, then
\begin{align*}
\Eb[t_n(\bz)]
&=\Prb[\bchat=\bc], \\
\Eb[t_n(\by)]
&=\frac{1}{\abs{\cC}}.
\end{align*}
The second identity holds because translation preserves the uniform distribution and the decoding regions partition $\bits^n$. Now, define the one-sided typical set
\begin{equation*}
\mathcal{Q}_\eps
\coloneqq
\left\{z\in\bits^n:
-\log{\Prb[\be=z]}\geq n[h(p)-\eps]
\right\},
\end{equation*}
where $\be \sim \Ber(p)^n$. For an independent $\bz \sim \Ber(p)^n$, notice that $\Eb_\bz\left[-\log{\Prb_\be[\be = \bz]}\right] = nh(p)$. Since $-\log{\Prb_\be[\be = \bz]} = \sum_{i=1}^n{-\log{\Prb_{\be_i}[\be_i = \bz_i]}}$, each of which is independent, applying Hoeffding's inequality on $\bz$ therefore gives
\begin{equation*}
\Prb[\bz\notin\mathcal{Q}_\eps]
\leq
\exp\left(
-\frac{2\eps^2n}{[\log((1-p)/p)]^2}
\right).
\end{equation*}
For every $z\in\mathcal{Q}_\eps$, the definition of the set and the identity $\Prb[\by=z]=2^{-n}$ give us the likelihood ratio bound
\begin{equation*}
\Prb[\bz=z]
\leq
2^{-n[h(p)-\eps]}
=
2^{n[1-h(p)+\eps]}\Prb[\by=z].
\end{equation*}
Since $0\leq t_n(z)\leq1$, we can split the success probability according to whether the noise is typical or not and apply this comparison term by term. That is,
\begin{align*}
\Prb[\bchat=\bc]
&=\sum_{z\in\bits^n}{\Prb[\bz=z]t_n(z)} \\
&\leq
\sum_{z\in\mathcal{Q}_\eps}{\Prb[\bz=z]t_n(z)}
+\Prb[\bz\notin\mathcal{Q}_\eps] \\
&\leq
2^{n[1-h(p)+\eps]}
\sum_{z\in\mathcal{Q}_\eps}{\Prb[\by=z]t_n(z)}
+\Prb[\bz\notin\mathcal{Q}_\eps] \\
&\leq
2^{n[1-h(p)+\eps]}\Eb[t_n(\by)]
+\Prb[\bz\notin\mathcal{Q}_\eps] \\
&\leq
2^{-n[R(\cC)-(1-h(p))-\eps]}
+
\exp\left(
-\frac{2\eps^2n}{[\log((1-p)/p)]^2}
\right).
\end{align*}
Thus, whenever $R(\cC)>1-h(p)$ by a fixed amount, choosing $\eps$ smaller than that gap makes both terms exponentially small. This finishes the BSC strong converse proof.

The cq channel proof repeats these three steps. Output symmetry plays a similar role as XOR in the BSC case, producing a matrix $T_n$ that plays the role of $t_n$. Its expectations under $\sigma_0^{\otimes n}$ and $\bar{\sigma}^{\otimes n}$ reproduce the two identities above. Finally, two spectral cutoffs replace the BSC typical set---one cutoff to say that the probability of $\bz$ (resp. $\sigma_0^{\otimes n}$) is not too large, and another to say that the probability of $\by$ (resp. $\bar{\sigma}^{\otimes n}$) is not too small---and provide the required matrix version of the likelihood ratio bound.

\begin{proof}[Proof of \cref{thm:cq-strong-converse}]
Recall that each decoder matrix $D_c$ is designed to recognize its corresponding codeword state $\sigma_c$. Let $U$ be the unitary exchanging $\sigma_0$ and $\sigma_1$, and, for $c\in\cC$, set $V_c\coloneqq U^{c_1}\otimes\cdots\otimes U^{c_n}$. Then we see that $V_c\sigma_0^{\otimes n}V_c^\dagger=\sigma_c$ and $V_c\bar{\sigma}^{\otimes n}V_c^\dagger=\bar{\sigma}^{\otimes n}$. Following the BSC definition of $t_n$, define the matrix
\begin{equation*}
T_n \coloneqq \frac{1}{\abs{\cC}}\sum_{c \in \cC}{V_c^\dagger D_cV_c}.
\end{equation*}
Since every $D_c$ satisfies $0\preceq D_c\preceq I$, their average satisfies $0\preceq T_n\preceq I$. The analogs of the two BSC expectations are thus
\begin{align}
\Tr(\sigma_0^{\otimes n}T_n)
&=\frac{1}{\abs{\cC}}\sum_{c\in\cC}{\Tr(\sigma_0^{\otimes n}V_c^\dagger D_cV_c)}
=\frac{1}{\abs{\cC}}\sum_{c\in\cC}{\Tr(\sigma_cD_c)}
=\Prb[\bchat=\bc], \notag\\
\Tr(\bar{\sigma}^{\otimes n}T_n)
&=\frac{1}{\abs{\cC}}\sum_{c\in\cC}{\Tr(\bar{\sigma}^{\otimes n}V_c^\dagger D_cV_c)}
=\frac{1}{\abs{\cC}}\sum_{c\in\cC}{\Tr(\bar{\sigma}^{\otimes n}D_c)}
=\frac{1}{\abs{\cC}}.
\label{eq:symmetrized-decoder}
\end{align}
For the second identity, we used the invariance of $\bar{\sigma}^{\otimes n}$ under every $V_c$ and the POVM identity $\sum_{c\in\cC}{D_c}=I$. Thus, $T_n$ is exactly the matrix counterpart of $t_n$.

Let $\lambda_{\min}$ be the smallest positive eigenvalue occurring in either of $\sigma_0$ or $\bar{\sigma}$. The channel-dependent constant in the theorem statement is defined to be
\begin{equation*}
L_\sigma \coloneqq \max\{1,\log(1/\lambda_{\min})\}.
\end{equation*}

The first cutoff removes eigenvalues of $\sigma_0^{\otimes n}$ that are atypically large, so as to make the bulk of its spectrum ``spread out.'' To expose the classical random variable underneath this step, choose an eigenbasis of $\sigma_0$ with positive eigenvalues $\lambda_j$. Measuring $\sigma_0$ in this basis returns $j$ with probability $\lambda_j$, so the ``surprisal'' $-\log{\lambda_j}$ lies in $[0,L_\sigma]$ and has mean $S(\sigma_0)$. Let $Q$ be the spectral projector of $\sigma_0^{\otimes n}$ onto the positive eigenvalues $\lambda$ satisfying $-\log{\lambda} \geq n[S(\sigma_0)-\eps]$. Hoeffding's inequality~\cite{Hoe63}, applied to the sum of $n$ independent terms, gives
\begin{equation}
\label{eq:spectral-truncation-q}
\Tr(\sigma_0^{\otimes n}Q) \geq 1-\exp(-2\eps^2n/L_\sigma^2).
\end{equation}
The second cutoff removes eigenspaces on which $\bar{\sigma}^{\otimes n}$ is atypically small. Note that the BSC did not need to perform this, since the distribution corresponding to $\bar{\sigma}^{\otimes n}$ was the uniform distribution over $\bits^n$, which evidently satisfies this property. Now, to see the relevant probability distribution, choose an eigenbasis $\{\ket{j}\}$ of $\bar{\sigma}$, write $\mu_j > 0$ for the eigenvalue associated with $\ket{j}$, and measure $\sigma_0$ in this basis. The outcome $j$ has probability $\bra{j}\sigma_0\ket{j}$, so $-\log{\mu_j}$ lies in $[0,L_\sigma]$ and has mean $-\Tr(\sigma_0\log{\bar{\sigma}}) = S(\bar{\sigma})$. Let $P$ be the spectral projector of $\bar{\sigma}^{\otimes n}$ onto the eigenvalues $\mu$ satisfying $-\log{\mu} \leq n[S(\bar{\sigma})+\eps]$. For $n$ copies, the product-basis measurement gives $n$ independent copies of this random variable, so a second application of Hoeffding's inequality yields
\begin{equation}
\label{eq:spectral-truncation-p}
\Tr(\sigma_0^{\otimes n}P) \geq 1-\exp(-2\eps^2n/L_\sigma^2).
\end{equation}
We now show that imposing both cutoffs changes $\sigma_0^{\otimes n}$ only slightly. Set
\begin{equation*}
\widetilde{\sigma}_n \coloneqq PQ\sigma_0^{\otimes n}QP.
\end{equation*}
Since $Q$ commutes with $\sigma_0^{\otimes n}$, we have $0 \preceq Q\sigma_0^{\otimes n}Q \preceq \sigma_0^{\otimes n}$. From~\eqref{eq:spectral-truncation-q} and~\eqref{eq:spectral-truncation-p}, we see that
\begin{align*}
\norm{\sigma_0^{\otimes n}-Q\sigma_0^{\otimes n}Q}_1
&= \Tr\left(\sigma_0^{\otimes n}(I-Q)\right)
\leq \exp(-2\eps^2n/L_\sigma^2), \\
\Tr\left(Q\sigma_0^{\otimes n}Q(I-P)\right)
&\leq \Tr\left(\sigma_0^{\otimes n}(I-P)\right)
\leq \exp(-2\eps^2n/L_\sigma^2).
\end{align*}
Using these two inequalities, the triangle inequality, and applying the gentle measurement lemma (\cref{lem:gentle-measurement}) to the subnormalized matrix $Q\sigma_0^{\otimes n}Q$, we thus see that
\begin{align}
\norm{\sigma_0^{\otimes n}-\widetilde{\sigma}_n}_1
&\leq \norm{\sigma_0^{\otimes n}-Q\sigma_0^{\otimes n}Q}_1+\norm{Q\sigma_0^{\otimes n}Q-PQ\sigma_0^{\otimes n}QP}_1 \notag \\
&\leq \exp(-2\eps^2n/L_\sigma^2)+2\sqrt{\exp(-2\eps^2n/L_\sigma^2)} \notag \\
&\leq 4\exp\left(-\eps^2n/L_\sigma^2\right).
\label{eq:spectral-truncation-distance}
\end{align}
Next, we establish the matrix analog of the likelihood ratio bound. By the definitions of $Q$ and $P$, notice that $Q\sigma_0^{\otimes n}Q \preceq 2^{-n[S(\sigma_0)-\eps]}Q$ and $P \preceq 2^{n[S(\bar{\sigma})+\eps]}\bar{\sigma}^{\otimes n}$. By combining these two inequalities along with $PQP \preceq P$, we see that
\begin{align}
\widetilde{\sigma}_n 
&= PQ\sigma_0^{\otimes n}QP \notag \\
&\preceq 2^{-n[S(\sigma_0)-\eps]}PQP \notag \\
&\preceq 2^{-n[S(\sigma_0)-\eps]}P \notag \\
&\preceq 2^{n[S(\bar{\sigma})-S(\sigma_0)+2\eps]}\bar{\sigma}^{\otimes n} \notag \\
&= 2^{n[\chi(\sigma_0,\sigma_1)+2\eps]}\bar{\sigma}^{\otimes n}.
\label{eq:spectral-truncation-domination}
\end{align}
This is the quantum replacement for the classical assertion that the likelihood ratio is at most $2^{n(\chi+2\eps)}$ on the retained typical set.

Now, we have all the ingredients to conclude this theorem. Indeed, by the inequality $0\preceq T_n \preceq I$ and \cref{eq:symmetrized-decoder,eq:spectral-truncation-distance,eq:spectral-truncation-domination}, we see that
\begin{align*}
\Prb[\bchat = \bc]
&= \Tr(\sigma_0^{\otimes n}T_n) \\
&= \Tr(\widetilde{\sigma}_nT_n)+\Tr((\sigma_0^{\otimes n}-\widetilde{\sigma}_n)T_n) \\
&\leq 2^{n[\chi(\sigma_0,\sigma_1)+2\eps]}\Tr(\bar{\sigma}^{\otimes n}T_n)
+\norm{\sigma_0^{\otimes n}-\widetilde{\sigma}_n}_1 \\
&= 2^{-n[R(\cC)-\chi(\sigma_0,\sigma_1)-2\eps]}
+4\exp\left(-\eps^2n/L_\sigma^2\right).
\end{align*}
For the final assertion, consider any family of codes for which $R(\cC)-\chi(\sigma_0,\sigma_1)\geq\Delta$ for some fixed $\Delta>0$. Choosing $\eps=\Delta/4$ in the preceding bound gives the exponentially decaying estimate
\begin{equation*}
\Prb[\bchat = \bc]
\leq 2^{-n\Delta/2}+4\exp\left(-\Delta^2n/(16L_\sigma^2)\right).
\qedhere
\end{equation*}
\end{proof}

\subsection{\texorpdfstring{\cref{lem:coarse-graining}}{Lemma \getrefnumber{lem:coarse-graining}}: PGM is closed under coarse-graining}
\label{appsubsec:proof-coarse-graining}

In this subsection, we prove~\cref{lem:coarse-graining}. The proof proceeds by exploiting the linearity of PGMs.

\begin{proof}
First, observe that
\begin{equation*}
\bar{\omega}
\coloneqq \sum_{y \in \cY}{\nu_y \omega_y}
= \sum_{y \in \cY}{\Prb[f(\bx) = y]\Eb[\sigma_{\bx} \mid f(\bx) = y]}
= \Eb[\sigma_{\bx}]
= \bar{\sigma}.
\end{equation*}
Let $\{M_{\hat{x}}\}_{\hat{x} \in \cX}$ and $\{N_{\hat{y}}\}_{\hat{y} \in \cY}$ be the PGMs of $\cE$ and $\cF$, respectively. The PGM formula (\cref{def:pgm}) then gives
\begin{equation*}
N_{\hat{y}} = \nu_{\hat{y}}\bar{\omega}^{-1/2}\omega_{\hat{y}}\bar{\omega}^{-1/2} = \bar{\sigma}^{-1/2}
\left(\sum_{\substack{\hat{x} \in \cX : \\ f(\hat{x}) = \hat{y}}}{\pi_{\hat{x}}\sigma_{\hat{x}}}\right)
\bar{\sigma}^{-1/2} = \sum_{\substack{\hat{x} \in \cX : \\ f(\hat{x}) = \hat{y}}}{M_{\hat{x}}}.
\end{equation*}
Thus, the PGM $\{N_{\hat{y}}\}_{\hat{y} \in \cY}$ is the classical post-processing of $\{M_{\hat{x}}\}_{\hat{x} \in \cX}$ under $f$. In particular, an outcome $\hat{x}$ of the first PGM becomes $f(\hat{x})$ under the second.

For every $y,\hat{y} \in \cY$, the PGM error formula (\cref{def:pgm-error}) gives us
\begin{equation*}
\Prb[f(\bx) = y,\ f(\bxhat) = \hat{y}]
= \sum_{\substack{x \in \cX \\ f(x) = y}}{\sum_{\substack{\hat{x} \in \cX \\ f(\hat{x})
= \hat{y}}}{\pi_x\Tr(\sigma_xM_{\hat{x}})}}
= \sum_{\substack{x \in \cX \\ f(x) = y}}{\pi_x\Tr(\sigma_xN_{\hat{y}})}
= \nu_y\Tr(\omega_yN_{\hat{y}}).
\end{equation*}
By the PGM error formula (\cref{def:pgm-error}) once again and the definition of $\omega_y$, we conclude that
\begin{equation*}
\Prb[f(\bxhat) \neq f(\bx)]
= \sum_{\substack{y,\hat{y} \in \cY : \\ y \neq \hat{y}}}{\Prb[f(\bx) = y,\ f(\bxhat) = \hat{y}]}
= \sum_{\substack{y,\hat{y} \in \cY : \\ y \neq \hat{y}}}{\nu_y\Tr(\omega_yN_{\hat{y}})}
= p_{\mathrm{e}}(\cF).
\qedhere
\end{equation*}
\end{proof}

\subsection{\texorpdfstring{\cref{lem:uniform-prior-worst}}{Lemma \getrefnumber{lem:uniform-prior-worst}}: Worst-case uniform prior}
\label{appsubsec:proof-uniform-prior}

\begin{proof}[Proof of \cref{lem:uniform-prior-worst}]
Consider the binary ensemble that asks to decode a uniformly random bit $\bb$ when given as input $(\sigma_\bb,\bb \oplus \mathbf{t})$ for $\mathbf{t} \sim \Ber(\alpha)$. More formally, consider the binary ensemble $\cE_\alpha' \coloneqq \{(\frac{1}{2}, \widehat{\sigma}_0), (\frac{1}{2}, \widehat{\sigma}_1)\}$ where
\begin{equation*}
\widehat{\sigma}_0
\coloneqq \sigma_0 \otimes (1-\alpha)\ketbra{0}{0}_T + \sigma_0 \otimes \alpha\ketbra{1}{1}_T,
\qquad
\widehat{\sigma}_1
\coloneqq \sigma_1 \otimes \alpha\ketbra{0}{0}_T + \sigma_1 \otimes (1-\alpha)\ketbra{1}{1}_T.
\end{equation*}
Here, $T$ is a two-dimensional classical system. We claim that
\begin{equation*}
p_{\mathrm{e}}(\widehat{\sigma}_0,\widehat{\sigma}_1) = p_{\mathrm{e},\alpha}(\sigma_0,\sigma_1).
\end{equation*}
Indeed, one can check that the PGM of $\cE_\alpha'$ performs the following procedure: it first reads $T$. Then, if $T = 0$, it performs the PGM for the ensemble $\{(1-\alpha,\sigma_0),(\alpha,\sigma_1)\}$. Otherwise, if $T = 1$, it performs the PGM for the ensemble $\{(\alpha,\sigma_0),(1-\alpha,\sigma_1)\}$. Since $p_{\mathrm{e},\alpha}(\sigma_0,\sigma_1) = p_{\mathrm{e},1-\alpha}(\sigma_0,\sigma_1)$ by output-symmetry, we thus see that
\begin{equation*}
p_{\mathrm{e}}(\widehat{\sigma}_0,\widehat{\sigma}_1)
= \frac{1}{2}p_{\mathrm{e},\alpha}(\sigma_0,\sigma_1) + \frac{1}{2}p_{\mathrm{e},\alpha}(\sigma_0,\sigma_1)
= p_{\mathrm{e},\alpha}(\sigma_0,\sigma_1).
\end{equation*}
Now, observe that the quantum channel that traces $T$ out sends $\widehat{\sigma}_b$ to $\sigma_b$ for $b \in \bits$. The PGM data-processing lemma (\cref{lem:pgm-data-processing}) therefore gives
\begin{equation*}
p_{\mathrm{e},\alpha}(\sigma_0,\sigma_1) = p_{\mathrm{e}}(\widehat{\sigma}_0,\widehat{\sigma}_1) \leq p_{\mathrm{e},1/2}(\sigma_0,\sigma_1).\qedhere
\end{equation*}
\end{proof}

\subsection{\texorpdfstring{\cref{prop:binary-pgm-error,prop:binary-pgm-preprocessing}}{Propositions \getrefnumber{prop:binary-pgm-error} and \getrefnumber{prop:binary-pgm-preprocessing}}: Binary PGM formulas}
\label{appsubsec:proof-binary-pgm-formulas}

In this subsection, we prove \cref{prop:binary-pgm-error,prop:binary-pgm-preprocessing}.

\begin{proof}[Proof of \cref{prop:binary-pgm-error}]
By~\cref{def:pgm}, the PGM matrices for $\cE_\alpha$ are $M_0 = (1-\alpha)\bar{\sigma}_\alpha^{-1/2}\sigma_0\bar{\sigma}_\alpha^{-1/2}$ and $M_1 = \alpha\bar{\sigma}_\alpha^{-1/2}\sigma_1\bar{\sigma}_\alpha^{-1/2}$.
From~\cref{def:pgm-error} and the cyclicity of the trace, we thus find that
\begin{align*}
p_{\mathrm{e}}(\cE_\alpha)
&= (1-\alpha)\Tr(\sigma_0M_1)+\alpha\Tr(\sigma_1M_0) \\
&= \alpha(1-\alpha)\Tr\left(\sigma_0\bar{\sigma}_\alpha^{-1/2}\sigma_1\bar{\sigma}_\alpha^{-1/2}\right)
+\alpha(1-\alpha)\Tr\left(\sigma_1\bar{\sigma}_\alpha^{-1/2}\sigma_0\bar{\sigma}_\alpha^{-1/2}\right) \\
&= 2\alpha(1-\alpha)\Tr\left(\sigma_0\bar{\sigma}_\alpha^{-1/2}\sigma_1\bar{\sigma}_\alpha^{-1/2}\right).
\end{align*}
This proves \eqref{eq:binary-pgm-error}.

As for \eqref{eq:non-uniform-pure-pgm-error}, observe that the average $\bar{\sigma}_\alpha$ admits the Gram factorization $\bar{\sigma}_\alpha = VV^\dagger$, where $V$ is the matrix consisting of the two columns $\sqrt{1-\alpha}\ket{\psi_0}$ and $\sqrt{\alpha}\ket{\psi_1}$. Now, define the matrix
\begin{equation*}
G
\coloneqq
V^\dagger V = 
\begin{pmatrix}
1-\alpha & \sqrt{\alpha(1-\alpha)}\braket{\psi_0|\psi_1} \\
\sqrt{\alpha(1-\alpha)}\braket{\psi_1|\psi_0} & \alpha
\end{pmatrix}.
\end{equation*}
Observe that $\sqrt{G} = V^\dagger(VV^\dagger)^{-1/2}V = V^\dagger \bar{\sigma}_\alpha^{-1/2} V$. Combining this with the cyclicity of the trace and \eqref{eq:binary-pgm-error}, notice that
\begin{align}
p_{\mathrm{e}}(\cE_\alpha) &= 2\alpha(1-\alpha)\Tr\left(\sigma_0\bar{\sigma}_\alpha^{-1/2}\sigma_1\bar{\sigma}_\alpha^{-1/2}\right) \notag \\
&= 2\alpha(1-\alpha)\Tr\left(\ketbra{\psi_0}{\psi_0}\bar{\sigma}_\alpha^{-1/2}\ketbra{\psi_1}{\psi_1}\bar{\sigma}_\alpha^{-1/2}\right) \notag \\
&= 2\left(\sqrt{\alpha(1-\alpha)}\bra{\psi_0}\bar{\sigma}_\alpha^{-1/2}\ket{\psi_1}\right)^2 \notag \\
&= 2\abs{\sqrt{G}_{01}}^2. \label{eq:prop-12-proof-1}
\end{align}
Now, since $G$ is a $2 \times 2$ matrix, the span of $I$ and $G$ includes all the projections onto the eigenspaces of $G$. Thus, we can express $\sqrt{G}$ as
\begin{equation*}
\sqrt{G} = aG + bI
\end{equation*}
for some $a,b \in \mathbb{R}$. If $\lambda_1,\lambda_2 \in \mathbb{R}$ denote the eigenvalues of $G$, then this leads to a simple linear system $a\lambda_i + b = \sqrt{\lambda_i}$ for $i = 1,2$. Solving for just $a$, we get
\begin{equation*}
a = \frac{\sqrt{\lambda_1} - \sqrt{\lambda_2}}{\lambda_1 - \lambda_2} = \frac{1}{\sqrt{\lambda_1} + \sqrt{\lambda_2}} = \frac{1}{\sqrt{\lambda_1 + \lambda_2 + 2\sqrt{\lambda_1\lambda_2}}} = \frac{1}{\sqrt{\Tr(G) + 2\sqrt{\det(G)}}}.
\end{equation*}
From the definition of $G$, we see that $\Tr(G) = 1$ and $\det(G) = \alpha(1-\alpha)(1-f)$. Thus
\begin{equation}
(\sqrt{G})_{01}
= a \cdot (G)_{01} + b \cdot 0
= \frac{1}{\sqrt{\Tr(G) + 2\sqrt{\det(G)}}} \cdot \sqrt{\alpha(1-\alpha)}\braket{\psi_0|\psi_1}
= \frac{\sqrt{\alpha(1-\alpha)}\braket{\psi_0|\psi_1}}
{\sqrt{1+2\sqrt{\alpha(1-\alpha)(1-f)}}}. \label{eq:prop-12-proof-2}
\end{equation}
Therefore, combining~\eqref{eq:prop-12-proof-1} and~\eqref{eq:prop-12-proof-2}, we conclude that
\begin{equation*}
p_{\mathrm{e}}(\cE_\alpha)
= \frac{2\alpha(1-\alpha)f}{1+2\sqrt{\alpha(1-\alpha)(1-f)}}.
\qedhere
\end{equation*}
\end{proof}

\begin{proof}[Proof of \cref{prop:binary-pgm-preprocessing}]
Write $\bar{\sigma}_\theta \coloneqq (1-\theta)\sigma_0 + \theta\sigma_1$. Let $N_y \coloneqq \Prb[\by = y]\bar{\sigma}_\theta^{-1/2}\sigma_y\bar{\sigma}_\theta^{-1/2}$ be the PGM observables for the ensemble indexed by $\by$, and let $\{M_0,M_1\}$ be the PGM observables for the composed ensemble indexed by $\bb$. Since PGM is the Petz recovery map for cq channels (cf.~\cref{appsubsec:petz-recovery-properties}), by Petz functoriality (\cref{fact:petz-properties}\ref{item:petz-functoriality}), it follows that $\{M_0,M_1\}$ is a classical post-processing of $\{N_0,N_1\}$. More concretely,
\begin{align*}
M_b
&= \Prb[\bb = b]\bar{\sigma}_\theta^{-1/2}\tau_b\bar{\sigma}_\theta^{-1/2} \\
&= \sum_{y \in \bits}{\Prb[\bb = b,\by = y]\bar{\sigma}_\theta^{-1/2}\sigma_y\bar{\sigma}_\theta^{-1/2}} \\
&= \sum_{y \in \bits}{\Prb[\bb = b \mid \by = y]N_y}.
\end{align*}
Thus, the collective PGM first produces $\byhat$ using $\{N_0,N_1\}$ and then draws $\bbhat$ from the posterior law of $\bb$ given $\by = \byhat$, as asserted.

Now, the joint law of the intermediate input and PGM output is symmetric. Indeed, cyclicity of trace gives, for $y,\hat{y} \in \bits$,
\begin{align*}
\Prb[\by = y,\byhat = \hat{y}]
= \Prb[\by = y]\Prb[\by = \hat{y}]\Tr\left(\sigma_y\bar{\sigma}_\theta^{-1/2}\sigma_{\hat{y}}\bar{\sigma}_\theta^{-1/2}\right)
=\Prb[\by = \hat{y},\byhat = y].
\end{align*}
Put $\eps \coloneqq p_{\mathrm{e},\theta}(\sigma_0,\sigma_1)$. The two marginals of $(\by,\byhat)$ are therefore both $\Ber(\theta)$, and the two off-diagonal probabilities are each $\frac{\eps}{2}$. Thus, its $2 \times 2$ joint-distribution matrix, with rows indexed by $\by$ and columns by $\byhat$, is
\begin{equation*}
J_{\by,\byhat}
= \begin{pmatrix}
1-\theta-\frac{\eps}{2} & \frac{\eps}{2} \\
\frac{\eps}{2} & \theta-\frac{\eps}{2}
\end{pmatrix}.
\end{equation*}
Let $R$ be the $2 \times 2$ probability matrix of the posterior sampler of $\bb$ given $\by$, i.e., $R_{b,y} \coloneqq \Prb[\bb = b \mid \by = y]$. By Bayes' rule, we have that
\begin{equation*}
R
= \begin{pmatrix}
1-\frac{\alpha\eta}{1-\theta} & 1-\frac{\alpha(1-\eta)}{\theta} \\
\frac{\alpha\eta}{1-\theta} & \frac{\alpha(1-\eta)}{\theta}
\end{pmatrix},
\qquad
\det(R) = \frac{\alpha(1-\alpha)(1-2\eta)}{\theta(1-\theta)}.
\end{equation*}
Conditioned on $(\by,\byhat)$, the original bit $\bb$ and the fresh posterior draw $\bbhat$ are independent, so their joint-distribution matrix is
\begin{equation*}
J_{\bb,\bbhat} = R J_{\by,\byhat}R^{\mathsf T}.
\end{equation*}
For any pair of binary random variables $(\mathbf{x},\mathbf{z})$, covariance is precisely the determinant of their $2 \times 2$ joint-distribution matrix. Indeed, if $J_{\mathbf{x},\mathbf{z}} = (p_{xz})_{x,z \in \bits}$, then
\begin{align*}
\operatorname{Cov}(\mathbf{x},\mathbf{z})
&= p_{11} - (p_{11} + p_{10})(p_{11} + p_{01}) \\
&= p_{11} - p_{11}^2 - p_{11}p_{01} - p_{11}p_{10} - p_{10}p_{01} \\
&= p_{11}p_{00}-p_{10}p_{01} \\
&= \det(J_{\mathbf{x},\mathbf{z}}).
\end{align*}
Consequently,
\begin{equation*}
\operatorname{Cov}(\bb,\bbhat) = \det(R)^2\det(J_{\by,\byhat}) = \left[\frac{\alpha(1-\alpha)(1-2\eta)}{\theta(1-\theta)}\right]^2
\left[\theta(1-\theta)-\frac{\eps}{2}\right].
\end{equation*}
Since $\byhat$ has the same marginal as $\by$, posterior resampling gives $\bbhat \sim \Ber(\alpha)$. Hence
\begin{align*}
p_{\mathrm{e},\alpha}(\tau_0,\tau_1)
&= \Eb\left[(\bb-\bbhat)^2\right] \\
&= 2\alpha(1-\alpha)-2\operatorname{Cov}(\bb,\bbhat) \\
&= 2\alpha(1-\alpha) - \left[
\frac{\alpha(1-\alpha)(1-2\eta)}{\theta(1-\theta)}
\right]^2
\left[2\theta(1-\theta)-p_{\mathrm{e},\theta}(\sigma_0,\sigma_1)\right],
\end{align*}
which is \eqref{eq:pgm-preprocessing-error}. Setting $\alpha = \theta = \frac{1}{2}$ in \eqref{eq:pgm-preprocessing-error} gives
\begin{equation*}
p_{\mathrm{e}}(\tau_0,\tau_1)
= \frac{1}{2}\left[1-(1-2\eta)^2(1-2p_{\mathrm{e}}(\sigma_0,\sigma_1))\right].
\qedhere
\end{equation*}
\end{proof}

\subsection{\texorpdfstring{\cref{prop:qubit-entropy}}{Proposition \getrefnumber{prop:qubit-entropy}}: Qubit entropy in Pauli coordinates}
\label{appsubsec:proof-qubit-entropy}

In this subsection, we prove \cref{prop:qubit-entropy}.

\begin{proof}[Proof of \cref{prop:qubit-entropy}]
Since $X^2 = Z^2 = I$ and $XZ = -ZX$, we see that
\begin{equation*}
(xX+zZ)^2 = (x^2+z^2)I.
\end{equation*}
Since the Hermitian matrix $xX+zZ$ is traceless, we therefore see that the eigenvalues of $xX+zZ$ are $\{-\sqrt{x^2+z^2},+\sqrt{x^2+z^2}\}$. Therefore, the two eigenvalues of $\frac{1}{2}(I+xX+zZ)$ are $\frac{1}{2}(1-\sqrt{x^2+z^2})$ and $\frac{1}{2}(1+\sqrt{x^2+z^2})$. Substituting them into the definition of von Neumann entropy thus gives
\begin{equation*}
S\left(\frac{1}{2}(I+xX+zZ)\right)
= h\left(\frac{1-\sqrt{x^2+z^2}}{2}\right)
= s(1-x^2-z^2).
\qedhere
\end{equation*}
\end{proof}

\subsection{\texorpdfstring{\cref{prop:masking-pgm-error,prop:masking-holevo}}{Propositions \getrefnumber{prop:masking-pgm-error} and \getrefnumber{prop:masking-holevo}}: Masking}
\label{appsubsec:proof-masking}

In this subsection, we prove \cref{prop:masking-pgm-error,prop:masking-holevo}.

\begin{proof}[Proof of \cref{prop:masking-pgm-error}]
First, observe that the average state is $\bar{\Sigma} \coloneqq \frac{1}{2}(\Sigma_0+\Sigma_1) = \frac{1}{2}\bar{\tau}_\alpha \oplus \frac{1}{2}\bar{\tau}_\alpha$. Let
\begin{equation*}
N_0 \coloneqq (1-\alpha)\bar{\tau}_\alpha^{-1/2}\tau_0\bar{\tau}_\alpha^{-1/2},
\qquad
N_1 \coloneqq \alpha\bar{\tau}_\alpha^{-1/2}\tau_1\bar{\tau}_\alpha^{-1/2},
\end{equation*}
be the PGM observables of the non-uniform binary ensemble $\{(1-\alpha,\tau_0),(\alpha,\tau_1)\}$. Since $\bar{\Sigma}^{-1/2} = \sqrt{2}(\bar{\tau}_\alpha^{-1/2} \oplus \bar{\tau}_\alpha^{-1/2})$ on its support, the PGM definition gives
\begin{equation*}
M_0 = N_0 \oplus N_1,
\qquad
M_1 = N_1 \oplus N_0.
\end{equation*}
Thus, reading the block first leaves the non-uniform PGM in either block, with the outcome labels reversed in the second. Its error is
\begin{align*}
p_{\mathrm{e}}(\Sigma_0,\Sigma_1)
&= \frac{1}{2}\Tr(\Sigma_0M_1) + \frac{1}{2}\Tr(\Sigma_1M_0) \\
&= \frac{1-\alpha}{2}\Tr(\tau_0N_1) + \frac{\alpha}{2}\Tr(\tau_1N_0) + \frac{\alpha}{2}\Tr(\tau_1N_0) + \frac{1-\alpha}{2}\Tr(\tau_0N_1) \\
&= (1-\alpha)\Tr(\tau_0N_1)+\alpha\Tr(\tau_1N_0) \\
&= \alpha(1-\alpha)\Tr(\tau_0 \bar{\tau}_\alpha^{-1/2} \tau_1 \bar{\tau}_\alpha^{-1/2})
+\alpha(1-\alpha)\Tr(\tau_1 \bar{\tau}_\alpha^{-1/2}\tau_0 \bar{\tau}_\alpha^{-1/2}) \\
&= 2\alpha(1-\alpha)\Tr(\tau_0 \bar{\tau}_\alpha^{-1/2}\tau_1 \bar{\tau}_\alpha^{-1/2}),
\end{align*}
where cyclicity of the trace gives the last equality. This is also the PGM error of $\{(1-\alpha,\tau_0),(\alpha,\tau_1)\}$, proving \eqref{eq:masking-pgm-error}.
\end{proof}

\begin{proof}[Proof of \cref{prop:masking-holevo}]
If the eigenvalues of $\tau_0$ and $\tau_1$ are $\{\lambda_i\}_i$ and $\{\mu_j\}_j$, respectively, then the nonzero eigenvalues of $\Sigma_0$ are $\{(1-\alpha)\lambda_i\}_i$ and $\{\alpha\mu_j\}_j$. Therefore
\begin{equation*}
S(\Sigma_0) = S(\Sigma_1)
= h(\alpha)+(1-\alpha)S(\tau_0)+\alpha S(\tau_1).
\end{equation*}
Likewise, the two copies of $\bar{\tau}_\alpha$ in $\bar{\Sigma} = \frac{1}{2}\bar{\tau}_\alpha \oplus \frac{1}{2}\bar{\tau}_\alpha$ give $S(\bar{\Sigma}) = 1+S(\bar{\tau}_\alpha)$. Substitution in the Holevo identity proves \eqref{eq:masking-holevo}.
\end{proof}

\end{document}